\documentclass[12pt, titlepage]{article}

\usepackage{color}
\usepackage{float}
\usepackage[bf,small]{caption2}
\usepackage{comment}
\usepackage{amsmath}
\usepackage{bigints}
\usepackage{rotating, booktabs}
\usepackage{amsopn}
\usepackage{amsfonts}
\usepackage{amsthm}
\usepackage{amssymb}
\usepackage{amsbsy}
\usepackage{multirow}
\usepackage{titling}
\usepackage{natbib}
\usepackage{tocbibind}
\usepackage[a4paper,colorlinks,breaklinks,bookmarksopen,bookmarksnumbered]{hyperref}
\usepackage{epsfig,psfrag}
\usepackage{graphicx}
\usepackage{amsmath}
\usepackage{epstopdf}
\usepackage{tabularx}
\usepackage{subfigure}
\usepackage[mathscr]{euscript}
\usepackage[margin=1in]{geometry}

\usepackage{amsfonts} 
\usepackage{calc}
\usepackage{titlesec}

\newcommand{\ueq}[1][]{%
  \if\relax\detokenize{#1}\relax
    \sbox0{$\underbrace{=}_{}$}%
    \mathrel{\mathmakebox[\wd0]{=}}
  \else
    \mathrel{\underbrace{=}_{\mathclap{#1}}}
  \fi}

\newcommand {\ctn}{\cite}

\usepackage{url}
\renewcommand{\t}{\ensuremath{\theta}}

\newcommand{\bd}{\boldsymbol{d}}

\newcommand{\boeta}{\boldsymbol{\eta}}

\newcommand{\btheta}{\boldsymbol{\theta}}

\newcommand{\bbeta}{\boldsymbol{\beta}}
\newcommand{\bdelta}{\boldsymbol{\delta}}

\newcommand{\bxi}{\boldsymbol{\xi}}

\newcommand{\bgamma}{\boldsymbol{\gamma}}
\newcommand{\bSigma}{\boldsymbol{\Sigma}}

\newcommand{\bpsi}{\boldsymbol{\psi}}
\newcommand{\bpi}{\boldsymbol{\pi}}

\newcommand{\bmu}{\boldsymbol{\mu}}
\newcommand{\bzeta}{\boldsymbol{\zeta}}

\newcommand{\bV}{\boldsymbol{V}}
\newcommand{\bG}{\boldsymbol{G}}

\newcommand{\bA}{\boldsymbol{A}}
\newcommand{\bS}{\boldsymbol{S}}
\newcommand{\bt}{\boldsymbol{t}}

\newcommand{\bu}{\boldsymbol{u}}
\newcommand{\bv}{\boldsymbol{v}}
\newcommand{\bs}{\boldsymbol{s}}

\newcommand{\bx}{\boldsymbol{x}}
\newcommand{\bX}{\boldsymbol{X}}

\newcommand{\bY}{\boldsymbol{Y}}

\newcommand{\bz}{\boldsymbol{z}}

\newcommand{\bw}{\boldsymbol{w}}

\newtheorem{theorem}{Theorem}

\newtheorem{corollary}[theorem]{Corollary}

\newtheorem{definition}[theorem]{Definition}

\newtheorem{lemma}[theorem]{Lemma}

\newtheorem{remark}[theorem]{Remark}

\newcommand{\statesp}{\ensuremath{\mathcal X}}
\newcommand{\Y}{\ensuremath{\mathcal Y}}
\newcommand{\D}{\ensuremath{\mathcal D}}

\newcommand{\e}{\ensuremath{\epsilon}}

\newcommand{\supr}[2]{{#1}^{(#2)}}

\newcommand{\be}{\pmb\e}
\newcommand{\bm}{\mathbf}
\renewcommand{\t}{\ensuremath{\theta}}

\newcommand{\topline}{\hrule height 1pt width \textwidth \vspace*{2pt}}
\newcommand{\botline}{\vspace*{2pt}\hrule height 1pt width \textwidth \vspace*{4pt}}
\newtheorem{algo}{Algorithm} 

\numberwithin{equation}{section}
\numberwithin{algo}{section}
\numberwithin{table}{section}
\numberwithin{figure}{section}

\newcommand{\bi}[1]{\mbox{\boldmath{$ #1 $}}}

\begin{document}

\title{\vspace{-0.8in}
\textbf{Nonstationary, Nonparametric, Nonseparable Bayesian Spatio-Temporal Modeling Using Kernel Convolution of Order Based Dependent Dirichlet Process}}
\author{Moumita Das and Sourabh Bhattacharya\thanks{
Moumita Das is a postdoctoral research fellow at Basque Center for Applied Mathematics, Spain and
Sourabh Bhattacharya is an Associate Professor in Interdisciplinary Statistical Research Unit, Indian Statistical
Institute, 203, B. T. Road, Kolkata 700108.
Corresponding e-mail: sourabh@isical.ac.in.}}
\date{\vspace{-0.5in}}
\maketitle%
	
\begin{abstract}

Spatio-temporal processes are important modeling tools for varieties of problems in
environmental science, biological science, geographical science, etc. 
It is generally assumed that the underlying model is parametric, typically a Gaussian process,
and that the covariance function is stationary and separable.
That this structure does not need to be always realistic have been perceived by many researchers
and attempts have been made to construct nonparametric processes consisting of
neither stationary nor separable covariance functions. 
But, as we elucidate, some desirable and important spatio-temporal properties are
not guaranteed by the existing approaches, thus calling for further innovative ideas. 

In this article, using kernel convolution of order based dependent Dirichlet process (\ctn{Griffin06}) we construct
a nonstationary, nonseparable, nonparametric space-time process, which, as we show,
satisfies desirable properties, and includes the stationary, separable, parametric
processes as special cases. We also investigate the smoothness properties
of our proposed model. 

Since our model entails an infinite random series, for Bayesian model fitting purpose
we must either truncate the series or more appropriately consider a random number of summands,
which renders the model dimension a random variable. We attack the variable dimensionality problem
using Transdimensional Transformation based Markov Chain Monte Carlo  
introduced by \ctn{Das14}, which can update all the variables and also change
dimensions in a single block using essentially a single random variable drawn from some arbitrary density defined on a relevant support.
For the sake of completeness we also address the problem of truncating the infinite series by 
providing a uniform bound on the error incurred by truncating 
the infinite series.

	We illustrate the effectiveness of our model and methodologies on a simulated data set and demonstrate that our approach significantly 
outperforms that of \ctn{Fuentes03} which is based on principles somewhat similar to ours. We also fit two real, spatial and spatio-temporal datasets with our 
approach and obtain quite encouraging results in both the cases.
\\[2mm]

\end{abstract}
	
\tableofcontents
\pagebreak

\section{Introduction}

Recent years have witnessed considerable amount of research on spatial and spatio-temporal modeling.
The major inferential objectives of spatio-temporal modeling are to predict a plausible value at some point in space and time, forecasting the 
future value at some location, and to make inference about the parameters of the spatio temporal processes.
 A model must take account of spatio-temporal dependence structure of the given process. 
It is common practice to assume that the underlying spatial or spatio-temporal
process is stationary and isotropic Gaussian process, as  it facilitates prediction. 
In particular, the geostatistical method of kriging assumes a Gaussian process structure for the unknown spatial or spatio-temporal field and focuses on calculating the 
optimal linear predictor of the field. When performing kriging, researchers  generally assume a stationary, often 
isotropic, covariance function. The covariance of responses at any two locations is assumed to be a function of the
separation vector or 
of the distance between locations, but not a function of the actual locations.
Researchers often estimate the parameters of an isotropic covariance function from the semivariogram, the estimation of which is based 
on the squared differences between the responses as a function of the distance between locations. The standard kriging approach allows
one to flexibly estimate a smooth spatial field, with no pre-specified parametric stochastic model for the data. However, these approaches have several 
drawbacks. The most important is that the
true covariance structure may not be stationary. This is because there may be local influences affecting the correlation structure
of the random process. For instance, orographic effects influence the atmospheric transport of pollutants, and result in a correlation
structure that depends on different spatial locations (\ctn{Guttorp94}). If one is modelling an environmental variable across the United States, 
the field is likely to be much more 
smooth in the topographically-challenged Great Plains than in the Rocky Mountains.
This is manifested as different covariance structures in those two regions.
Assuming a stationary covariance structure will result in oversmoothing the field in the mountains and undersmoothing the field in 
great plains (\ctn{Paciorek03}).


Realizing the limitations of stationary parametric processes (almost invariably Gaussian
processes) researchers have come up
with many novel ideas for constructing nonstationary and/or nonparametric
processes. 
The first significant work
in the framework of nonstationary parametric processes is by \ctn{Sampson92}, who proposed an approach
based on spatial deformation. This work is followed up by \ctn{Damian01}
and \ctn{Schmidt03}, providing the corresponding Bayesian generalizations. 
Nonstationarity has been induced in parametric space-time models by 
\ctn{Haas95} by proposing a moving window regression residual kriging.
A similar approach has been proposed by \ctn{Nott02}.
\ctn{Higdon98} (see also \ctn{Higdon99}, \ctn{Higdon02}) proposed a kernel convolution approach for 
inducing nonstationarity in Gaussian processes.
Similar approaches are also proposed by \ctn{Fuentes01} and \ctn{Fuentes02}.
Approaches that attempt to model the underlying process as nonparametric, in addition
to modeling the covariance structure as nonstationary are more recent in comparison, 
the approach of \ctn{Gelfand05} based on Dirichlet processes (see, for example, 
\ctn{Ferguson73}, \ctn{Ferguson74}) being the first in this regard; see \ctn{Duan07} for
a generalization. \ctn{Duan09} use stochastic differential equations to construct a nonstationary,
non-Gaussian process. We discuss these proposals in some detail in Section \ref{sec:comparison}. 

\ctn{Fuentes03} proposed a nonparametric nonstationary model based on kernel processes mixing. 
In their study they showed that their proposed model outperformed all other models for several types of simulation designs (Stationary Gaussian, Nonstationary Gaussian, 
Stationary Non-Gaussian, Nonstationary Non-Gaussian). 
They illustrated their model with application to the monthly average values of ammonium and nitrate 
at 209 monitoring stations in the  US.
Their proposed nonstationary non-Gaussian model reduced the root mean square error (RMSE) by 24\% for ammonium and 18\% for nitrate when compared to the 
nonstationary Gaussian approach. RMSE also reduced compared to stationary Gaussian and stationary non-Gaussian approaches, although the gain is more moderate in these cases.

\ctn{Griffin06} (henceforth, GS) proposed the novel order-based dependent Dirichlet processes (ODDP). They introduced a framework for nonparametric modeling with dependence on continuous covariates. Dependence is induced 
through relevant weights utilizing similarities in the covariate information.
Each weight is a transformation of independently and identically distributed ($iid$) random variables. 
GS derived an ordering $\pi$ of these random variables at each covariate value such that distributions for similar covariate values are associated with similar orderings and
thus will be close. These orderings combined with Poisson point process give a simple analytical expression for the correlation function of the distributions, which ensures that if two points are similar in the covariate space they will get higher correlation compared to the points that are not. Furthermore when the distance between two points is large enough in the covariate space, the correlation approaches zero.  
In spatial/spatio-temporal context, it translates into the fact that when two observations are widely separated in space/space-time, the model based correlations 
tend to zero. But the ODDP process suffers from the limitation of being stationary. 

Preserving all the desirable properties of the correlation function of ODDP, we attempt to incorporate further flexibility in our spatial/temporal/spatio-temporal 
model in terms of nonstationarity and nonseparability through our proposed kernel convolution based methodology. Specifically, we propose a new class of 
spatial/temporal/spatio-temporal models that is nonparametric, nonstationary, nonseparable, and such that the correlation tends to zero if either of spatial and temporal
distance tends to infinity. All these properties are desirable in real data scenarios, and hence any effective, realistic model must satisfy these properties. 
Unfortunately, such a wholesome model does not seem to exist in the current literature, as we point out in our review.  
Hence, this paper is an attempt to create one class of such realistic stochastic processes. We illustrate our ideas not only wth simulation study, 
but also with a real spatial data on ozone and a real spatio-temporal data on
particulate matters. That these data sets are both strictly and weakly nonstationary, are inferred in a separate paper by \ctn{Roy20} using novel Bayesian methodologies.
In this article, we further show that the empirical correlations for the spatio-temporal data tend to zero as the spatio-temporal lags increase. 
A similar property is also expected of the spatial ozone data, but the small size of the data did not permit such rigorous analysis.
Moreover, these data sets are far from Gaussianity, as simple quantile-quantile plots indicate.
As we argued, these properties are expected in reality, and the general class of nonparametric spatio-temporal models that we propose, 
provides adequate fits to both these data sets.  
Moreover, comparison of our analyses with one of the most competent existing models, shows that our model is possibly indeed worth pursuing.  


The rest of our paper is structured as follows.
In Section \ref{sec:comparison} we provide a brief overview of the existing approaches to construction
of nonstationary, nonseparable space-time processes in both parametric and nonparametric frameworks, 
arguing that not all desirable
properties are necessarily accounted for in these approaches. Such issues necessitate
development of new approaches to construction of nonstationary, nonparametric, nonseparable space-time
models. 
In Section \ref{sec:oddp_convolution} we introduce our proposed space-time model based on
kernel convolution of ODDP and show that it satisfies the properties that are
not guaranteed by the existing models. We investigate continuity and smoothness properties
of our model in Section \ref{sec:smoothness}. Since our proposed model involves a random infinite
series, for model fitting one needs to either truncate the series or assume a random number of
summands and adopt variable dimensional Markov Chain Monte Carlo (MCMC) approaches. Although
we adopt the latter framework for our applications, and implement the recently developed Transdimensional Transformation
based Markov Chain Monte Carlo (TTMCMC) (\ctn{Das14}) for simulating from our variable dimensional model,
for the sake of completeness we also investigate the truncation approach. Indeed, in Section \ref{sec:truncation}
we consider the difference between the prior predictive models with and without truncation
of the random infinite series, providing a bound that depends upon the truncation parameter. Thus,
the truncation parameter can be chosen so that the bound falls below any desired level.
In Section \ref{sec:kernel_prior} we discuss the choice of suitable kernels, prior distributions
and choice of the spatio-temporal domain that is relevant for computational purpose. 
We describe the joint posterior distribution associated with our model, and provide a brief
discussion of TTMCMC in Section \ref{sec:joint_posterior}.
We detail a simulation study illustrating the performance of our model and comparison with \ctn{Fuentes03} in Section \ref{sec:simulation_study}.
Indeed, the model of \ctn{Fuentes03}, in spite of being very different from our ideas, comes closest to our model conceptually, among the existing models.
In Section \ref{sec:realdata_analyses} we consider application of our ideas to two real datasets: a spatial ozone dataset, and a
spatio-temporal dataset on particulate matters.
Finally, we summarize our contributions and provide concluding remarks in Section \ref{sec:conclusion}.

Proofs of our results and requisite details of TTMCMC, particularly in the context of our spatio-temporal
model, and details regarding generation of the data for the simulation experiment, are provided in the supplement \ctn{Das14supp}, 
whose sections and algorithms have 
the prefix ``S-" when referred to in this paper. 

\section{Overview of other available nonstationary approaches}
\label{sec:comparison}

\subsection{Parametric approaches}
\label{subsec:parametric}

The deformation approaches of \ctn{Sampson92}, \ctn{Damian01}, and \ctn{Schmidt03} are based on Gaussian processes. 
In these approaches replications of the data are necessary, which the authors relate to temporal independence
of the data. This also means that space-time data can not be modeled using these approaches, unless all the temporal dependence can be captured thrrough
a trend term in the mean structure.
Moreover, in the deformation-based approaches model based theoretical correlations between 
random observations separated by large enough 
distances need not necessarily tend to zero. 
Letting $Y(\bs,t)$ denote the response at spatial location $\bs$ and time $t$, \ctn{Sampson92} deal with the variogram of the following form:
\begin{equation}
\mbox{Var}(Y(\bs_1,t)-Y(\bs_2,t))=f(\|\bd(\bs_1,t)-\bd(\bs_2,t)\|),
\label{eq:deform1}
\end{equation}
for any $\bs_1,\bs_2,t$, where $f$ is an appropriate monotone function and $\bd$ is
a one-to-one nonlinear mapping. The technique of \ctn{Sampson92} involves appropriately approximating $f$ by $\hat f$ 
using the multidimensional scaling method, and obtaining
a configuration of points $\{\bu_1,\ldots,\bu_n\}$ in a ``deformed" space where the process is assumed
isotropic. Then, using thin-plate splines, a nonlinear approximation of $\bd$, which we denote by $\hat\bd$, is determined such 
that $\hat\bd (\bs_i)\approx\bu_i$, for $i=1,\ldots,n$. Bayesian versions of the key idea have been described
in \ctn{Damian01}, who use random thin-plate splines and \ctn{Schmidt03}, who use Gaussian process to implement
the nonlinear transformation $\bd$. Rather than estimate $f$ nonparametrically, both specify a parametric
functional form from a valid class of such monotone functions. 

As is clear, since large differences $\|\bs_1-\bs_2\|$ does not imply that $\|\bd(\bs_1)-\bd (\bs_2)\|$
is also large, the model based correlations between two observations widely separated need 
not necessarily tend to zero, in either of the
aforementioned deformation-based approaches.

The kernel convolution approaches of \ctn{Higdon99}, \ctn{Higdon02},
and \ctn{Fuentes01} overcome some of the difficulties of the deformation approach. 
In these approaches data replication is not necessary, and for appropriate choices of the kernel,
stationarity, nonstationarity, separability, and nonseparability can be achieved with respect to spatio-temporal data.
In the approach of \ctn{Higdon99}, \ctn{Higdon02},
\begin{equation}
Y(\bx)=\int K(\bx,\bu)Z(\bu) d\bu,
\label{eq:higdon1}
\end{equation}
where $K$ is a kernel function and $Z(\cdot)$ is a white noise process. 
Then the covariance between $Y(\bx_1)$ and $Y(\bx_2)$ is given by
\begin{equation}
C(\bx_1,\bx_2)=\int K(\bx_1,\bu)K(\bx_2,\bu)d\bu.
\label{eq:higdon2}
\end{equation}
In general, this does not depend upon $\bx_1$ and $\bx_2$ only through $\bx_1-\bx_2$, thus
achieving nonstationarity. However, it is clear from the covariance structure (\ref{eq:higdon2}) that
$C(\bx_1,\bx_2)$ does not generally tend to zero as $d=\|\bx_1-\bx_2\|\rightarrow\infty$.
But for separable space-time processes (see, for example, \ctn{Cressie11} for various illustrations) related to representation (\ref{eq:higdon1}) this property holds
under the additional assumption of isotropy with respect to either space or time. We elaborate this below.

Although representation (\ref{eq:higdon1}) can not achieve separability with respect to
space and time, a modified representation
of the following form does:
\begin{equation}
Y(\bs,t)=\int K_1(\bs,\bu)K_2(t,\bv)Z_1(\bu)Z_2(\bv) d\bu d\bv.
\label{eq:higdon3}
\end{equation}
In (\ref{eq:higdon3}), $K_1, K_2$ are two kernel functions, and $Z_1(\bx), Z_2(\bx)$ are independent
white noise processes.
Now the covariance is given by
\begin{align}
C((\bs_1,t_1),(\bs_2,t_2))&=\int K_1(\bs_1,\bu)K_1(\bs_2,\bu)K_2(t_1,\bv)K_2(t_2,\bv)d\bu d\bv\notag\\
&=C_1(\bs_1,\bs_2)\times C_2(t_1,t_2),
\label{eq:higdon4}
\end{align}
where 
\begin{align}
C_1(\bs_1,\bs_2)&=\int K_1(\bs_1,\bu)K_1(\bs_2,\bu)d\bu,\label{eq:higdon5}\\
C_2(t_1,t_2)&=\int K_2(t_1,\bv)K_2(t_2,\bv)d\bv,\label{eq:higdon6}
\end{align}
exhibiting separability. Further assuming that either of $C_1$ or $C_2$ is isotropic, it follows that
if either of $d_1=\|\bs_1-\bs_2\|$ or $d_2=|t_1-t_2|$ tends to infinity, the covariance given by
(\ref{eq:higdon4}) tends to zero even though either of $C_1$ or $C_2$ is nonstationary. But if both
$C_1$ and $C_2$ are nonstationary, then this result need not hold. 

The approach of \ctn{Fuentes01} comes close towards solving the problem of zero covariance in the limit 
with large enough separation between observations, which we now explain. They model the underlying process as
\begin{equation}
Y(\bx)=\int K(\bx-\bu)Z_{\btheta (\bu)}d\bu,
\label{eq:fuentes1}
\end{equation}
where $Z_{\btheta}(\bx);\bx\in D$ is a family of independent, stationary Gaussian processes
indexed by $\btheta$, where the covariance of $Z_{\btheta (\bu)}$ is given by 
\begin{equation}
\mbox{Cov}\left(Z_{\btheta (\bu)}(\bx_1),Z_{\btheta (\bu)}(\bx_2)\right)=C_{\btheta (\bu)}(\bx_1-\bx_2).
\label{eq:fuentes2}
\end{equation} 
Then, the covariance between $Y(\bx_1)$ and $Y(\bx_2)$ is given by
\begin{equation}
C(\bx_1,\bx_2;\btheta)=\int K(\bx_1-\bu)K(\bx_2-\bu)C_{\btheta (\bu)}(\bx_1-\bx_2)d\bu.
\label{eq:fuentes3}
\end{equation}
For practical purposes, \ctn{Fuentes01} approximate $Y(\bx)$ with
\begin{equation}
\hat Y(\bx)=\frac{1}{M}\sum_{m=1}^MK(\bx-\bu_m)Z_{\btheta (\bu_m)}(\bx),
\label{eq:fuentes4}
\end{equation}
and $C(\bx_1,\bx_2;\btheta)$ by
\begin{equation}
\hat C(\bx_1,\bx_2;\btheta)=\frac{1}{M}\sum_{m=1}^M K(\bx_1-\bu_m)K(\bx_2-\bu_m)C_{\btheta (\bu_m)}(\bx_1-\bx_2),
\label{eq:fuentes5}
\end{equation}
where $\{\bu_1,\ldots,\bu_M\}$ can be thought of as a set of locations drawn independently from the domain $D$.
Assuming that the family of independent Gaussian processes $Z_{\btheta}(\bx);\bx\in D$ is also isotropic, 
it follows, using the fact that $M$ is finite, that $\hat C(\bx_1,\bx_2;\btheta)\rightarrow 0$ as $\|\bx_1-\bx_2\|\rightarrow\infty$
since $C_{\btheta (\bu_m)}(\bx_1-\bx_2)\rightarrow 0$ for each $m=1,\ldots,M$.
However, this of course does not guarantee that $\hat C(\bx_1,\bx_2;\btheta)\rightarrow 0$ as $M\rightarrow\infty$.
That is, this does not necesasarily imply that $C(\bx_1,\bx_2;\btheta)\rightarrow 0$.

A nonstationary process has been constructed by \ctn{Chang10}, by representing the underlying process
as a linear combination of basis functions and stationary Gaussian processes. This approach also does not
guarantee that the correlation tends to zero if
$\|\bx_1-\bx_2\|\rightarrow\infty$. For other available parametric approaches to nonstationarity
we refer to the references provided in \ctn{Chang10}.

\subsection{Nonparametric approaches}
\label{subsec:nonparametric}

\ctn{Gelfand05} seem to be the first to propose a nonstationary, noparametric Bayesian model
based on Dirichlet process mixing. They represent
the random field $\bY_D=\{Y(\bx);\bx\in D\}$ as $\sum_{\ell=1}^{\infty}w_{\ell}\delta_{\btheta_{\ell,D}}$,
where $\btheta_{\ell,D}=\{\theta_{\ell}(\bx);\bx\in D\}$ are realizations from a specified stationary Gaussian process,
which we denote as $\bG_0$,
$w_1=V_1$, $w_{\ell}=V_{\ell}\prod_{r=1}^{\ell-1}(1-V_r)$ for $\ell\geq 2$, where $V_r\stackrel{iid}{\sim}Beta(1,\alpha);~r=1,2,\ldots$.
Thus, a random process $\bG$ is induced on the space of processes 
of $\bY_D$ with $\bG_0$ being the ``central" process. \ctn{Gelfand05} assume the space-time data 
$\bY_t=(Y(\bs_1,t),\ldots,Y(\bs_n,t))'$ to be time-independent for $t=1,\ldots,T$, which is the same
assumption of data replication used in the deformation-based approaches. The temporal-independence
assumption allows \ctn{Gelfand05} to model the data as follows: for $t=1,\ldots,T$, $\bY_t\stackrel{iid}{\sim}\bG^{(n)}$ 
and $\bG^{(n)}\sim DP(\bG^{(n)}_0)$, where $\bG^{(n)}$
and $\bG^{(n)}_0$ denote the $n$-variate distributions corresponding to the processes $\bG$ and $\bG_0$.
The development leads to the following covariance structure: for any $\bs_1,\bs_2,t$,
\begin{equation}
\mbox{Cov}(Y(\bs_1,t),Y(\bs_2,t)\mid\bG)=\sum_{\ell=1}^{\infty}w_{\ell}\theta_{\ell}(\bs_1)\theta_{\ell}(\bs_2)
-\left\{\sum_{\ell=1}^{\infty}w_{\ell}\theta_{\ell}(\bs_1)\right\}
\left\{\sum_{\ell=1}^{\infty}w_{\ell}\theta_{\ell}(\bs_2)\right\},
\label{eq:nonpara1}
\end{equation}
which is nonstationary.
However, marginalized over $\bG$, the covariance between
$Y(\bs_1,t)$ and $Y(\bs_2,t)$ turns out to be stationary.
Since, in \ctn{Gelfand05}, the Bayesian inference of the data $\bY_1,\ldots,\bY_n$
proceeds by integrating out $\bG^{(n)}$, the
entire flavour of nonstationarity is lost.
Also, given $\bG$, (\ref{eq:nonpara1}) is nonstationary 
but does not necessarily converge to zero if $\|\bs_1-\bs_2\|\rightarrow\infty$.

\ctn{Duan07} attempt to generalize the model of \ctn{Gelfand05} by specifying $\bG$ as
\begin{equation}
Pr\{Y(\bx_1)\in A_1,\ldots,Y(\bx_n)\in A_n\}=\sum_{i_1=1}^{\infty}\cdots\sum_{i_n=1}^{\infty}
p_{i_1,\ldots,i_n}\delta_{\theta_{i_1}(\bx_1)}(A_1)\cdots \delta_{\theta_{i_n}(\bx_n)}(A_n),
\label{eq:nonpara2}
\end{equation}
where $\btheta_j$'s are $iid$ $\bG_0$ as in \ctn{Gelfand05}, and  
$\{p_{i_1,\ldots,i_n}\geq 0: \sum_{i_1=1}^{\infty}\cdots\sum_{i_n=1}^{\infty}p_{i_1,\ldots,i_n}=1\}$ 
determine the site-specific joint selection
probabilities, which also must satisfy simple constraints to ensure consistency. The resulting  
conditional covariance (conditional on $\bG$) and the marginal covariance are somewhat modified
versions of those of \ctn{Gelfand05}, but now even the marginal covariance is nonstationary.
By choosing $\bG_0$ to be an isotropic Gaussian process it can be ensured that the marginal
covariance tends to zero as two observations are widely separated, but the same can not be ensured
for the conditional covariance. Moreover, replications of the data is necesary even for this generalized
version of \ctn{Gelfand05}, and modeling temporal dependence is precluded as before.
A methodology very similar to that of \ctn{Duan07} is proposed in \ctn{Petrone09}.

Although the aforementioned approaches are temporally independent, \ctn{Kottas07} have considered a first order autoregressive setup to model temporal
dependence as a simple parametric temporal extension of the temporally independent model proposed in \ctn{Gelfand05}.

A nonstationary, nonseparable non-Gaussian spatiotemporal process has been constructed by \ctn{Duan09}
using discretized versions of stochastic differential equations, but again, the correlations between largely separated
observations do not necessarily tend to zero under their model. Also, stationarity or separability can not be 
derived as special cases of this approach. 

A flexible approach using kernel convolution of L\'{e}vy random measures has been detailed in \ctn{Wolpert11}, but even
this approach does not guarantee that correlations tend to zero for largely separated distances for arbitrarily 
chosen kernels.

An univariate and multivariate nonparametric spatial model
based on kernel process mixing has been proposed by \ctn{Fuentes03} (henceforth, FR). 
In this work, the idea of  stick-breaking prior of \ctn{Sethuraman94} was extended to a spatial set up. 
A different, unknown distribution was assigned to each location, with a series of 
space-dependent kernel functions that have a space-varying bandwidth parameter. Essentially, the Beta-distributed sequence $\left\{V_r:r=1,2,\ldots\right\}$ in 
the stick-breaking construction
of the traditional Dirichlet process are multiplied with a sequence of space-dependent kernels $\left\{K_r(\bs):r=1,2,\ldots\right\}$, and the $\bG_0$-distribured
sequence is replaced with an isotropic Gaussian process with nonstationary variance. The kernel functions attempt to impose a natural ranking for the 
different mixture components based on distances 
of locations to knots, which seems to be an alternative way to mimic the role of the orderings imposed in GS.
As the bandwidths of the kernels tend to zero uniformly, the covariance conditional on $\left\{V_r:r=1,2,\ldots\right\}$ tends to the isotropic covariance of the 
underlying Gaussian process. Marginally, the covariance structure, albeit nonstationary, need not yield zero covariance even if the distance between the locations
tend to infinity. Moreover, this idea has been considered only for spatial modeling. Although it is simple to extend the method to spatio-temporal situations,
enforcing separability is needed, does not seem to be as straightforward.

Compared to the vast literature on continuous nonstationary spatio temporal processes, there are very few methods available to model non-smooth covariance structures 
over the space or both space-time (\ctn{Guttorp2013}). Among them, \ctn{Kim2005} developed a method based on a Bayesian approach to Voronoi tesselation. 
Since our approach hinges upon the idea of GS, and smoothness properties of the ODDP depends on the order generating process, it is discontinuous in nature. 
We will discuss in details the smoothness properties of our model in Section \ref{sec:smoothness}.
Another possible source of nonstationarity is the local influence of some covariates on the spatial process of interest. Recently, there have been some proposals 
in the literature that account for covariate information in the covariance structure of spatial and spatio-temporal processes; see, for example, \ctn{Reich11}, \ctn{Schmidt11},
\ctn{Neto14}, \ctn{Ing14}, \ctn{Risser15}, \ctn{Gilani16}, \ctn{Risser19}. 
Since in our model we introduce dependence via the ODDP, where weights in the Sethuraman representation are dependent on the covariate information, 
we can efficiently incorporate the local influence of covariate information into our model. The covariate information can also be incorporated in the kernel 
that we convolve the ODDP with.

In the next section we introduce our idea based on kernel convolution of ODDP and 
show that it overcomes the issues faced by the traditional approaches to construction of
flexible space-time models.

\section{Kernel convolution of ODDP}
\label{sec:oddp_convolution}

Before introducing our proposal, it is necessary to first provide an overview of ODDP.

\subsection{Overview of ODDP}
\label{subsection:overview}

In order to induce spatial dependence between observations at different locations
GS modify the nonparametric stick-breaking construction of \ctn{Sethuraman94}
in the following way:
for each point $\bx\in D$, where $D$ is some specified domain, they define the distribution:
\begin{equation}
G_{\bx}\stackrel{\mathcal D}{=}\sum_{i=1}^{\infty}p_i(\bx)\delta_{\btheta_{\pi_i(\bx)}},
\label{eq:odpp1}
\end{equation}
where 
\begin{equation}
p_i(\bx)=V_{\pi_i(\bx)}\prod_{j<i}(1-V_{\pi_j(\bx)}).
\label{eq:odpp2}
\end{equation}
In (\ref{eq:odpp1}) and (\ref{eq:odpp2}), $\bi{\pi}(\bx)=(\pi_1(\bx),\pi_2(\bx),\ldots)$
denotes the ordering at $\bx$, where $\pi_i(\bx)\in\{1,2,\ldots\}$ and $\pi_i(\bx)=\pi_j(\bx)$
if and only if $i=j$. 
For $j=1,2,\ldots$, the parameters $\btheta_j\stackrel{iid}{\sim} G_0$, where $G_0$ is some
specified parametric centering distribution, and $V_j\stackrel{iid}{\sim}Beta(1,\alpha)$, where $\alpha>0$
is a specified parameter. The process associated with specification (\ref{eq:odpp1}) is the ODDP.
Clearly, if $\pi_i(\bx)=i$ for each $\bx$ and $i$, then the Dirichlet process (DP) results at all locations.

GS construct $\bi{\pi}(\bx)$ in a way such that 
it is associated with the realization of a point process. Specifically, they consider a stationary Poisson process
$\Phi$ and a sequence of sets $U(\bx)$ for $\bx\in D$, the latter determining the relevant region for the ordering purpose.
In the case of only spatial problems, if $\bx\in D\subset\mathbb R^d$, for $d\geq 1$, 
then GS suggest $U(\bx)=D$ for all $\bx\in D$
as a suitable construction of $U(\bx)$. For time series problems
they suggest $D=\mathbb R$ and $U(x)=(-\infty,x]$. When $\bx=(\bs',t)'$, that is, when $\bx$ consists
of both spatial and temporal co-ordinates, for our modeling purpose, we use $U(\bx)=D\times (-\infty,t]$.

Letting $\{\bz_1,\bz_2,\ldots\}$ denote a realization of the stationary Poisson point process, the ordering $\bi{\pi}(\bx)$
is chosen to satisfy
$\|\bx-\bz_{\pi_1(\bx)}\|<\|\bx-\bz_{\pi_2(\bx)}\|<\|\bx-\bz_{\pi_3(\bx)}\|<\cdots$,
where $\|\cdot\|$ is a distance measure and $\bz_{\pi(\bx)}\in\Phi\cap U(\bx)$.
Thus, although the set of probabilities $\{p_i(\bx);i=1,2,\ldots\}$ remains same
for all locations, they are randomly permuted.   
This random permutation, in turn, induces spatial dependence. Assuming a homogeneous Poisson point process
with intensity $\lambda$, ODDP is characterized by $G_0$, $\alpha$, and $\lambda$. We express
dependence of ODDP on these parameters by $\mbox{ODDP}(\alpha G_0,\lambda)$.

Assuming that data $\{y_1,\ldots,y_n\}$ are available at sites $\{\bx_1,\ldots,\bx_n\}$, 
GS embed the ODDP in a hierarchical Bayesian model:
\begin{align}
y_i &\sim f_{\btheta_i}(\cdot)\notag\\
\btheta_i&\sim G_{\bx_i}\notag\\
G_{\bx_i} &\sim \mbox{ODDP}(\alpha G_0,\lambda).\notag
\end{align}
Note that the same theory can be extended to space-time situations with
$\bx=(\bs',t)'$, where $\bs$ stands for the spatial location and $t$ stands for the time point. 

Next, we introduce our proposed idea of kernel convolution of ODDP.

\subsection{Kernel convolution of ODDP}
\label{subsec:kernel_convolution}

We consider the following model for the data $\bY=\{y_1,\ldots,y_n\}$ at locations/times  
$\{\bx_i=(\bs'_i,t_i)';~i=1,\ldots,n\}$:
\begin{equation}
y_i=f(\bx_i)+\epsilon_i,
\label{eq:model}
\end{equation}
where $\epsilon_i\stackrel{iid}{\sim}N(0,\sigma^2)$, for unknown $\sigma^2$.
We represent the spatio-temporal process $f(\bi{x})$ as a convolution of ODDP $G_{\bx}$ with a smoothing 
kernel $K(\bi{x},\cdot)$:
\begin{equation}
f(\bx)=\int K(\bx,\btheta)dG_{\bx}(\btheta)= \sum_{i=1}^{\infty}K(\bx,\btheta_{\pi_{i}(\bx)})p_{i}(\bx) \mbox{       }
\forall \bx\in D\subseteq\mathbb{R}^d,
\label{eq:kernel_convolution}
\end{equation}
$d~(\geq 1)$ being the dimension of $\bx$.
Thus, given $\bG_{\bx_i}$, 
\begin{equation}
	y_i\sim N\left(f(\bx_i),\sigma^2\right),
	\label{eq:kc1}
\end{equation}
the normal distribution with mean $f(\bx_i)$ of the form (\ref{eq:kernel_convolution}) and variance $\sigma^2$.
Thus, given $\bG_{\bx_i}$ and $\bG_{\bx_j}$, $y_i$ and $y_j$ are independent.

Since the ODDP model of GS can also be viewed as a convolution, it is important to clarify its differences with (\ref{eq:kernel_convolution}) and (\ref{eq:kc1}). 
Indeed, note that with respect to GS, the response data $y_i$ has the following distribution: 
\begin{equation}
	y_i\sim \int f_{\btheta}(\cdot)dG_{\bx_i}(\btheta)=\sum_{j=1}^{\infty}f_{\btheta_{\pi_{j}(\bx)}}(\cdot)p_{j}(\bx_i).
	\label{eq:gs_conv}
\end{equation}
Thus, under (\ref{eq:gs_conv}) (that is, under the model proposed by GS), for any choice of $f_{\btheta}$, $y_i$ arises from an infinite mixture 
with mixture density components $f_{\btheta_{\pi_{j}(\bx_i)}}(\cdot)$
and corresponding mixture probabilities $p_{j}(\bx_i)$. On the other hand, our model postulates a normal distribution for $y_i$ via (\ref{eq:kernel_convolution})
and (\ref{eq:kc1}), where the mean is the kernel convolution given by (\ref{eq:kernel_convolution}).
The convolutions given by (\ref{eq:kernel_convolution}) and (\ref{eq:gs_conv}) also have different interpretations. The latter is a density, whereas, the former
is any real-valued function. Note that unlike the case of (\ref{eq:model}), given $\bG_{\bx_i}$ and $\bG_{\bx_j}$, $y_i$ and $y_j$ are not independent if 
$f_{\btheta}(\cdot)=\delta_{\btheta}(\cdot)$, that is, when $y_i\sim \bG_{\bx_i}$. 
Further implications, with respect to nonstationarity and correlation structure tending to zero with widely separated distances, are discussed 
following Theorem \ref{theorem:corr}.

In spatio-temporal processes we have to specify the joint distribution for an uncountable number of random variables. 
But, in practice we observe the process at a finite number of locations only. To infer about the process, it is better to have finite moments, that 
ensures existence of the posterior distribution. It also facilitates the prediction of the process at an arbitrary unobserved 
location. 
The following theorem, the proof of which is presented in Section S-1 of the supplement, 
gives an expression of the expectation of $f(\bx)$.
\begin{theorem} Let $\int |K(\bi{x},\bi{\theta})|dG_{0}(\bi{\theta}) < \infty$. 
Then $\int |K(\bi{x},\bi{\theta})|dG_{\bi{x}}(\bi{\theta}) < \infty$ with probability one, 
and 
\[
 {{E}}(f(\bi{x}))={{E}} \int K(\bi{x},\bi{\theta})dG_{\bi{x}}(\bi{\theta})=\int K(\bi{x},\bi{\theta})d{{E}}G_{\bi{x}}(\bi{\theta})=\int K(\bi{x},\bi{\theta})dG_{0}(\bi{\theta})=E_{G_0}K(\bx,\btheta).
\]
\end{theorem}

Before deriving the covariance structure of $f(\cdot)$, we define the necessary notation following GS.
Let
\[
T(\bi{x_{1}},\bi{x_{2}})=\{k:\mbox{ there exists } i,j \mbox{ such that } \pi_{i}(\bi{x_{1}})=\pi_{j}(\bi{x_{2}})=k\}.
\]
$\mbox{For } k\in T(\bi{x_{1}},\bi{x_{2}})$, we further define
$
A_{lk}=\{\pi_{j}(\bi{x_{l}}): j<i \mbox{ where } \pi_i(\bi{x_{l}})=k\},
$
$S_{k}$ = $A_{1k}\cap A_{2k}$ and $S'_{k}$ = $A_{1k}\cup A_{2k}-S_{k}$.
Then, the following theorem, the proof of which is deferred to Section S-2 of the supplement, 
provides an expression for the covariance structure 
of $f(\cdot)$, which will be our reference point for arguments regarding nonstationarity and other desirable spatial properties
in comparison with the existing methods.
\begin{theorem} 
\label{theorem:cov}
If $\int |K(\bi{x},\bi{\theta})|dG_{0}(\bi{\theta}) < \infty$ and $\int|K(\bi{x_{1}},\bi{\theta})K(\bi{x_{2}},\bi{\theta})|dG_{0}(\bi{\theta}) < \infty$, then for a fixed ordering at $\bx_1$ and $\bx_2$,
\begin{align}
\mbox{Cov}(f(\bi{x_{1}}),f(\bi{x_{2}}))=&\mbox{Cov}_{G_0}(K(\bi{x_{1}},\bi{\theta}),K(\bi{x_{2}},\bi{\theta}))\notag\\
& \ \ \ \times \frac{2}{(\alpha+1)(\alpha+2)}\sum_{k\in T(\bi{x_{1}},\bi{x_{2}})}\left(\frac{\alpha}{\alpha+2}\right)^{\#S_{k}}
\left(\frac{\alpha}{\alpha+1}\right)^{\#S_{k}'}.\label{eq:cov_f}
\end{align}
where
\begin{equation}
\mbox{Cov}_{G_0}(K(\bi{x_{1}},\bi{\theta}),K(\bi{x_{2}},\bi{\theta}))=\int K(\bi{x_{1}},\bi{\theta})K(\bi{x_{2}},\bi{\theta})dG_{0}(\bi{\theta})
-E_{G_0}(K(\bi{x_{1}},\bi{\theta}))E_{G_0}(K(\bi{x_{2}},\bi{\theta})).\label{eq:cov_G0}
\end{equation}
%
\end{theorem}
\begin{corollary}
It follows from the above theorem that for $i=1,2$, if 
$\int K^2(\bi{x}_i,\bi{\theta})dG_{0}(\bi{\theta}) < \infty$, then
\begin{equation}
\mbox{Var}(f(\bx_i))=\frac{\mbox{Var}_{G_0}(K(\bx_i,\btheta))}{\alpha+1}
\label{eq:var_f}
\end{equation}
and
\begin{align}
\mbox{Corr}(f(\bi{x_{1}}),f(\bi{x_{2}}))=&\mbox{Corr}_{G_0}(K(\bi{x_{1}},\bi{\theta}),K(\bi{x_{2}},\bi{\theta}))
 \times \mbox{Corr}(G_{\bx_1},G_{\bx_2}),\label{eq:corr_f}
\end{align}
where 
\begin{equation}
\mbox{Corr}(G_{\bx_1},G_{\bx_2})=
\frac{2}{\alpha+2}\sum_{k\in T(\bi{x_{1}},\bi{x_{2}})}\left(\frac{\alpha}{\alpha+2}\right)^{\#S_{k}}
\left(\frac{\alpha}{\alpha+1}\right)^{\#S_{k}'}.
\label{eq:corr_G}
\end{equation}
\end{corollary}
\noindent The expression for the correlation in (\ref{eq:corr_G}) has been obtained by GS.

The above results provide an expression for the correlation conditional on a fixed ordering. 
To obtain the unconditional correlation it is necessary to marginalize the conditional correlation
over the point process $\Phi$. Following GS we also modify the notation as follows:
we now let $T(\bx_1,\bx_2)=\Phi\cap U(\bx_1)\cap U(\bx_2)$, $A_{\ell k}=A_{\ell}(\bz_k)$, where
$A_{\ell}(\bz)=\{\bw\in\Phi\cap U(\bx_{\ell}):\|\bw-\bx_{\ell}\|<\|\bz-\bx_{\ell}\|\}$,
for $\bz\in\Phi\cap U(\bx_{\ell})$. As already mentioned in Section \ref{subsection:overview}, 
when $\bx=(\bs',t)'$, we define $U(\bx)=D\times (-\infty,t]$.

Also, for $\bz\in T(\bx_1,\bx_2)$, we let $S(\bz)=A_1(\bz)\cap A_2(\bz)$ and 
$S'(\bz)=A_1(\bz)\cup A_2(\bz)-S(\bz)$, which imply that 
$S(\bz)=\{\bw\in T(\bx_1,\bx_2):\|\bw-\bx_1\|<\|\bz-\bx_1\|
~\mbox{and}~\|\bw-\bx_2\|<\|\bz-\bx_2\|\}$. 

We further define, as in GS,
$S_{-\bz}(\bz)$ and $S'_{-\bz}(\bz)$ to be translations of $S(\bz)$ and $S'(\bz)$, respectively, by $-\bz$.
Then, the refined Campbell theorem yields, in the case where $\Phi$ is a stationary point process
with intensity $\lambda$:
\begin{align}
\mbox{Corr}(f(\bx_1),f(\bx_2))&=\mbox{Corr}_{G_0}(K(\bx_1,\btheta),K(\bx_2,\btheta))\notag\\
&\times\frac{2\lambda}{\alpha+2}\int_{U(\bx_1)\cap U(\bx_2)}\int
\left(\frac{\alpha}{\alpha+2}\right)^{\phi_{-\bz}(S_{-\bz})}
 \left(\frac{\alpha}{\alpha+1}\right)^{\phi_{-\bz}(S'_{-\bz})}P_0(d\phi)d\bz
 \label{eq:uncond_corr}
\end{align} 
In (\ref{eq:uncond_corr}), $P_0(d\phi)$ is the Palm distribution of $\Phi$ at the origin,
and $\phi_{-\bz}$ is the realization of $\Phi$ translated by $-\bz$. 
Note also that the second factor of the above correlation is the unconditional correlation between
$G_{\bx_1}$ and $G_{\bx_2}$ (see GS). 

\begin{remark}
\label{remark:remark1}
It is worth pointing out that unlike \ctn{Gelfand05} who obtained covariance structure conditional on the random process $\bG$, 
in our case, the covariance structures conditional on the random measures $\bG_{\bx}$ are not relevant,
since it follows from (\ref{eq:kc1}) and the subsequent discussion that $\mbox{Cov}\left(y_1,y_2|\bG_{\bx_1},\bG_{\bx_2}\right)=0$. 
Indeed, dependence among the responses is induced through dependence among $\bG_{\bx}$.
\end{remark}

The following theorem, the proof of which is provided in Section S-3 of the supplement, 
shows that the above correlation structure of our kernel convolution based
ODDP satisfies desirable properties.
 
\begin{theorem}
\label{theorem:corr}
$\mbox{Corr}(f(\bx_1),f(\bx_2))\rightarrow 1$ as $\|\bx_1-\bx_2\|\rightarrow 0$
and $\mbox{Corr}(f(\bx_1),f(\bx_2))\rightarrow 0$ as $\|\bx_1-\bx_2\|\rightarrow\infty$.
\end{theorem}
It is clear from the above theorem and model (\ref{eq:model}) that $\mbox{Corr}(y_i,y_j)\rightarrow 1$ as $\|\bx_i-\bx_j\|\rightarrow 0$
and $\mbox{Corr}(y_i,y_j)\rightarrow 0$ as $\|\bx_i-\bx_j\|\rightarrow \infty$.


Under a stationary Poisson process assumption
for $\Phi$, and for particular specifications of $U(\bx)$ mentioned in Section \ref{subsection:overview}, 
the calculations of 
GS show that the second factor of (\ref{eq:uncond_corr}) depends 
upon $\bx_1$ and $\bx_2$ only through $\|\bx_1-\bx_2\|$, leading to isotropy of the
process. There does not seem to exist any result analogous to the refined Campbell
theorem in the context of nonstationary Poisson process which might allow one to construct a
nonstationary correlation structure in this case. The analytic form of the ODDP correlation structure 
need not be available
for other constructions of $U(\bx)$ either. Isotropy results even in the case of the more flexible
Cox processes. Note that the correlations between any two responses $y_i$ and $y_j$ may correspond to nonstationarity if
their expectations under the density $f_{\btheta}(\cdot)$ are nonlinear in $\btheta$. However, there is no guarantee that the correlation tends to zero
as $\|\bx_i-\bx_j\|\rightarrow\infty$.

On the other hand, our kernel convolution idea neatly solves this problem of attainment of 
nonstationarity via the first factor of our correlation structure given in (\ref{eq:uncond_corr}).
Indeed, the kernel $K(\bx,\btheta)$ can be chosen in the spirit of \ctn{Higdon99},
for instance, such that $\mbox{Corr}_{G_0}(K(\bx_1,\btheta),K(\bx_2,\btheta))$ does not depend upon
$\bx_1-\bx_2$ alone. In other words, by simply controlling the kernel we can ensure nonstationarity
of our process $f(\cdot)$ even if the underlying ODDP is stationary or even isotropic.
Of course, our process can be made stationary as well by choosing the kernel, say, in the spirit of \ctn{Higdon98},
and setting $U(\bx)$ to be of the forms specified by GS, when $\bx$ consists of either only spatial co-ordinates
or only temporal co-ordinate. When $\bx=(\bs',t)'$, then we set $U(\bx)=D\times (-\infty,t]$, as already
mentioned before. 

We further note that our general space-time correlation structure given by (\ref{eq:corr_f}) is nonseparable, that is,
in general, $\mbox{Corr}(f(\bs_1,t_1),f(\bs_2,t_2))\neq \mbox{Corr}_1(\bs_1,\bs_2)\times \mbox{Corr}_2(t_1,t_2)$, where 
$\mbox{Corr}_1$
and $\mbox{Corr}_2$ are spatial and temporal structures respectively. However, if desired,
separability can be easily induced by allowing the kernel to depend upon only the spatial location and
by allowing the ordering $\bpi$ to depend only upon time, or the vice versa. 
Specifically,  
letting $K(\bx,\btheta)=K(\bs,\btheta)$ and $\bpi(\bx)=\bpi(t)$, we obtain
\begin{equation}
\mbox{Corr}(f(\bs_1,t_1),f(\bs_2,t_2))=\mbox{Corr}_{G_0}(K(\bs_1,\btheta),K(\bs_2,\btheta))\times \mbox{Corr}(G_{t_1},G_{t_2}),
\label{eq:separable1}
\end{equation}
and letting $K(\bx,\btheta)=K(t,\btheta)$ and $\bpi(\bx)=\bpi(\bs)$, we obtain
\begin{equation}
\mbox{Corr}(f(\bs_1,t_1),f(\bs_2,t_2))=\mbox{Corr}_{G_0}(K(t_1,\btheta),K(t_2,\btheta))\times \mbox{Corr}(G_{\bs_1},G_{\bs_2}),
\label{eq:separable2}
\end{equation}
In contrast, under the ODDP approach of GS, it is clear from the correlation structure that 
$\mbox{Corr}(G_{\bx_1},G_{\bx_2})\neq \mbox{Corr}_1(\bs_1,\bs_2)\times \mbox{Corr}_2(t_1,t_2)$, 
showing that separability can not be
enforced if desired.

Thus, following our approach it is easy to construct nonparametric covariance structures that are either 
stationary or nonstationary,
which, in turn, can be constructed as either separable or nonseparable, as desired. 
These illustrate the considerable flexibility inherent
in our approach, while satisfying at the same time the desirable conditions that the correlation
between $f(\bx_1)$ and $f(\bx_2)$ tends to 1 or zero accordingly as the distance between $\bx_1$ and $\bx_2$
tends to zero or infinity. 

\section{Continuity and smoothness properties of our model}
\label{sec:smoothness}

For stationary models, properties like continuity and smoothness can be quite generally characterized by 
the continuity and smoothness of the correlation function.
In particular, continuity and smoothness of stationary processes typically depend upon the behaviour of 
the correlation function
at zero; see \ctn{Yaglom87a} and \ctn{Yaglom87b} for details. For nonstationary processes, however, such elegant
theory is not available. Indeed, the structure of the correlation function itself may be difficult to get hold of,
rendering it difficult to investigate the properties of the underlying nonstationary stochastic process.
For our purpose, we utilize the notions of almost sure continuity, mean square continuity and mean square
differentiability of stochastic processes (see, for example, \ctn{Stein99}, \ctn{Banerjee03}) to study the properties of 
our nonstationary spatio-temporal process.


\begin{definition}
A process $\{X(\bi{x}),\bi{x}\in \mathbb{R}^d\}$ is $L_2$ continuous at $\bi{x}_0$ if
$\underset{\bi{x}\rightarrow \bi{x}_0}{\lim} E[X(\bi{x})-X(\bi{x}_0)]^2 =0$.
Continuity in the $L_2$ sense is also referred to as mean square
continuity and will be denoted by $X(\bi{x})\stackrel{L_2}\rightarrow X(\bi{x}_0)$.
\end{definition}

\begin{definition}
A process $\{X(\bi{x}),\bi{x}\in \mathbb{R}^d\}$ is almost surely continuous at $\bi{x}_{0}$ if
$X(\bi{x})\rightarrow X(\bi{x}_0)$ $a.s.$ as $\bi{x}\rightarrow\bi{x_{0}}$.
If the process is almost surely continuous for every $\bi{x_{0}}\in \mathbb{R}^{d}$ then the process is said
to have continuous realizations.
\end{definition}

\begin{theorem}
Assume the following conditions:
\begin{enumerate}
 \item[(A1)] For all $\bx$ and $\btheta$, $\left|K(\bx,\btheta)\right|<M$ for some $M<\infty$.
 \item[(A2)] Given any $\btheta$, $K(\bx,\btheta)$ is a continuous function of $\bx$.
\end{enumerate}
Then $f(\cdot)$ is both almost surely continuous and mean square continuous in the interior of 
$\cap_{k=1}^{\infty} A_{ki_k}$, where
$A_{ki_k}=\{\bx:\pi_k(\bx)=i_k\}$, and for each $k=1,2,\ldots,$ $i_k\in\{1,2, \ldots\}$;
$i_k\neq i_{k'}$ for any $k\neq k'$. On the other hand, $f(\cdot)$ is almost surely discontinuous 
at any point $\bx_0\in\cap_{k=1}^{\infty} A_{ki_k}$ lying on the boundary of $A_{ki_k}$, for any
$i_k$.
\end{theorem}

See Section S-4 for a proof of this result.
%
Now we examine mean square differentiability of our process. 

\begin{definition}
A process $\{X(\bi{x}),\bx\in \mathbb{R}^d\}$ is said to be mean square differentiable at $\bi{x}_0$ if for any 
direction $\bu$, there exists a process $L_{\bi{x}_0}(\bu)$, linear in $\bu$ such that
\begin{align*}
 X(\bi{x}_0+\bu)=X(\bi{x}_0)+ L_{\bi{x}_0}(\bu) + R(\bi{x}_0,\bu),
\mbox{ where } \frac{R(\bi{x}_0,\bu)}{\|\bu\|} \stackrel{L_2}\rightarrow 0.
\end{align*}
\end{definition}

\begin{theorem}
Assume the following conditions:
\begin{enumerate}
 \item[(B1)] For all $\bx$ and $\btheta$, $\left|K(\bx,\btheta)\right|<M$ for some $M<\infty$.
 \item[(B2)] Given any $\btheta$, $K(\bx,\btheta)$ is a continuously differentiable function of $\bx$.
\end{enumerate}
Then $f(\cdot)$ is mean square differentiable in the interior of $\cap_{k=1}^{\infty} A_{ki_k}$.
\end{theorem}

See Section S-5 for a proof of this theorem.

In real life applications most of the spatio-temporal processes are expected to be irregular in nature. One of the desirable properties
of a spatio-temporal model is that, it allows the different degrees of smoothness across space than across time. 
Our model has achieved this property regarding smoothness.
For example, if we associate the ODDP prior only to the spatial locations, then the process  becomes smoother 
across time than across space depending on the choice of the kernel.


\section{Truncation of the infinite summand}
\label{sec:truncation}

Since our proposed model $f(\bx)=\sum_{k=1}^{\infty}K(\bx,\btheta_{\pi(\bx)})p_i(\bx)$
is an infinite (random) series, for model-fitting purpose it is necessary to truncate the series
to $f(\bx)=\sum_{k=1}^{N}K(\bx,\btheta_{\pi(\bx)})p_i(\bx)$, where $N$ is to be determined,
or to implement variable-dimensional Markov chain methods where $N$ is to considered a random variable
so that the number of parameters associated with $f(\bx)$ is also a random variable. 

Although we will describe and implement TTMCMC,
we first prove a theorem with respect to truncation of the infinite
random series. Note that in the context of traditional Dirichlet process characterized by Sethuraman's
stick breaking construction (\ctn{Sethuraman94}) which involves infinite random series, \ctn{Ishwaran01} proposed a method of
truncating the infinite series. 



We now state our theorem on truncation, the proof of which is provided in Section S-6
of the supplement.
But before stating the theorem it is necessary to define some required notation.
Let 
$$P_N(\bx_{i})=\sum_{i=1}^N K(\bx_{i},\theta_{i})p_{i},\ \ \mbox{and} \ \ 
P(\bx_{i})=\sum_{i=1}^\infty K(\bx_{i},\theta_{i})p_{i},$$
where $N$ needs to be determined. 
Also let 
$$P_N=(P_N(\bx_1),\ldots,P_N(\bx_n))' \ \ \mbox{and} \ \ P=(P(\bx_1),\ldots,P(\bx_n))',$$
and denote by $\Theta_N$ and $\Theta$ the sets of random quantities $(V_i,\btheta_i)$ associated
with $P_N$ and $P$ respectively.
We define the following marginal densities of the vector of observations $\bi{y}=\{y_1,\ldots,y_n\}$, where $[\cdot|\cdot]$ and $[\cdot]$ denote
conditional and marginal densities, respectively: 
\begin{align*}
m_N({\bi{y}}) &=\int_{\Theta_N} [\bi{y}|P_{N}][P_{N}]d\Theta_N
\\
 &=\int_\Theta [\bi{y}|P_{N}][P]d\Theta,
\end{align*}
and
\[
 m_\infty(\bi{y})=\int_\Theta[\bi{y}|P][P]d\Theta.
\]

\begin{theorem}
\label{theorem:truncation2}
Under the assumption
that $\underset{\btheta}{\sup} K(x_{i},\btheta)\leq M $ for $i=1,\ldots,n$, where $M>0$
is a finite constant,
we have 
\[
 \int _{\mathbb{R}^n} \left|m_N(\bi{y})-m_\infty(\bi{y})\right|d\bi{y} 
 \leq 4M^2n \left(\frac{\alpha}{\alpha+2}\right)^N + 2\sqrt{\frac{2}{\pi}}Mn\left(\frac{\alpha}{\alpha+1}\right)^N.
\]
\end{theorem}

\section{Choice of kernel, prior distributions and computational region}
\label{sec:kernel_prior}
The choice of kernel $K(\cdot,\cdot)$ plays a crucial role in nonstationary spatio-temporal data analysis.
For instance, if $K(\bx,\btheta)=K(\bx-\btheta)$, then the correlation between $Y(\bx_1)$ and $Y(\bx_2)$
turns out to be a function of $\bx_1-\bx_2$, thus inducing stationarity. 
For the purpose of nonstationarity, it is necessary to make the parameters of the kernel depend upon space and
time. In the spatial context such nonstationary kernels are considered in \ctn{Higdon99}.
In this paper, we consider a nonstationary space-time kernel; for the spatial part of the kernel
we essentially adopt the dependence structure and the associated prior distributions proposed by \ctn{Higdon99} and for the temporal part
we allow the relevant coefficient to be time varying, modeled by a stationary Gaussian process.

%
%
%

In particular, we consider the following kernel for our applications:
\[
 K(\bs,t,\btheta,\tau)=\exp\left\{-\frac{1}{2}(\bs-\btheta)^{T}\Sigma (\bs)(\bs-\btheta)-\delta (t)|t-\tau|\right\},
\]
where $\Sigma(\bs)$ is a $2\times 2$ positive definite dispersion matrix depending upon $\bs$, and 
$\delta(t)>0$ depends upon time $t$. We assume that $\log(\delta(t))$ is a zero mean Gaussian process
with covariance $c_{\delta}(t_1,t_2)=\sigma^2_{\delta}\exp\left\{(t_1-t_2)^2/a_{\delta}\right\}$.
We set
\begin{align}
\Sigma (\bs)^{\frac{1}{2}}&=\varphi\left(\begin{array}{cc}
\left[\frac{\sqrt{4 A^2+\|\psi (\bs)\|^4\pi^2}}{2\pi}+\frac{\|\psi (\bs) \|^2}{2}\right]^{\frac{1}{2}} & 0\\
0 & \left[\frac{\sqrt{4 A^2+\|\psi (\bs)\|^4\pi^2}}{2\pi}-\frac{\|\psi (\bs) \|^2}{2}\right]^{\frac{1}{2}}
\end{array}\right)
\left(\begin{array}{cc}\cos~\alpha (\bs) & \sin~\alpha (\bs)\\
-\sin~\alpha (\bs) & \cos~\alpha (\bs)
\end{array}
\right),
\notag
\end{align}
where $\|\psi (\bs)\|^2=\psi^2_1(\bs)+\psi^2_2(\bs)$ and $\alpha (\bs)=\tan^{-1}\left(\frac{\psi_2(\bs)}{\psi_1(\bs)}\right)$.
We assume 
that $\psi_1(\cdot)$ and $\psi_2(\cdot)$ are independent and identical 
zero mean Gaussian processes with covariance 
$c_{\psi}(\bs_1,\bs_2)=\sigma^2_{\psi}\exp\left\{-\|\bs_1-\bs_2\|^2/b_{\psi}\right\}$.
We put the $U(3,200)$ prior on $\varphi$, $a_{\delta}$, 
and $b_{\psi}$; we set $\sigma^2_{\delta}=\sigma^2_{\psi}=1$. Also, 
we set $A=3.5$.
Since in our applications we center and scale the observed time points, for $\tau$ we specify the $N(0,1)$ prior.

\subsection{Elicitation of hyperparameters of the underlying ODDP}
\label{subsec:hyperprior_dp}

\subsubsection{Choice of $G_0$}
\label{subsubsec:G_0}

In our applications, we center and scale each of the two components $\{s_{1i};~i=1,\ldots,n\}$ 
and $\{s_{2i};~i=1,\ldots,n\}$ of the available spatial locations
$\{\bs_i=(s_{1i},s_{2i});~i=1,\ldots,n\}$. Consequently, the choosing $G_0$ to be the
bivariate normal distribution with both means zero, both variances equal to one, and correlation
$\rho$ appears to be reasonable. We estimate $\rho$ by the empirical correlation between
$\{s_{1i};~i=1,\ldots,n\}$ and $\{s_{2i};~i=1,\ldots,n\}$.

\subsubsection{Prior selection for $\alpha$}
\label{subsubsec:prior_alpha}
For the choice of prior distributions of the parameters associated with the ODDP we follow \ctn{Griffin04} and 
GS. 
In particular, we put the 
inverted Beta distribution prior on $\alpha$, given by
\[
 p(\alpha)=\frac{{n_0}^\eta\Gamma(2\eta)\alpha^{\eta-1}}{{\Gamma(\eta)}^2 (\alpha+n_0)^{2\eta}},
\]
where the hyperparameter $n_0$ is the prior median of $\alpha$. Note that the prior variance exists if 
$\eta>2$ and is a decreasing function of $\eta$. This prior implies that $\frac{\alpha}{\alpha+n_0}$ follows a 
$Beta(\eta,\eta)$ distribution.

\subsubsection{Prior selection for $\lambda$}
\label{subsubsec:prior_lambda}
Note that, for  small $\alpha$, only the first few elements of stick breaking representation are important, 
so fewer number of points from the underlying Poisson process is needed to induce  
the second factor of the correlation structure (\ref{eq:uncond_corr})
which roughly depends upon the ratio $\lambda/(\alpha+1)$ for $U(\bx)=D$ (spatial problem) and 
$U(x)=(-\infty,x]$ (temporal problem); see GS for the details. Thus
a relatively small value of $\lambda$ suffices in such cases.
Similarly, when $\alpha$ is larger, 
larger $\lambda$ is necessary to obtain the same correlation. Keeping these in mind, we select the log-normal prior for 
$\lambda$ with mean $\log(\alpha)$ and variance $b_{\lambda}$, say. For our applications, we choose
$b_{\lambda}=20$, so that we obtain a reasonably vague prior.

\subsection{Computational region}
\label{subsec:comp_region}

Following GS, we consider a truncated region for the point process $Z$ which includes the range 
of the observed $\bx$. This truncated region has been referred to as the computational region by GS.
In particular, we choose a bounding box of the form $(a_1,b_1)\times (a_2,b_2)\times\cdots\times (a_d,b_d)$ 
as the computational region, where $a_i=d_{a_i}-r$, $b_i=d_{b_i}-r$. Here   
$d_{a_i}$ and $d_{b_i}$ are the minimum and the maximum of $\bx$ in dimension $i$, and 
$r=2\left(\frac{\Gamma(d/2)d}{2\pi^{d/2}}\frac{\alpha+1}{\lambda}\log\frac{1}{\epsilon}\right)^{\frac{1}{d}}$,
with $\epsilon=\exp\left\{-\frac{\lambda}{\alpha+1}\frac{2\pi^{d/2}}{\Gamma(d/2)d}\left(\frac{r}{2}\right)^d\right\}$.
See GS for justification of these choices.

\section{Joint posterior and a briefing of TTMCMC for updating parameters in our
variable dimensional modeling framework}
\label{sec:joint_posterior}

Let $k$ denote the random number of summands in
\begin{equation}
f_k(\bx)=\sum_{i=1}^{k}K(\bx,\btheta_{\pi_{i}(\bx)})p_{i}(\bx) \mbox{       }
\forall \bx\in D\subseteq\mathbb{R}^d.
\label{eq:kernel_convolution2}
\end{equation}
Let $\bV=(V_1,\ldots,V_k)$, $\bz=(z_1,\ldots,z_k)$, $\btheta=(\btheta_1,\btheta_2)$,
with $\btheta_1=(\theta_{11},\ldots,\theta_{1k})$ and $\btheta_1=(\theta_{21},\ldots,\theta_{2k})$.
Let also $\bpsi_1=(\psi_1(\bs_1),\ldots,\psi_1(\bs_n))$, $\bpsi_2=(\psi_2(\bs_1),\ldots,\psi_2(\bs_n))$
and $\bdelta=(\delta(t_1),\ldots,\delta(t_n))$.
The joint posterior is of the form
\begin{align}
&\pi(k,\bV,\bz,\btheta_1,\btheta_2,\bpsi_1,\bpsi_2,\bdelta,\tau,\sigma,\alpha,\lambda,b_{\psi},a_{\delta}
\vert\bY)\notag\\
&\propto\pi(k)\pi(\bV,\bz,\btheta_1,\btheta_2\vert k)\pi(\bpsi_1,\bpsi_2,\bdelta)
\pi(\tau,\varphi,b_{\psi},a_{\delta})\pi(\sigma)\pi(\alpha)\pi(\lambda\vert\alpha)\times L(\bV,\bz,\btheta_1,\btheta_2,\sigma\vert k,\bY),
\label{eq:joint}
\end{align}
where $L(\bV,\bz,\btheta_1,\btheta_2,\sigma\vert k,\bY)$ is the joint normal likelihood of 
$\bV,\bz,\btheta_1,\btheta_2,\sigma$ under the model
\begin{equation}
y_i=f_k(\bx_i)+\e_i;~\e_i\stackrel{iid}{\sim}N(0,\sigma^2);~i=1,\ldots,n,
\label{eq:likelihood}
\end{equation}
conditional on $f_k(\cdot)$.

For our applications, as the prior $\pi(k)$ on $k$ we assume the discrete uniform prior on $\{1,2,\ldots,30\}$;
in our applications $k$ never even reached $30$. Under $\pi(\bV,\bz,\btheta_1,\btheta_2\vert k)$,
$V_i\stackrel{iid}{\sim}Beta(1,\alpha);~i=1,\ldots,k$, $\bz$ are realizations from the Poisson process
with intensity $\lambda$, and for $i=1,\ldots,k$, $(\theta_{1i},\theta_{2i})\stackrel{iid}{\sim} G_0$.
Under $\pi(\bpsi_1,\bpsi_2,\bdelta)$, $\psi_1,\psi_2,\delta$ are independent Gaussian processes, as detailed
in Section \ref{sec:kernel_prior}. The prior distribution of $\tau,\varphi,b_{\psi},a_{\delta}$,
denoted by $\pi(\tau,\varphi,b_{\psi},a_{\delta})$, is already provided in Section \ref{sec:kernel_prior}.
For the error standard deviation $\sigma$, the prior denoted by $\pi(\sigma)$ is the log-normal distribution
with parameters 0 and 1, so that the mean and variance of $\sigma$ are about 1.6 and 5, respectively. These
quantities appear to be reasonable, and yielded adequate inference.

In order to obtain samples from the joint posterior (\ref{eq:joint}) which involve the
variable dimensional $f_k(\cdot)$, we implement the TTMCMC methodology. In a nutshell, TTMCMC
updates all the parameters, both fixed and variable dimensional, as well as the number
of parameters of the underlying posterior distribution in a single block using simple deterministic
transformations of some low-dimensional random variable drawn from some fixed, but low-dimensional
arbitrary distribution defined on some relevant support. The idea is an extension of 
Transformation based Markov Chain Monte Carlo (TMCMC) introduced
by \ctn{Dutta14} for updating high-dimensional parameters with
known dimensionality in a single block using simple deterministic transformations of some
low-dimensional (usually one-dimensional) random variable having arbitrary distribution
on some relevant support. The strategy of updating high and variable dimensional
parameters using very low-dimensional random variables clearly reduces dimensionality dramatically,
thus greatly improving acceptance rate, mixing properties, and computational speed. 
In Section S-7 of the supplement we provide a detailed overview of TTMCMC, propose a general algorithm
(Algorithm S-7.1) with certain advantages, and in Section S-8 of the supplement we 
specialize the algorithm to our spatio-temporal modeling set-up,
providing full updating details (Algorithm S-8.1).

\section{Simulation study}
\label{sec:simulation_study}


To illustrate the performance of our model we first create a synthetic data generating process
which is nonstationary and non-Gaussian. One popular method to create such process is the kernel 
convolution approach. However, since we have developed our spatio-temporal model itself using the 
kernel convolution approach, it is perhaps desirable to obtain the synthetic data from some 
nonstationary, non-Gaussian process created using some approach independent of the kernel convolution method. 
In Section \ref{subsec:process_simdata} we detail such an approach.
Then we fit our proposed model to the data pretending that
the data-generating process is unknown. 

\subsection{A nonstationary non-Gaussian data generating process}
\label{subsec:process_simdata}

Let $X(\cdot)$ denote a stationary Gaussian process with mean function 
$\mu(t,\bs)=\beta_0+\beta_1 t+\beta_2 s_1+\beta_3 s_2$, with $\bs=(s_1,s_2)$,
and covariance function $$A(i,j)=c((t_i,\bs_i),(t_j,\bs_j))
=\exp\left\{-0.5\left(\sqrt{(t_i-t_j)^2+(s_{1i}-s_{1j})^2+(s_{2i}-s_{2j})^2}\right)\right\},$$
for any $t_i,t_j,\bs_i=(s_{1i},s_{2i}),\bs_j=(s_{1j},s_{2j})$. 

Let $\bX=(X(t_1,\bs_1),\ldots,X(t_n,\bs_n))'$ denote observed data points from the Gaussian process $X$
at the design points $\{(t_i,\bs_i);~i=1,\ldots,n\}$. Let $\bt=(t_1,\ldots,t_n)'$ and
$\bS=(\bs'_1,\ldots,\bs'_n)'$. Further, let us denote by $\bA=(A(i,j); i=1,\ldots,n;~j=1,\ldots,n)$  the
covariance matrix 
and $\bmu=(\mu(t_1,\bs_1),\ldots,\mu(t_n,\bs_n))'$.
Then the posterior process $[X(\cdot)\vert\bX]$ is non-stationary Gaussian with mean function
$\mu_X(t,\bs)=\mu(t,\bs)+\bA_{12}\bA^{-1}_{22}(\bX-\bmu)$ and variance
$\bA_{11}-\bA_{12}\bA^{-1}_{22}\bA_{21}$, where 
$\bA=\left(\begin{array}{cc}\bA_{11} & \bA_{12}\\
\bA_{21} & \bA_{22}\end{array}\right)$.

Let the posterior nonstationary Gaussian process $[X(\cdot)\vert\bX]$ be denoted by $X^*(\cdot)$.
Now, conditionally on the process $X^*(\cdot)$, consider another process $Y(\cdot)$
with mean function $\mu^*(t,\bs)=X^*(t,\bs)$ and covariance function
$c_Y((t_i,\bs_i),(t_j,\bs_j))=\exp\left\{-0.5\left|X^*(t_i,\bs_i)-X^*(t_j,\bs_j)\right|\right\}$.
Then marginally, $Y(\cdot)$ is a nonstationary non-Gaussian process.

For our illustration we will simulate the synthetic dataset from the process $Y(\cdot)$.
The algorithm for generation of this synthetic data is provided in supplementary material (S-9.1).

\subsection{Results of fitting our model to the simulated data}
\label{subsec:simulation_results}

Note that for this problem the number 
of parameters to be updated ranges between 300 to 400. 
Our TTMCMC based model implementation took 35 mins to yield 900000
realizations following a burn-in of 100000. 
Quite encouragingly, TTMCMC exhibited satisfactory acceptance rate and mixing properties. Traceplots are shown in Figure S-9.1 of supplement.  

\subsubsection{Leave-one-out cross-validation}
\label{subsubsec:simulation_xval}
We asses the predictive power of our model with the leave-one-out cross validation method. 
All the 95 cases were included in the 95\% highest posterior densities of the corresponding leave-one-out
posterior predictive densities
Figure \ref{fig:posterior_predictive2} displays the posterior predictive densities of six randomly selected
space-time points, along with the true values, the latter denoted by the vertical lines.
Thus, satisfactory performance of our proposed model is indicated by the results, particularly given
the fact that our model does not assume knowledge of the true, data-generating, parametric model.

\begin{figure}
\centering
\includegraphics[trim={0 0 0 0},clip, totalheight=0.25\textheight]{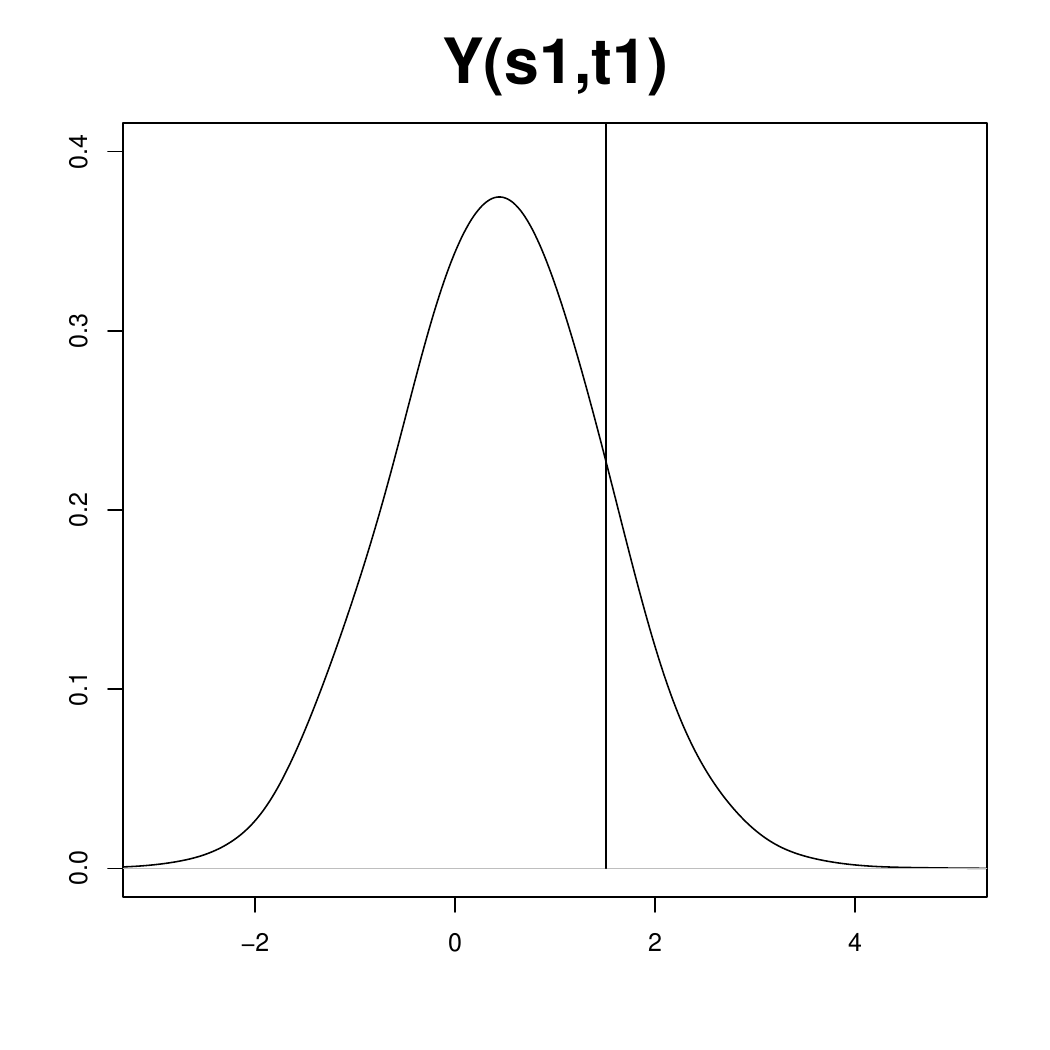}
\includegraphics[trim={0 0 0 0},clip, totalheight=0.25\textheight]{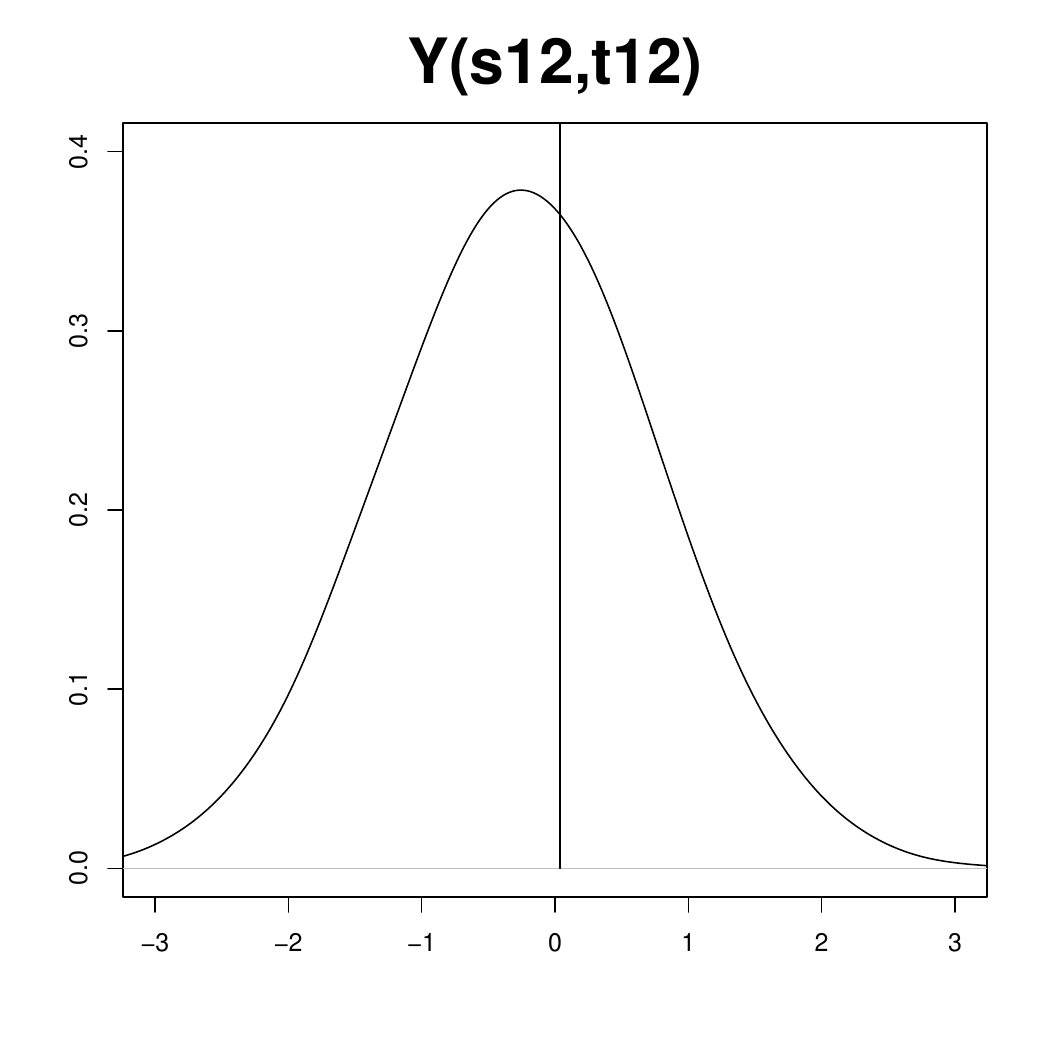}
\includegraphics[trim={0 0 0 0},clip, totalheight=0.25\textheight]{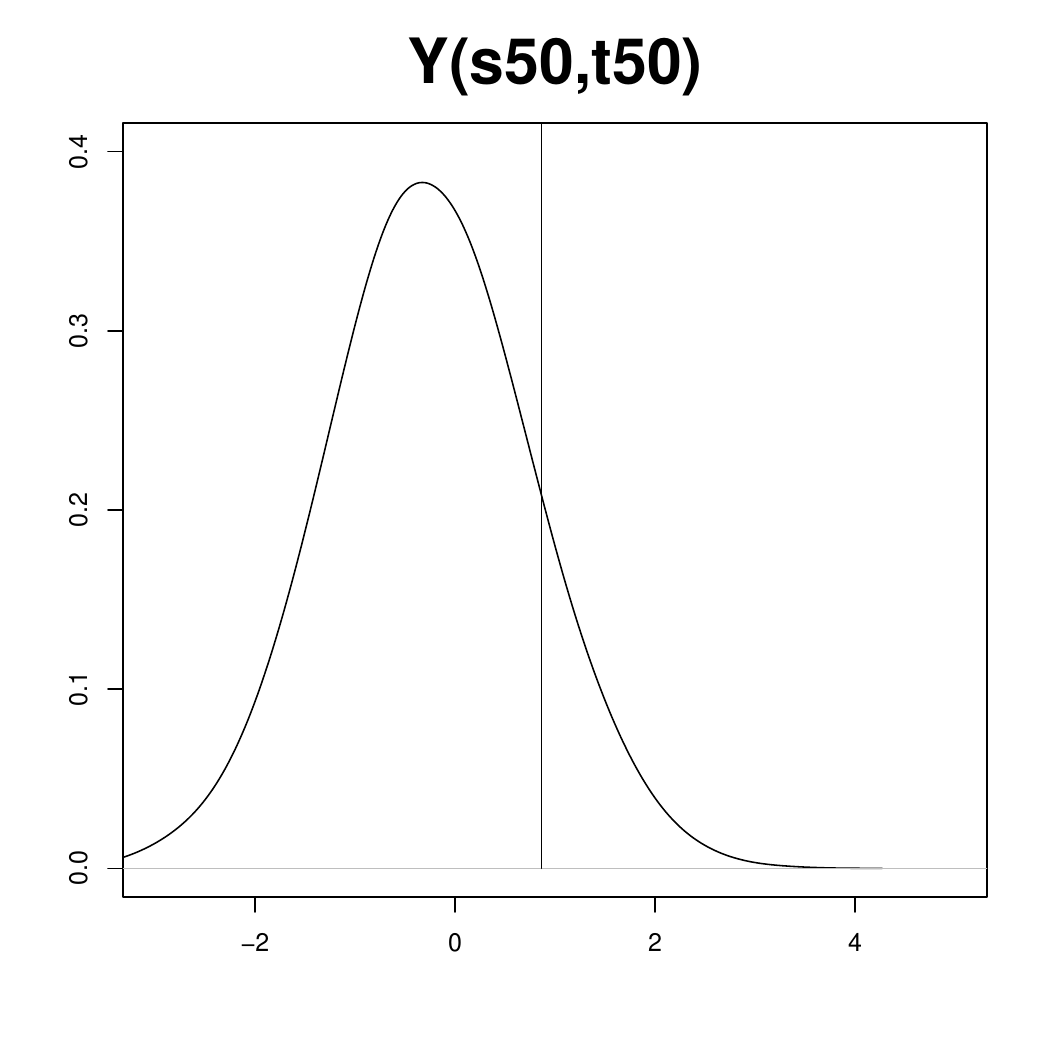}
\includegraphics[trim={0 0 0 0},clip, totalheight=0.25\textheight]{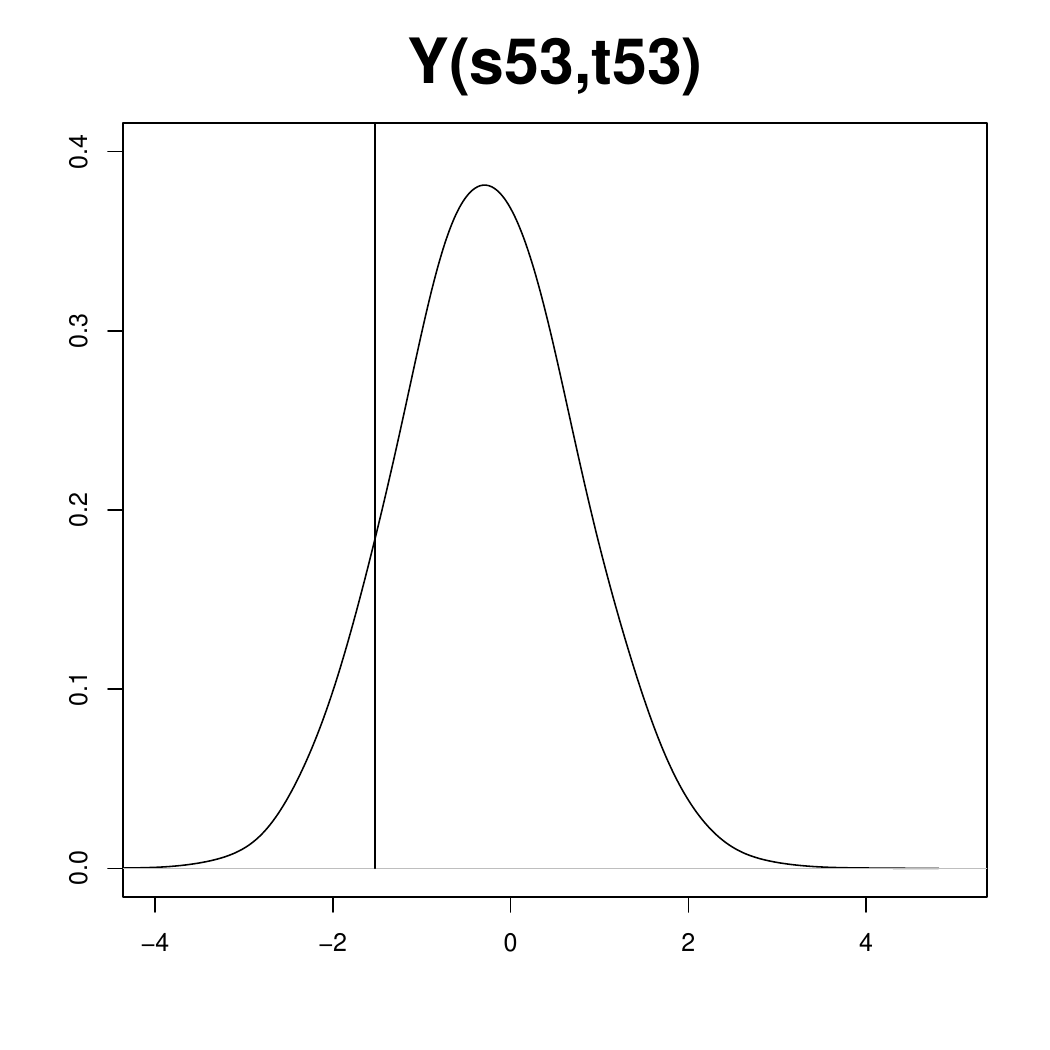}
\includegraphics[trim={0 0 0 0},clip, totalheight=0.25\textheight]{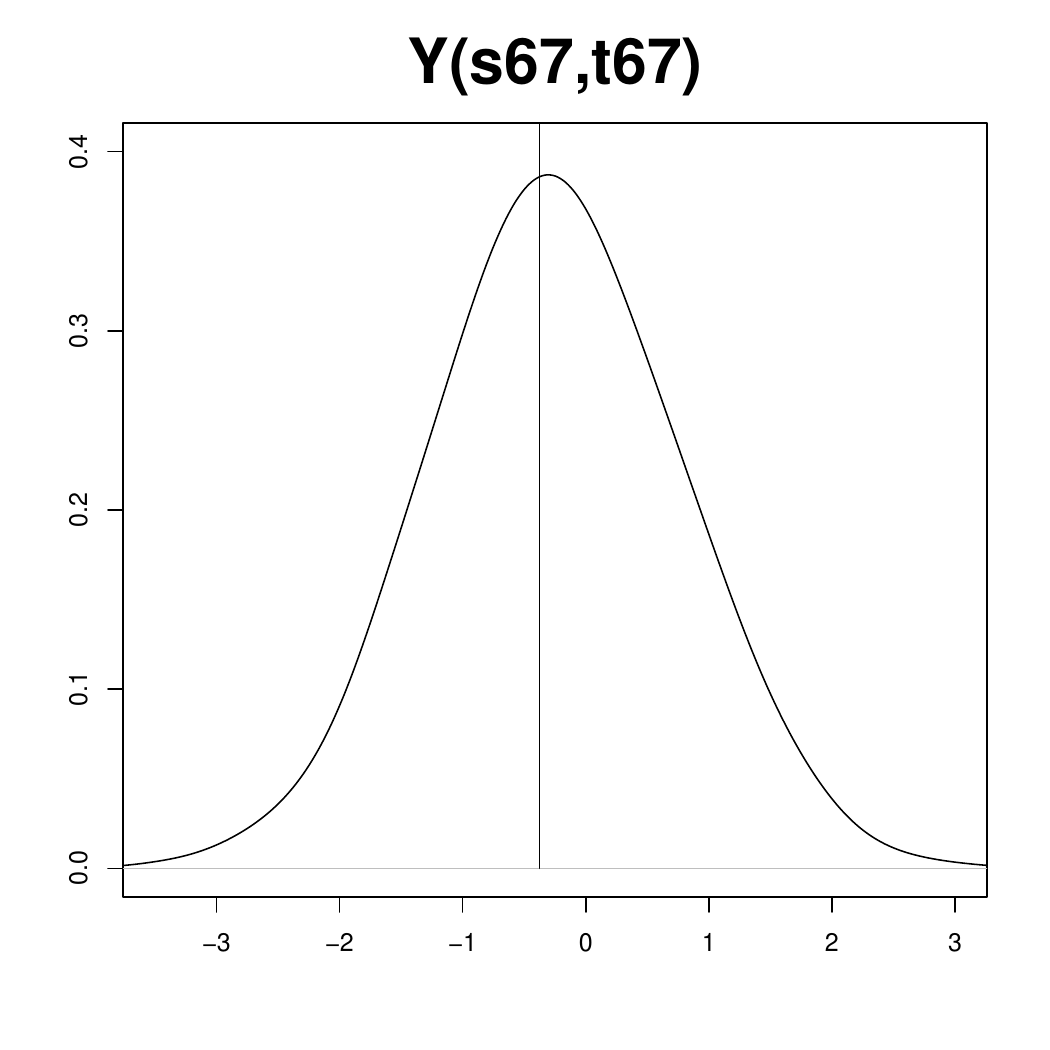}
\includegraphics[trim={0 0 0 0},clip, totalheight=0.25\textheight]{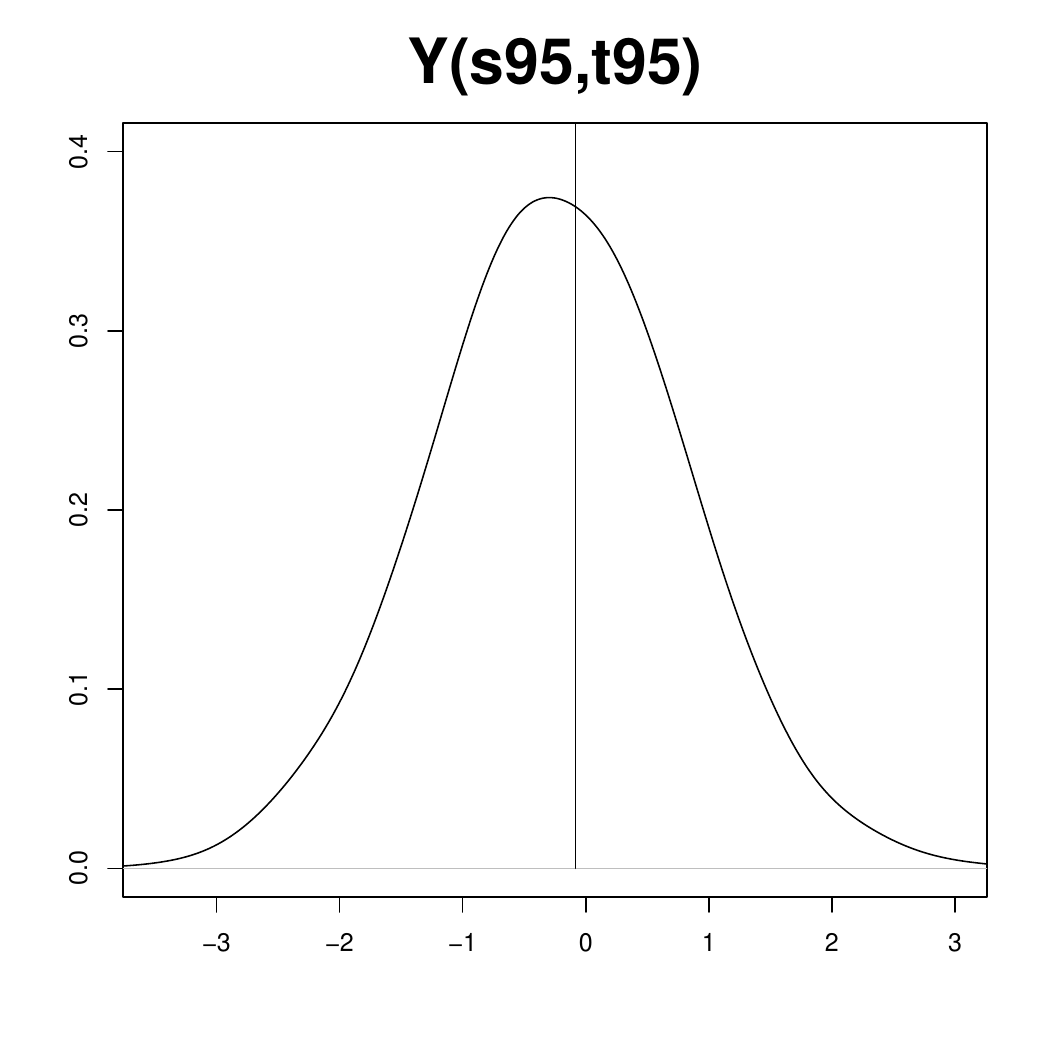}
\caption{{\bf Simulation study}: Posterior predictive densities of $Y(\bs,t)$ for the 6 different location-time pairs of our model
-- the corresponding true values
	are denoted by the vertical line.}
\label{fig:posterior_predictive2}
\end{figure}

\subsubsection{Correlation Analysis}
\label{subsubsec:simulation_corr}
Though our simulation mechanism is completely different from our proposed model, the simulated data
do exhibit the pattern that the correlations are close to zero for two widely separated locations and/or times. 
Indeed, from the structure of the covariance matrix $\bSigma_{(95|5)}$, it is easily seen that the $(i,j)$-th
element ($i\neq j$) of $\bSigma_{(95|5)}$ is close to zero whenever the distance between $t_i$ and $t_j$
and/or $\bs_i$ and $\bs_j$ is large.

We calculate the posterior densities of correlation for different pairs of space-time points. In formation of the pair, 
we select nearby locations, as well as locations which are widely separated, 
such that we obtain both high and low correlation values under the true, data-generating model. 
It is evident from Figure \ref{fig:correlation} that the true correlations, ranging from 
small to high values, lie well within 
their respective $95\%$ credible intervals, vindicating reasonable performance 
of our model in terms of capturing the true correlation structure.

\begin{figure}
\centering
\includegraphics[height=1.5in,width=1.75in]{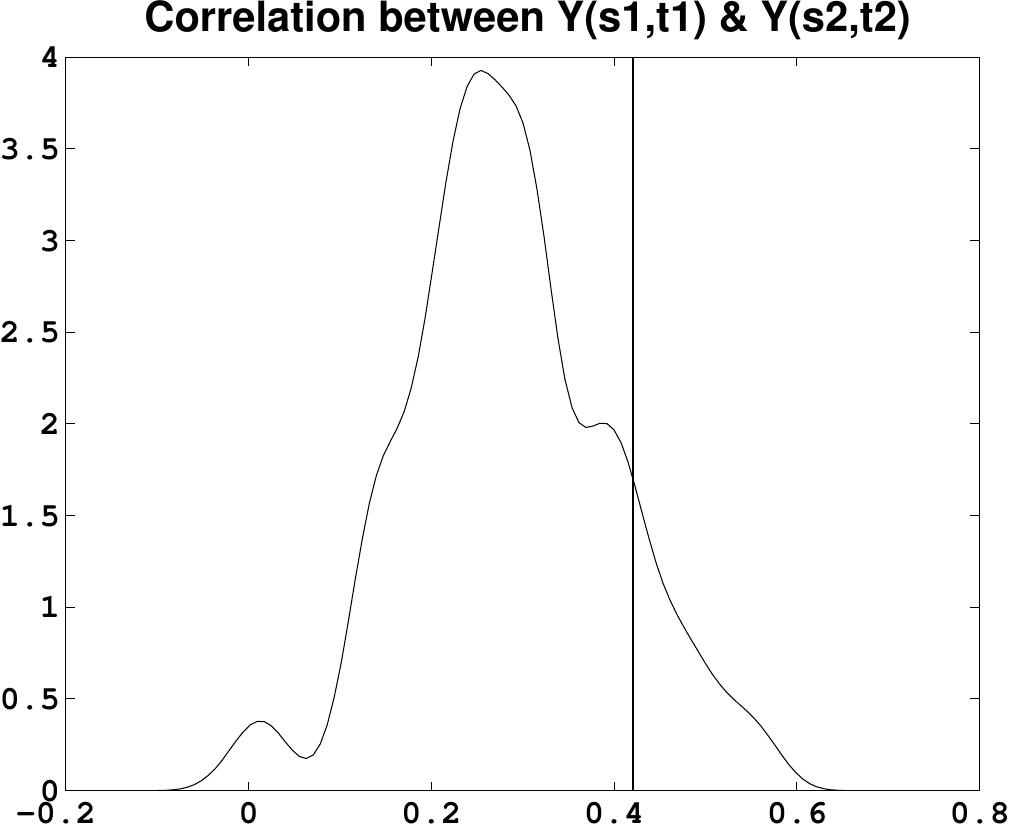}
\includegraphics[height=1.5in,width=1.75in]{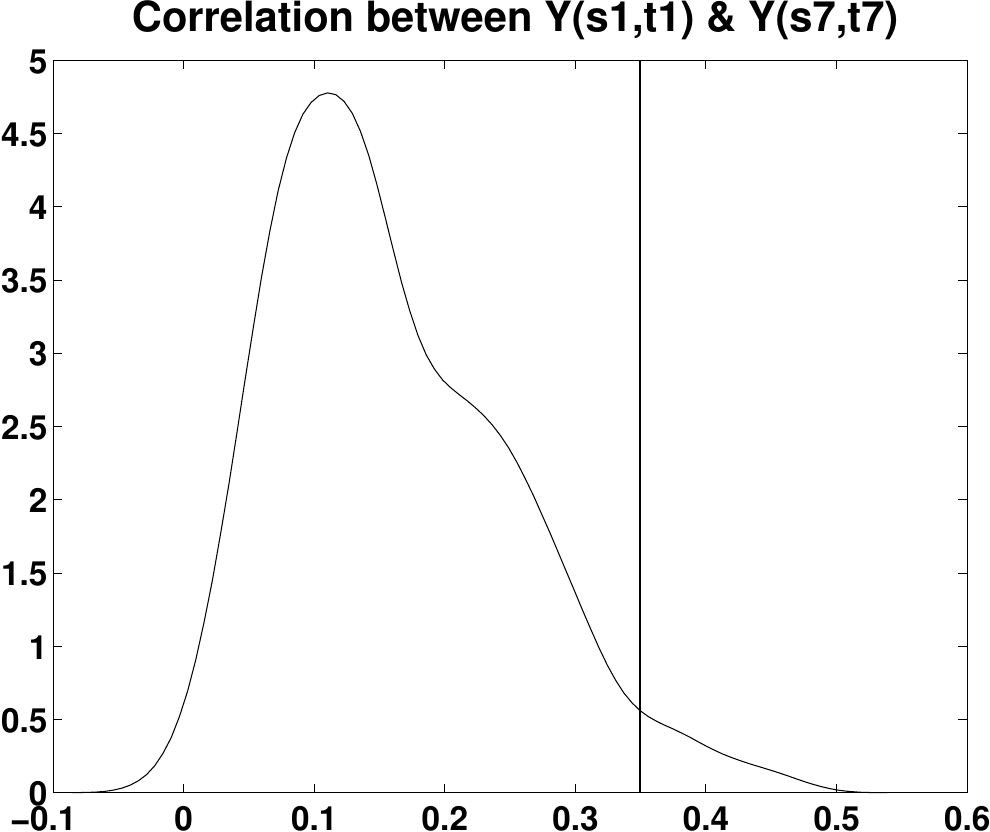}
\includegraphics[height=1.5in,width=1.75in]{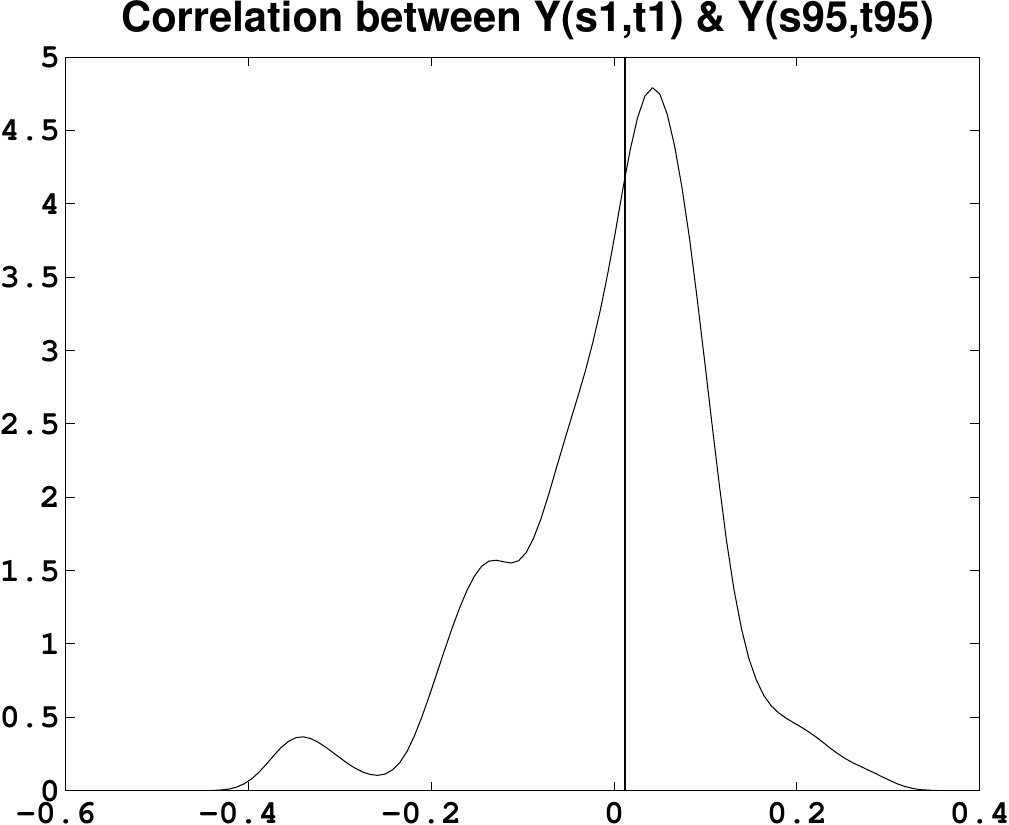}
\\[5mm]
\includegraphics[height=1.5in,width=1.75in]{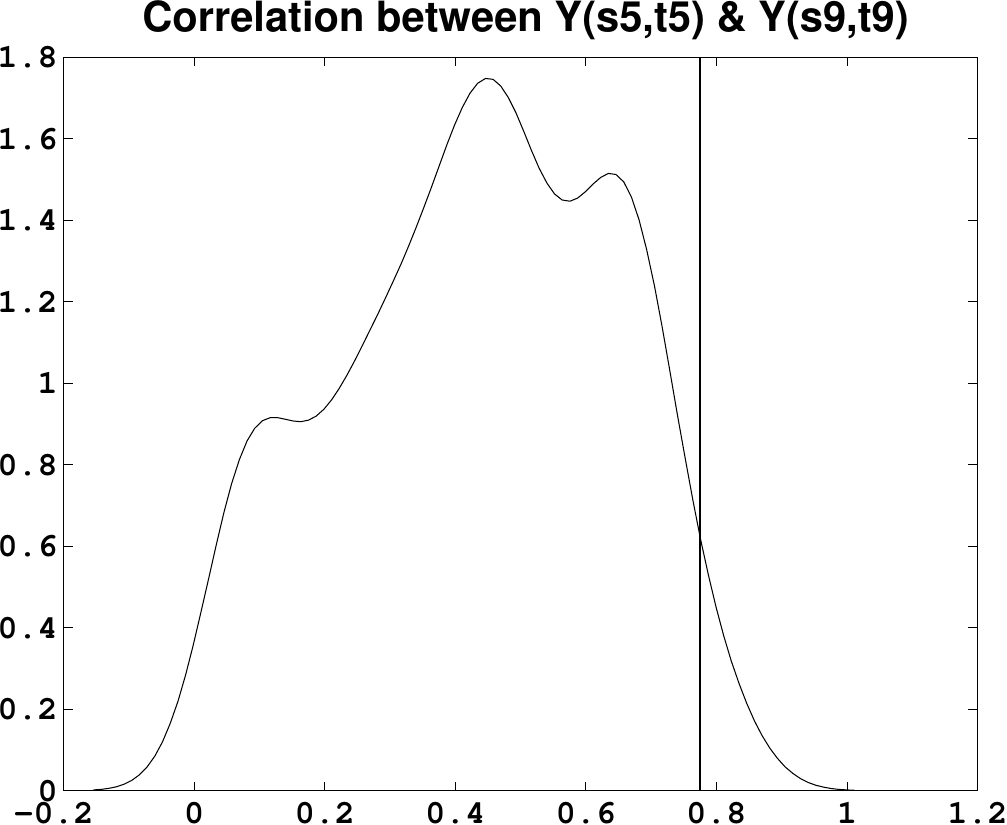}
\includegraphics[height=1.5in,width=1.75in]{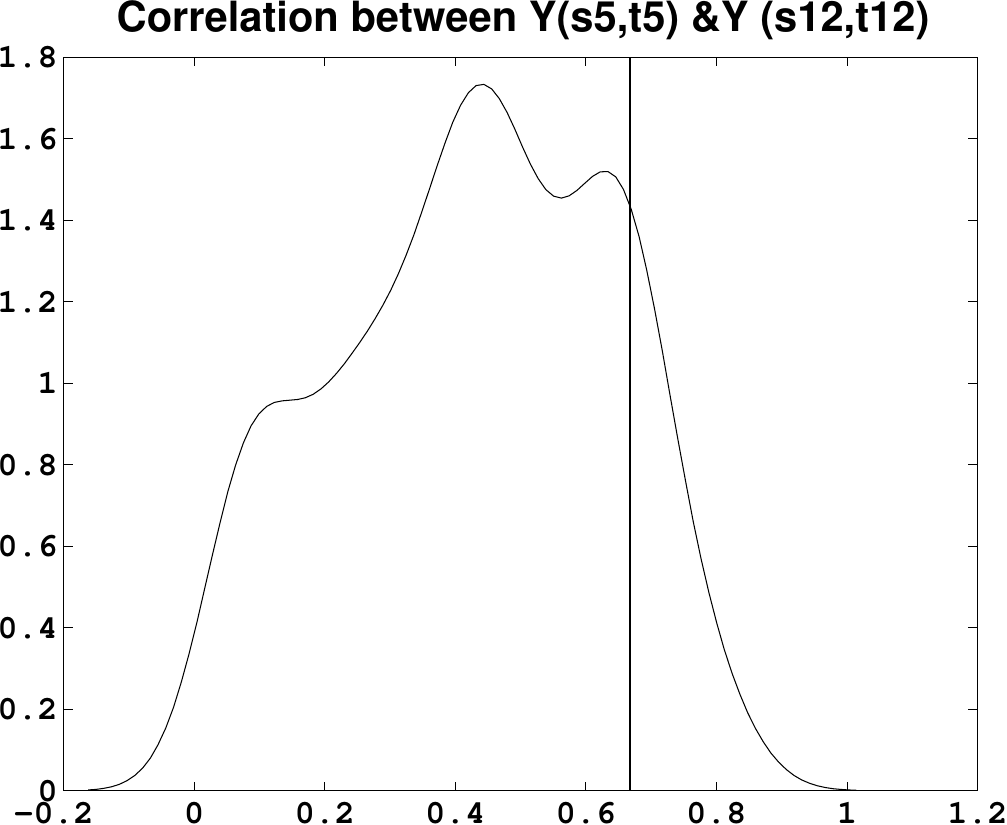}
\includegraphics[height=1.5in,width=1.75in]{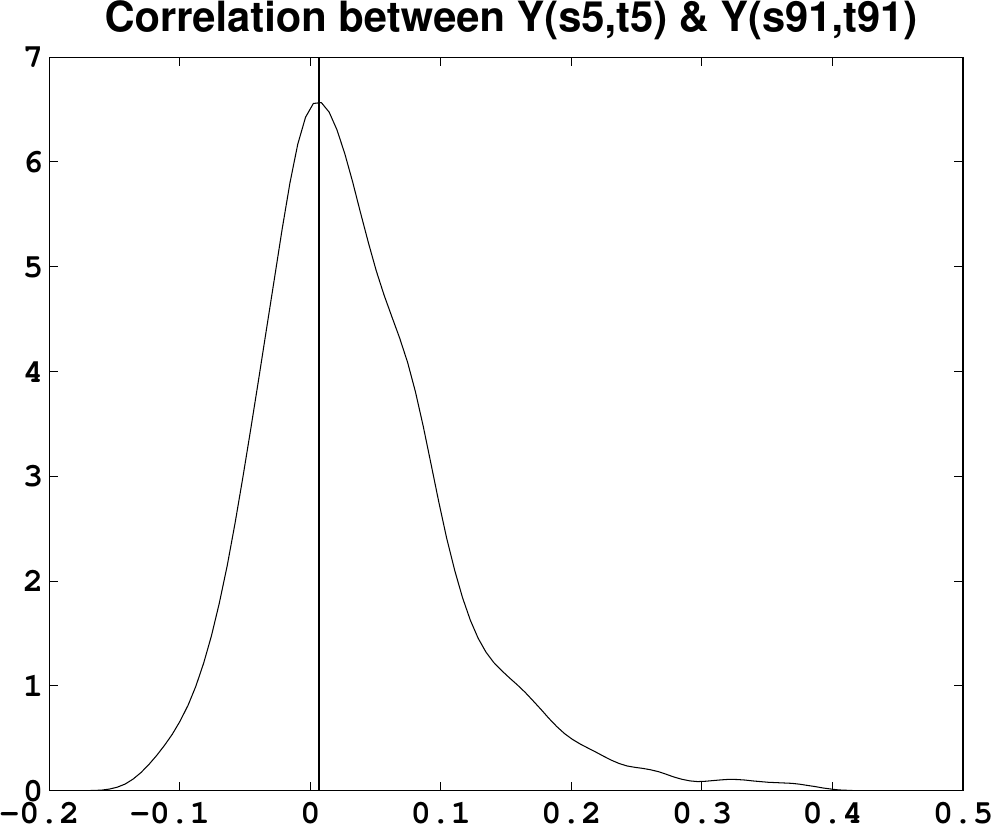}
\\[5mm]
\includegraphics[height=1.5in,width=1.75in]{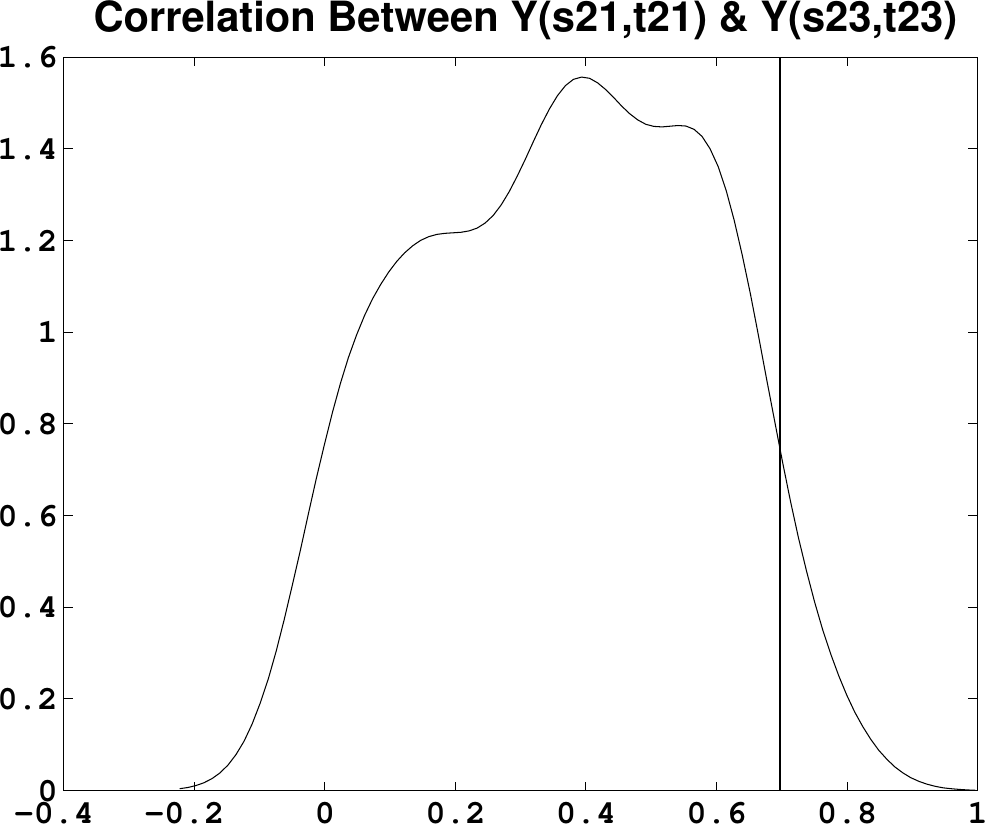}
\includegraphics[height=1.5in,width=1.75in]{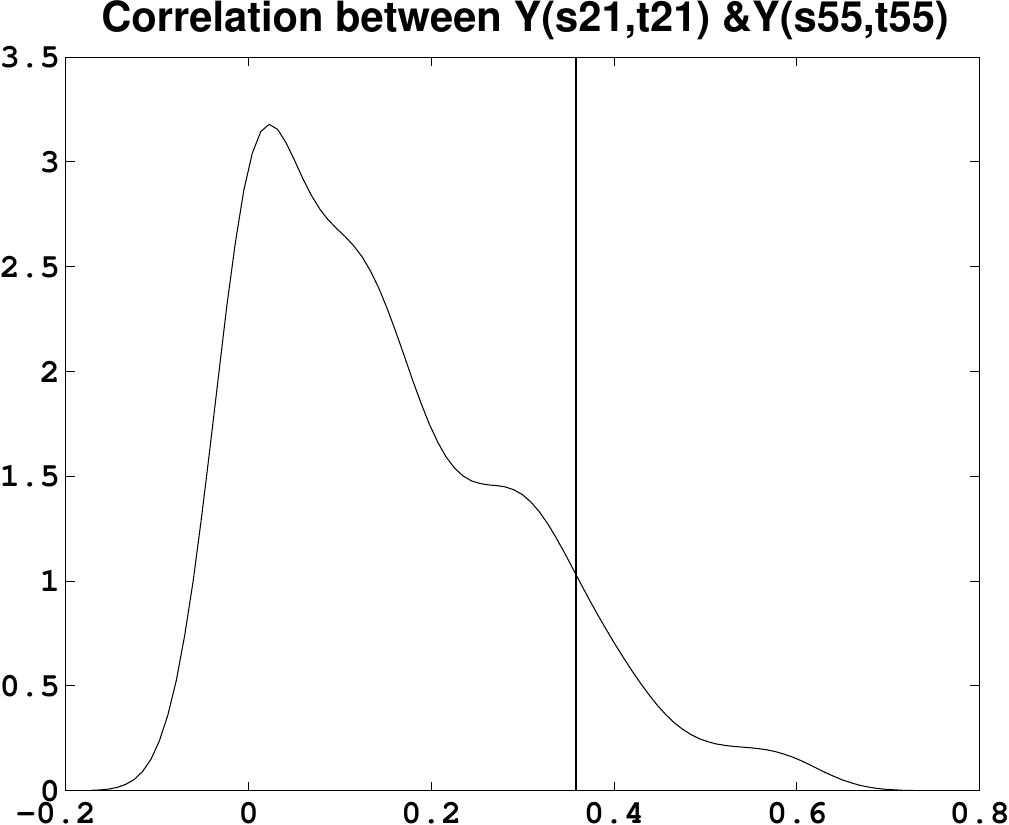}
\includegraphics[height=1.5in,width=1.75in]{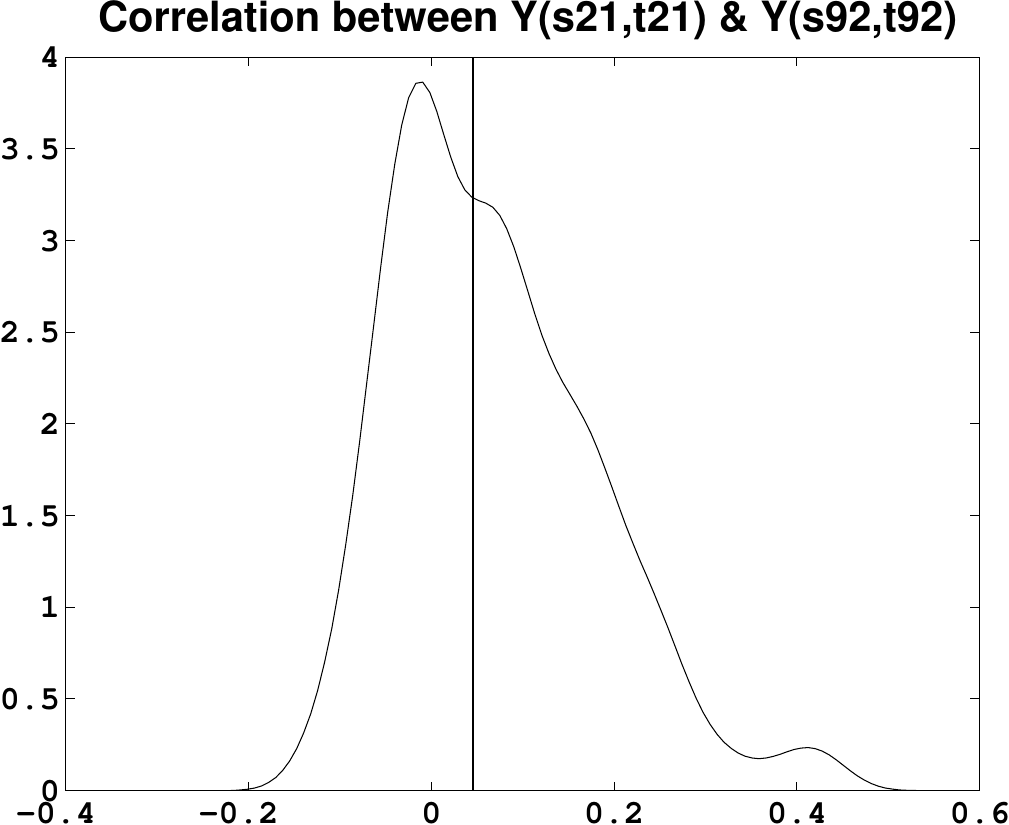}
\\[5mm]
\includegraphics[height=1.5in,width=1.75in]{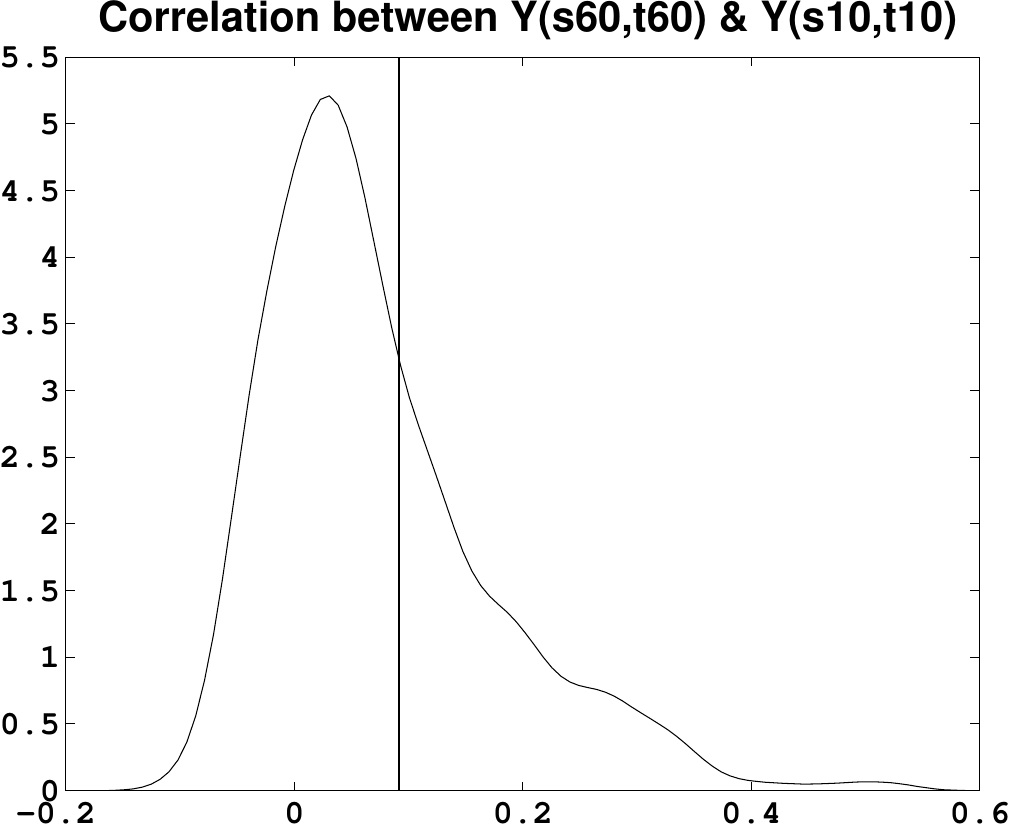}
\includegraphics[height=1.5in,width=1.75in]{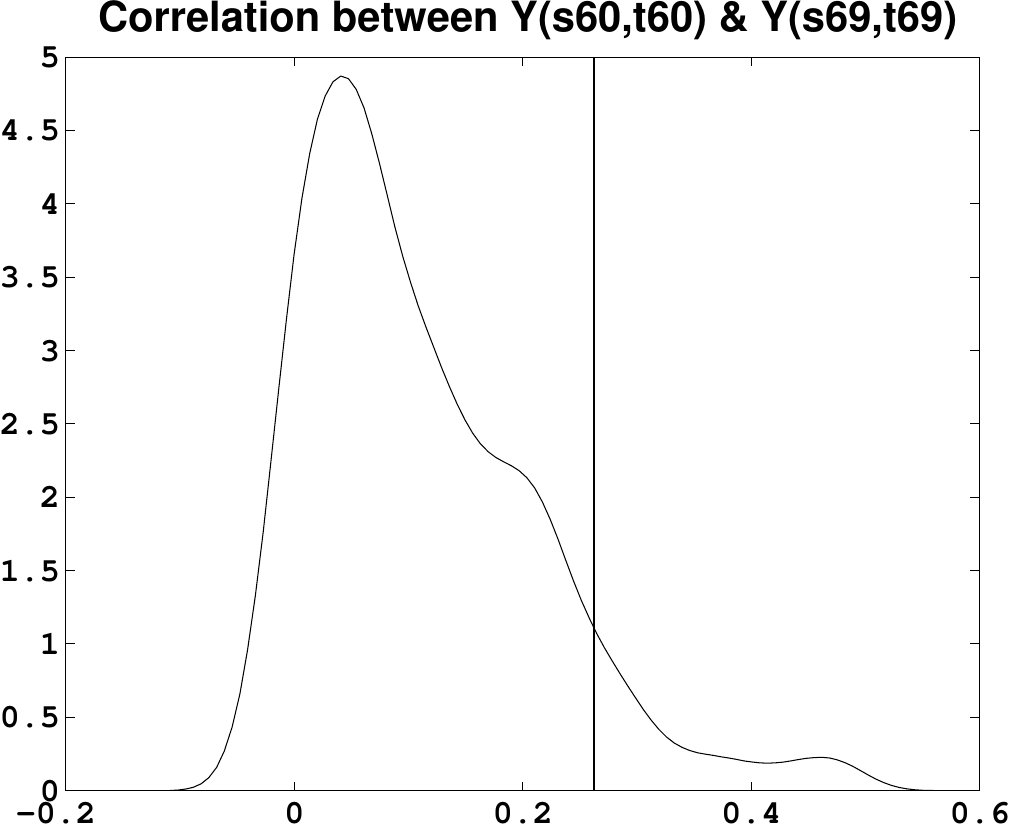}
\includegraphics[height=1.5in,width=1.75in]{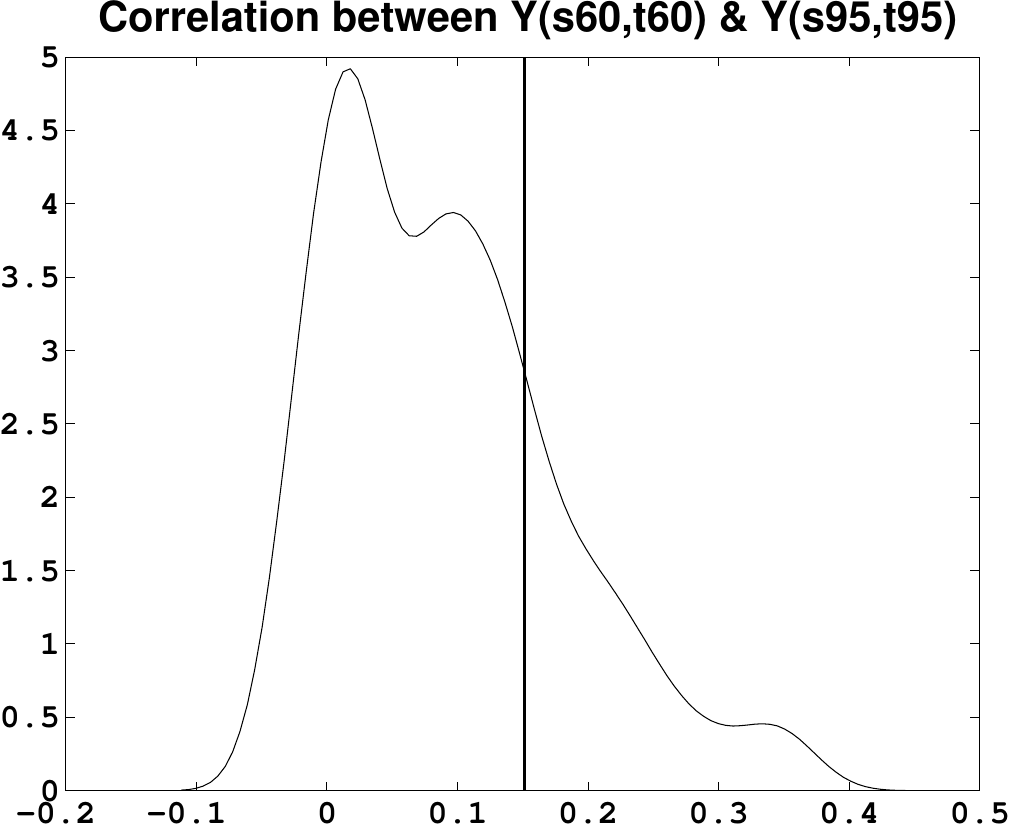}
\caption{{\bf Simulation study}: Posterior densities of the correlations for the  12 different pairs of spatio-temporal points of our model; the vertical lines 
indicate the true correlations.}
\label{fig:correlation}
\end{figure}
\subsection{Comparative study with respect to FR's approach}
\label{subsec:comparison_fuentes}
We now compare the performance of our model with FR. 
For the purpose of comparison, we extended the exclusively spatial model of FR 
to space-time model, and apply the same to our simulated data. 

We asses the predictive power of their model with the leave-one-out cross validation method. 
All the 95 cases were included in the 95\% highest posterior densities of the corresponding leave-one-out
posterior predictive densities.
Figure \ref{fig:posterior_predictive2_fuentes} displays the posterior predictive densities along with the true values obtained 
by employing FR's 
model for the same six locations that were investigated in our model. 
If we consider the CPO measure defined by $CPO_{i}= \pi(y_{i}^{obs}|y_{-i})$ (Conditional Predictive Ordinate)
(\ctn{Pettit1}, \ctn{Geisser}), except for a few locations 
where the CPO measure for our model is slightly smaller than the model proposed by FR, 
our model performance is 
significantly better  
for most of the locations. Moreover, variabilities of the leave-one-out posterior predictive densities associated with the model of FR 
are substantially larger for all the locations.

\subsubsection{Correlation Analysis}

We calculate the posterior densities of the correlation for the same 12 pairs of space-time points. that were investigated in our model. The main features of the correlation analysis are the following :
\begin{itemize}
 \item 
	 The posterior densities, which are highly multimodal in nature, are in keeping with the trace plots of the correlations (not shown), which clearly indicate
convergence to multimodal distributions.
\item
Analogous to the CPO measure described above, here we evaluate the correlation based performance of the models in terms of the
densities of the true correlations 
under the corresponding posterior distributions. From  Figures \ref{fig:correlation} and \ref{fig:correlation_Fuentes},
except for a few space-time pairs, our model significantly outperforms that of FR 
for all the remaining space-time pairs.

\item 
Moreover, when the true correlations are close to zero, for all the space-time pairs, the densities of the true correlations under the corresponding 
posterior distributions are significantly higher than that of FR. 

\item
The above facts strengthen our claim that, compared to other models, our correlation structure is sufficiently rich for capturing the actual correlations, specifically when the true correlation 
is close to zero for nonstationary models.

\end{itemize}

\begin{figure}
\centering
\includegraphics[trim={0 0 0 0},clip, totalheight=0.25\textheight]{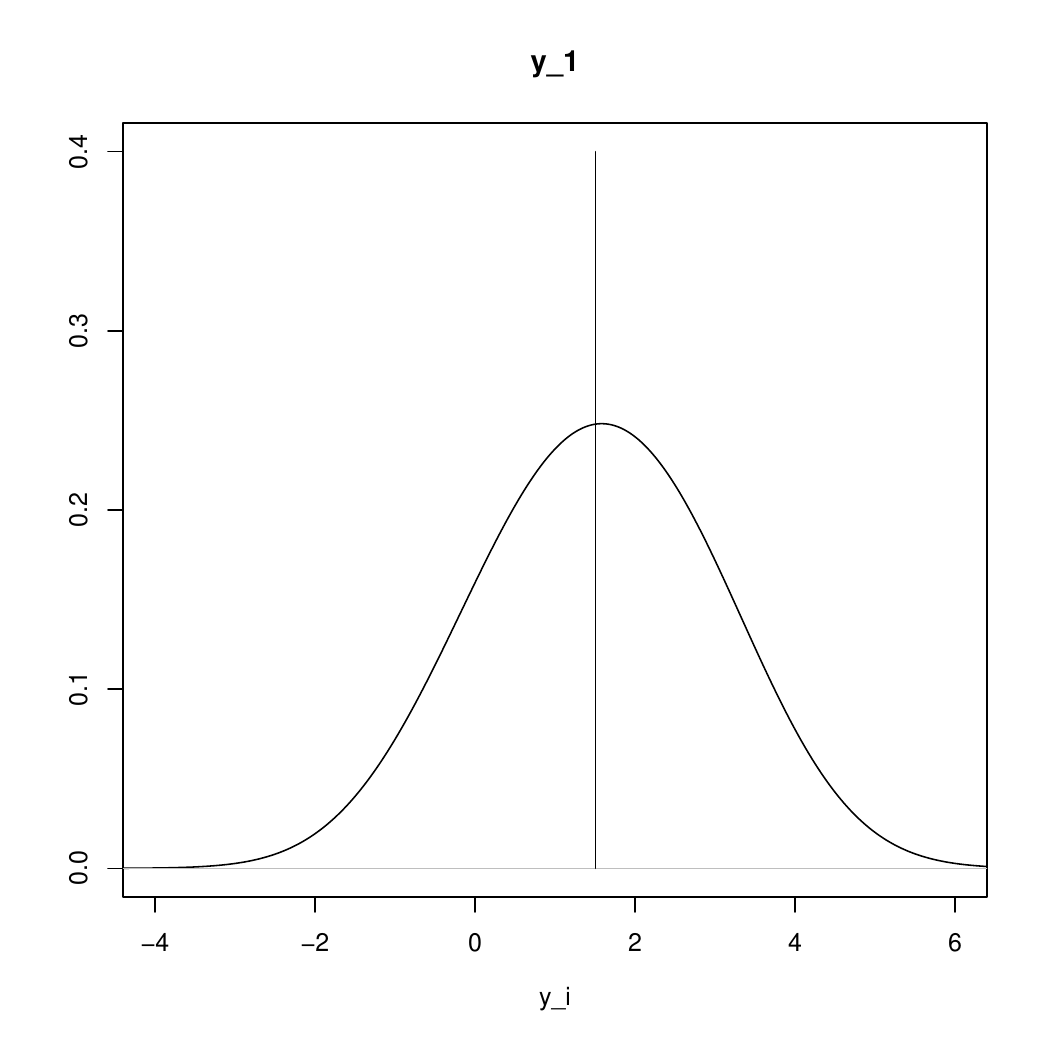}
\includegraphics[trim={0 0 0 0},clip, totalheight=0.25\textheight]{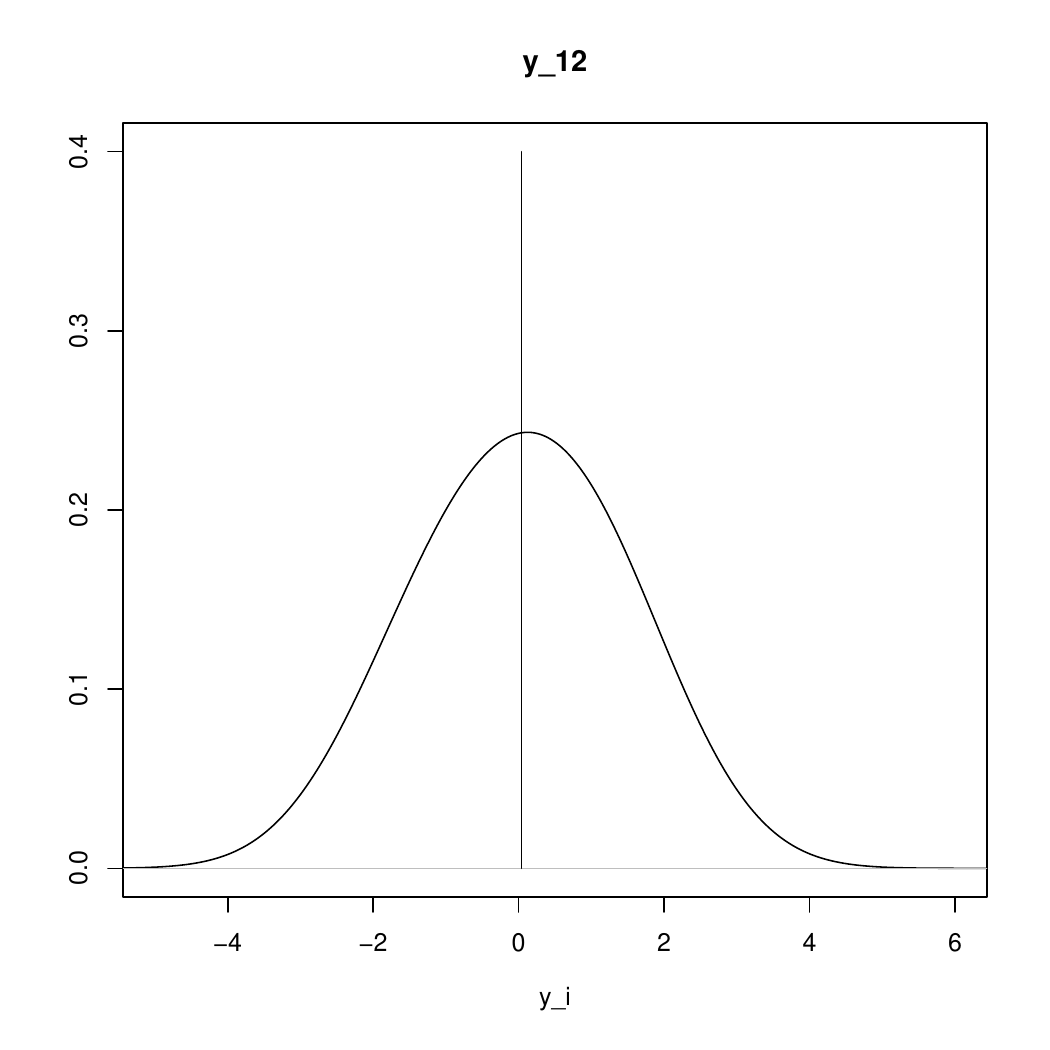}
\includegraphics[trim={0 0 0 0},clip, totalheight=0.25\textheight]{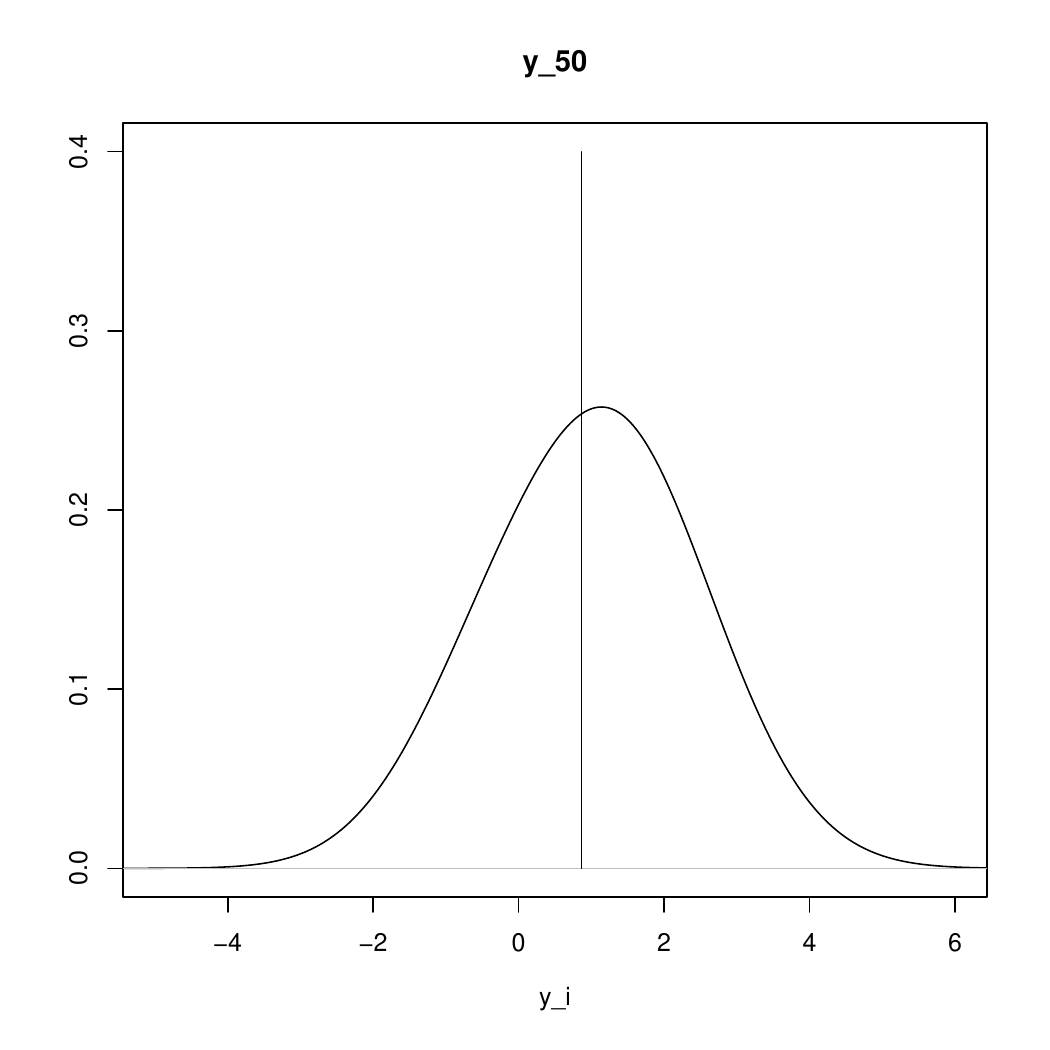}
\includegraphics[trim={0 0 0 0},clip, totalheight=0.25\textheight]{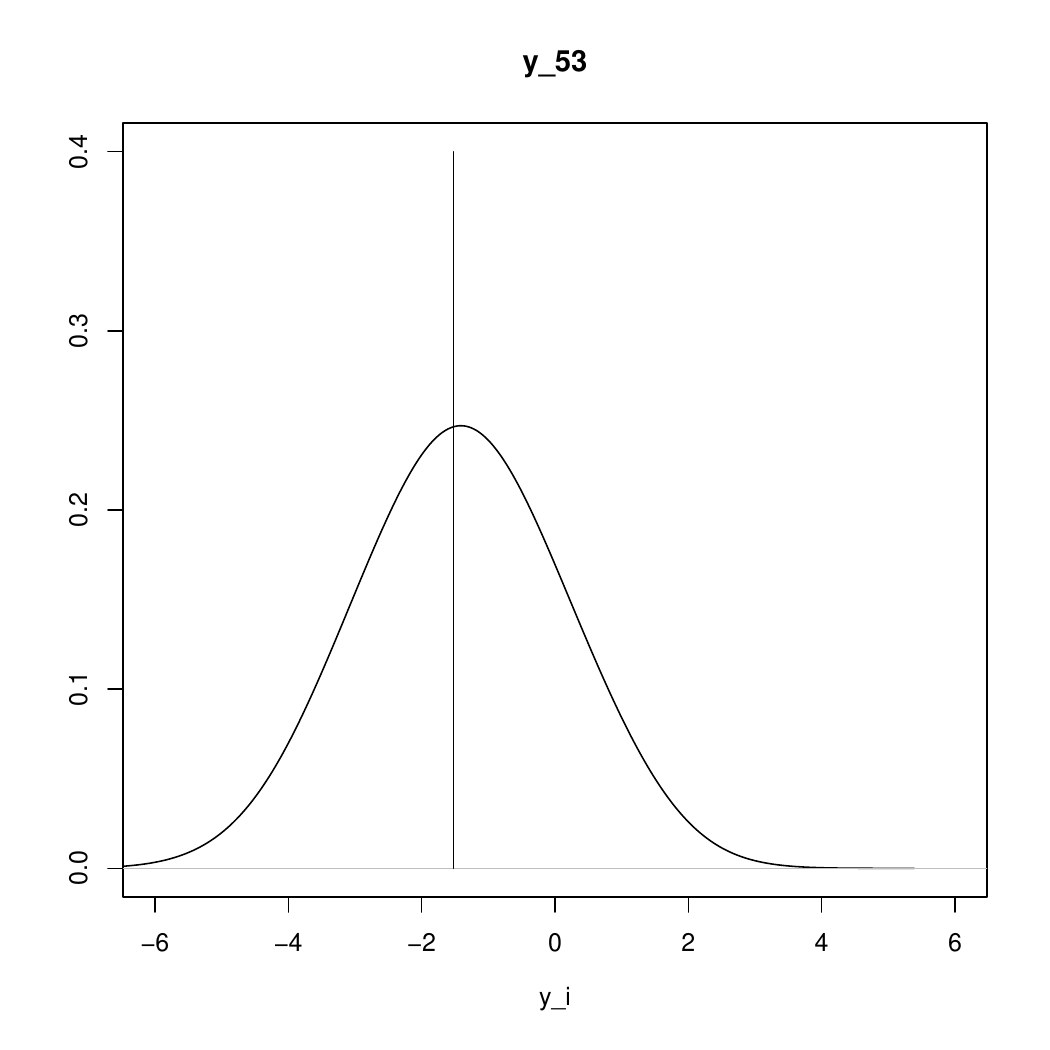}
\includegraphics[trim={0 0 0 0},clip, totalheight=0.25\textheight]{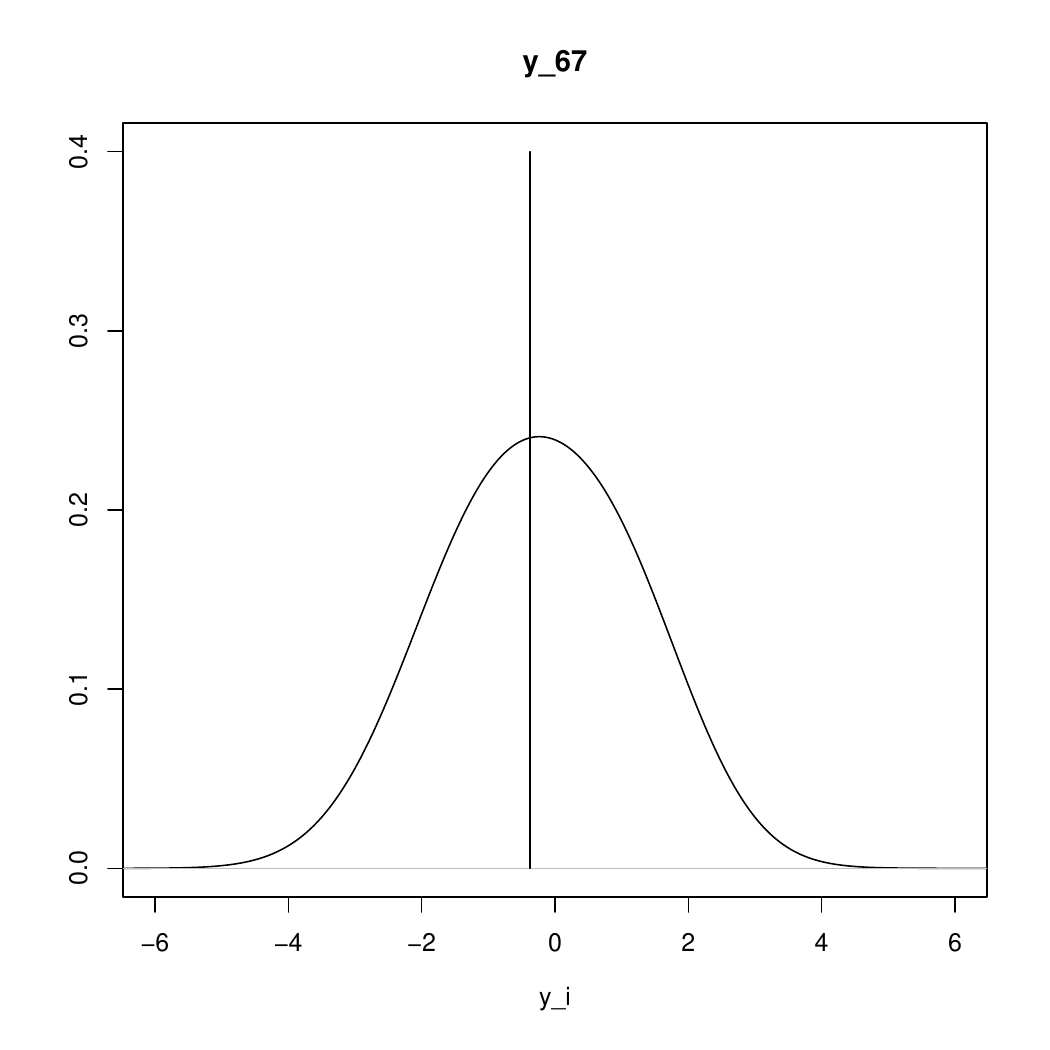}
\includegraphics[trim={0 0 0 0},clip, totalheight=0.25\textheight]{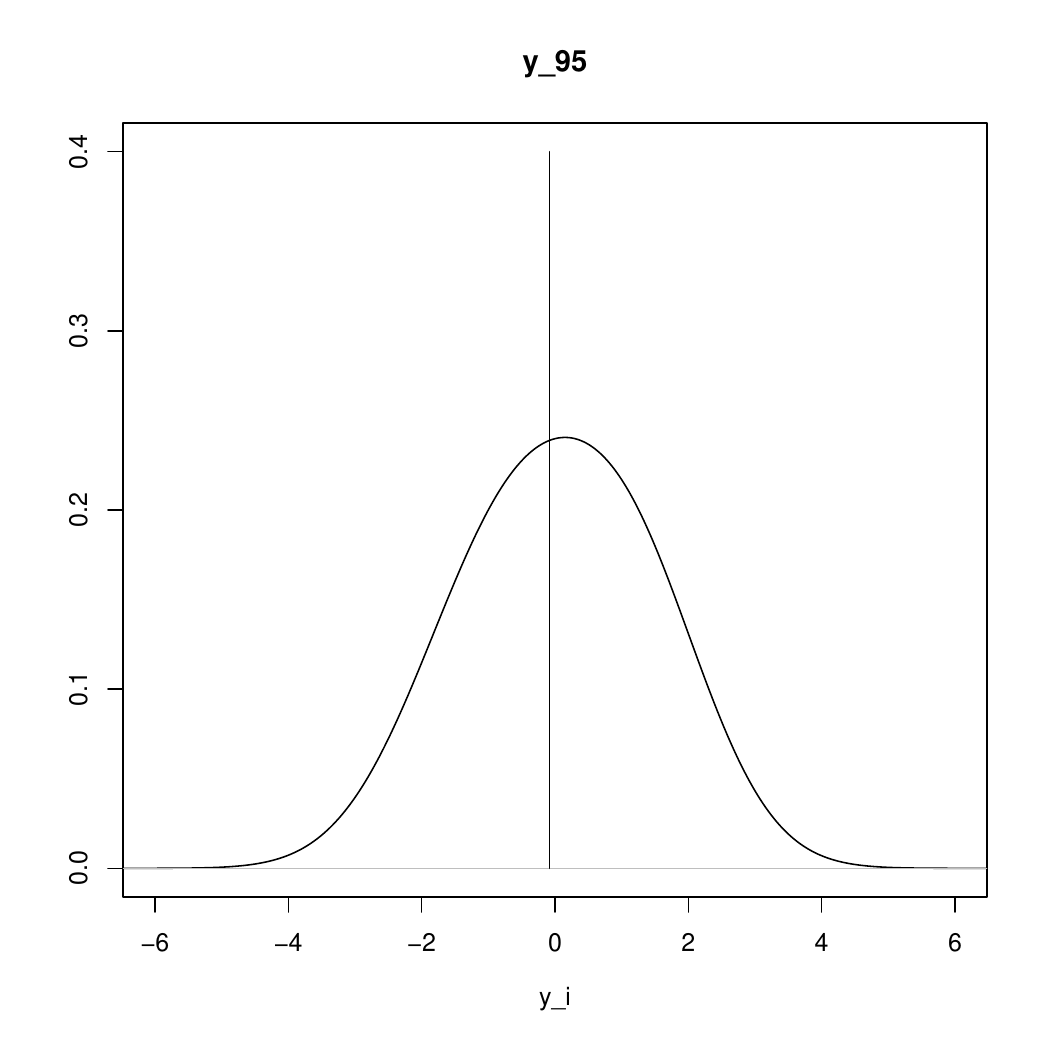}
\caption{{\bf Simulation study}: Posterior predictive densities of $Y(\bs,t)$ for the 6 different location-time pairs of FR -- 
the corresponding true values
are denoted by the vertical line.}
\label{fig:posterior_predictive2_fuentes}
\end{figure}

\begin{figure}
\centering
\includegraphics[height=1.5in,width=1.75in]{figures/correlation_plots_comparison/correlation_1_2_fuentes-crop.pdf}
\includegraphics[height=1.5in,width=1.75in]{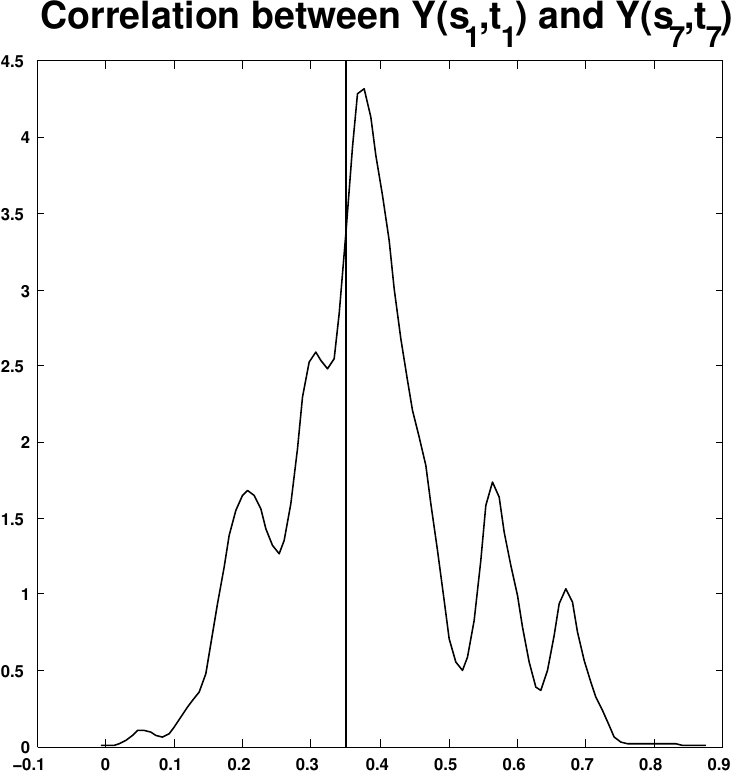}
\includegraphics[height=1.5in,width=1.75in]{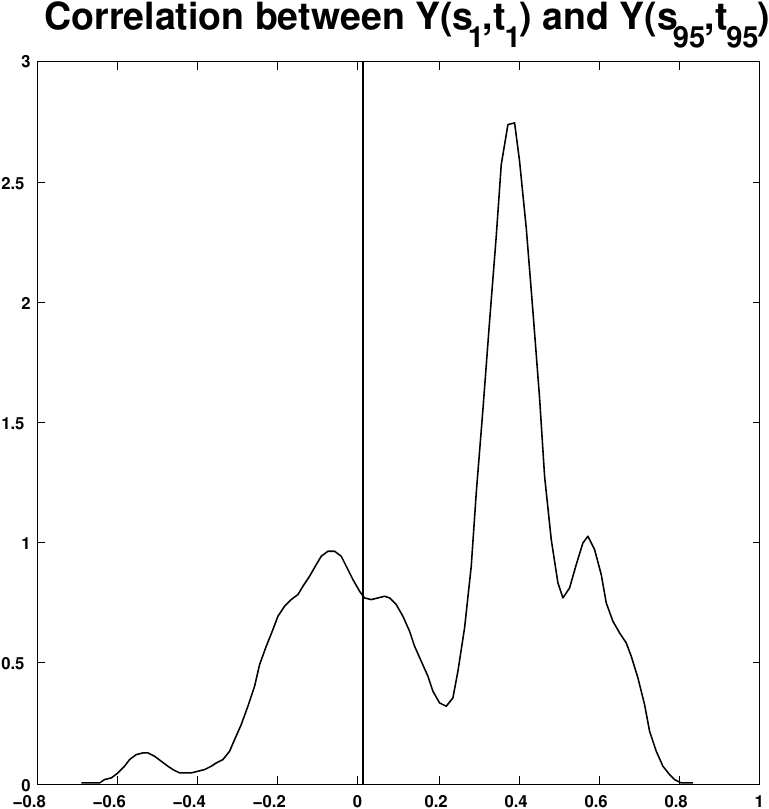}
\\[5mm]
\includegraphics[height=1.5in,width=1.75in]{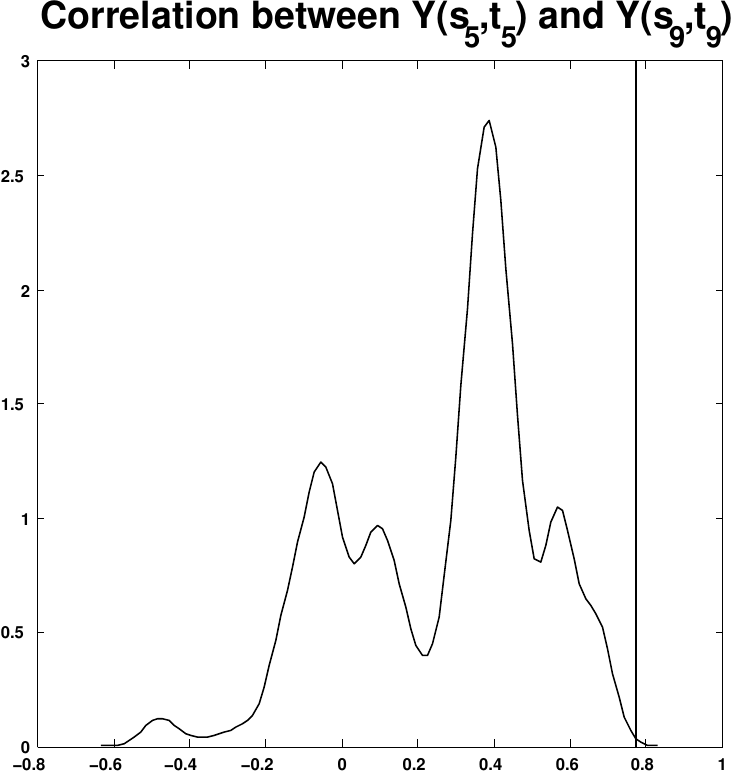}
\includegraphics[height=1.5in,width=1.75in]{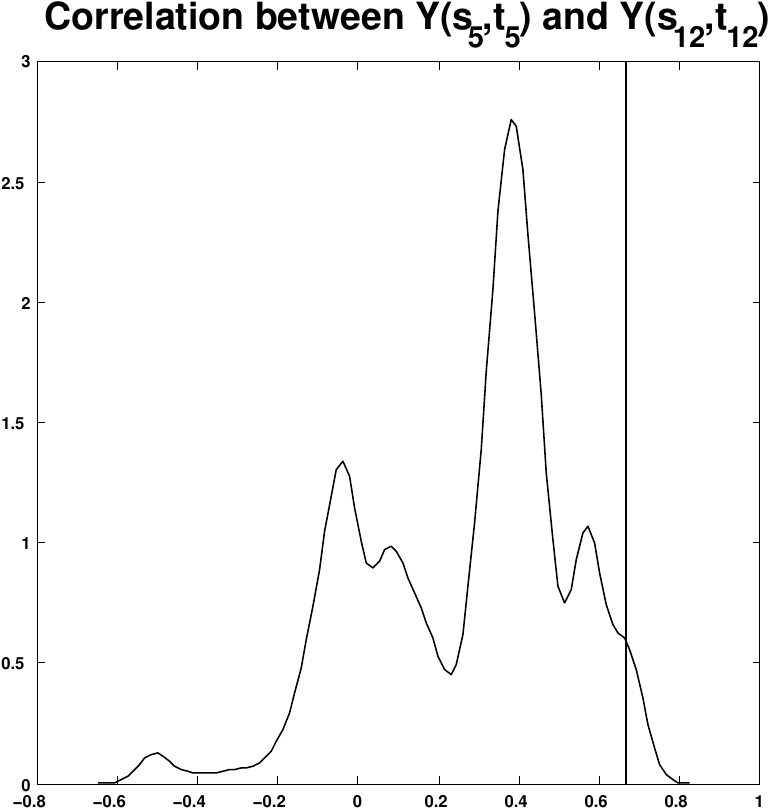}
\includegraphics[height=1.5in,width=1.75in]{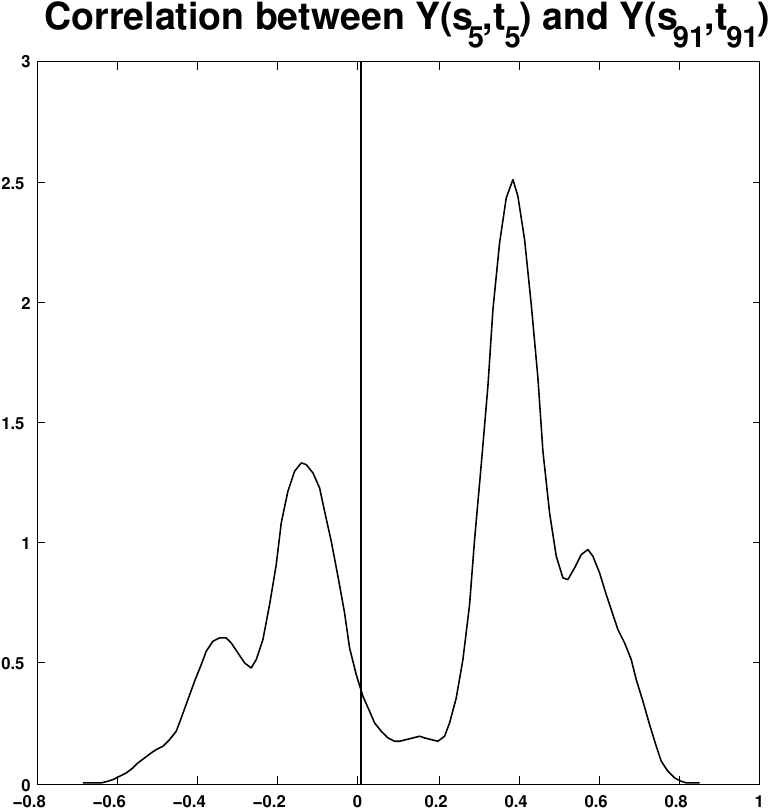}
\\[5mm]
\includegraphics[height=1.5in,width=1.75in]{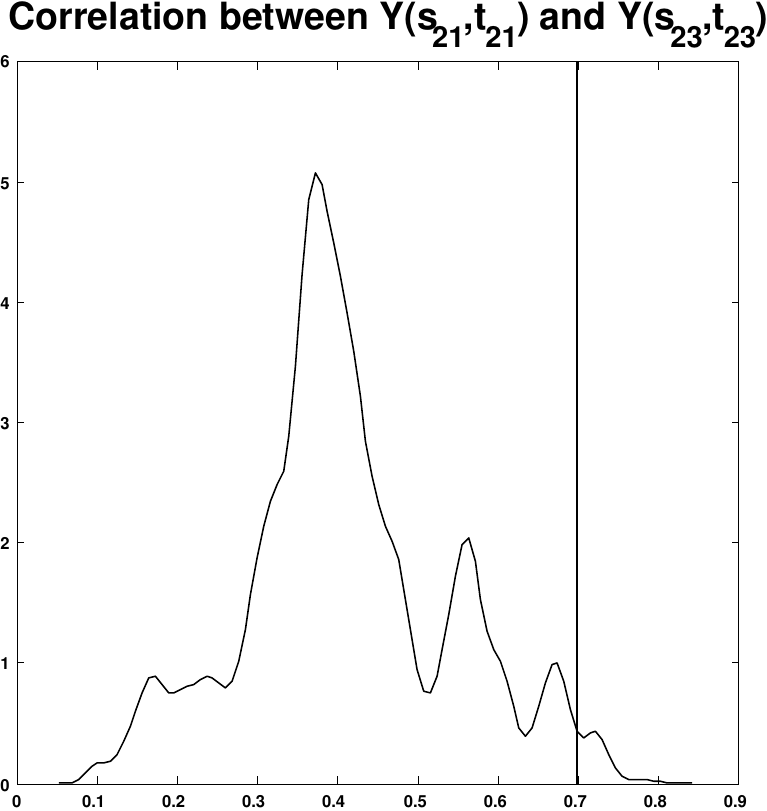}
\includegraphics[height=1.5in,width=1.75in]{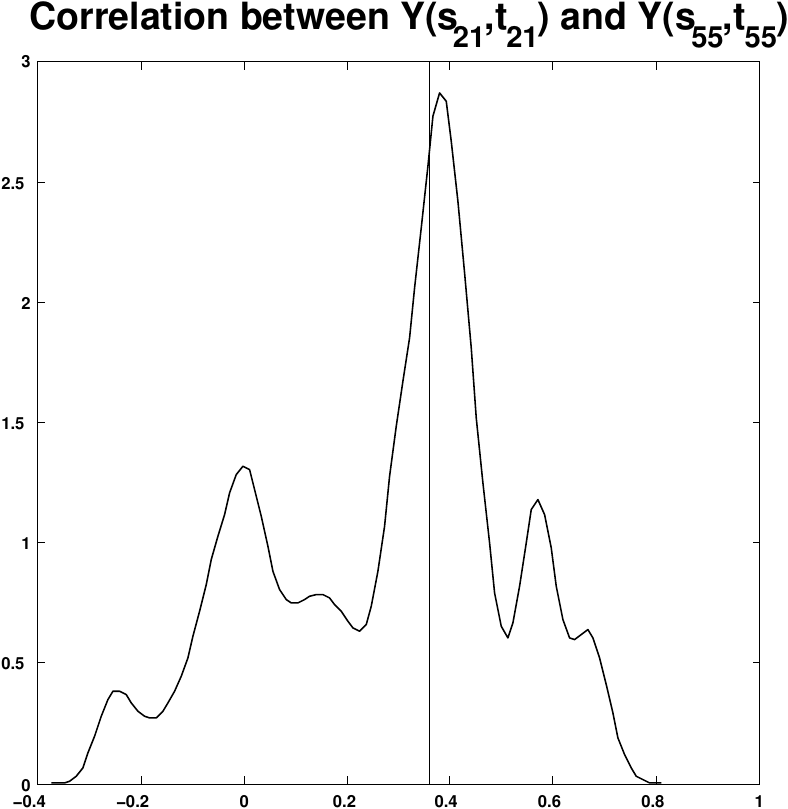}
\includegraphics[height=1.5in,width=1.75in]{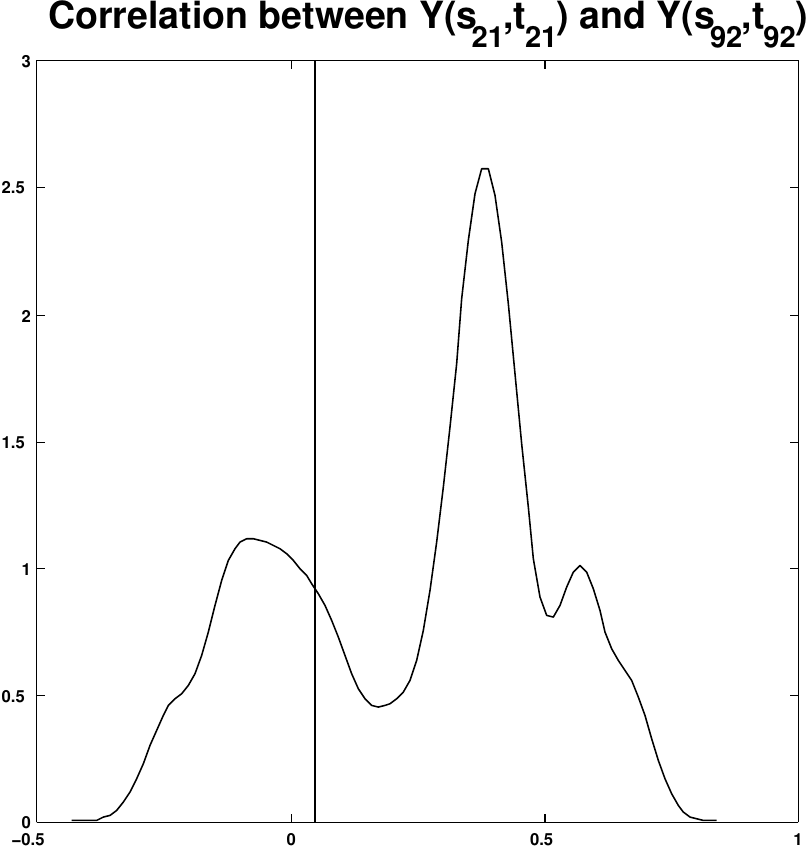}
\\[5mm]
\includegraphics[height=1.5in,width=1.75in]{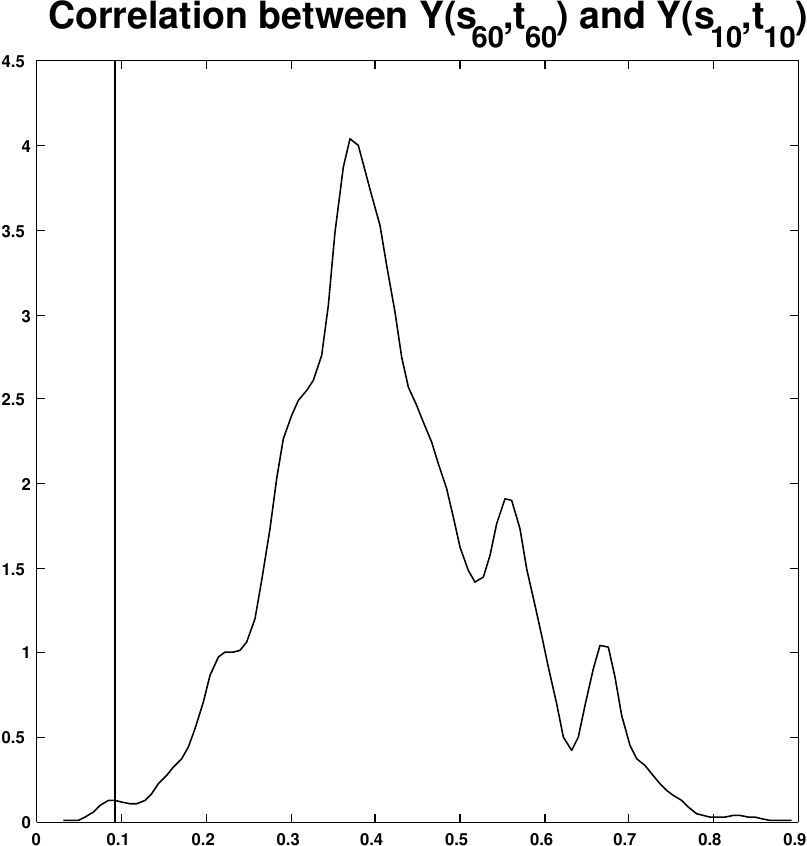}
\includegraphics[height=1.5in,width=1.75in]{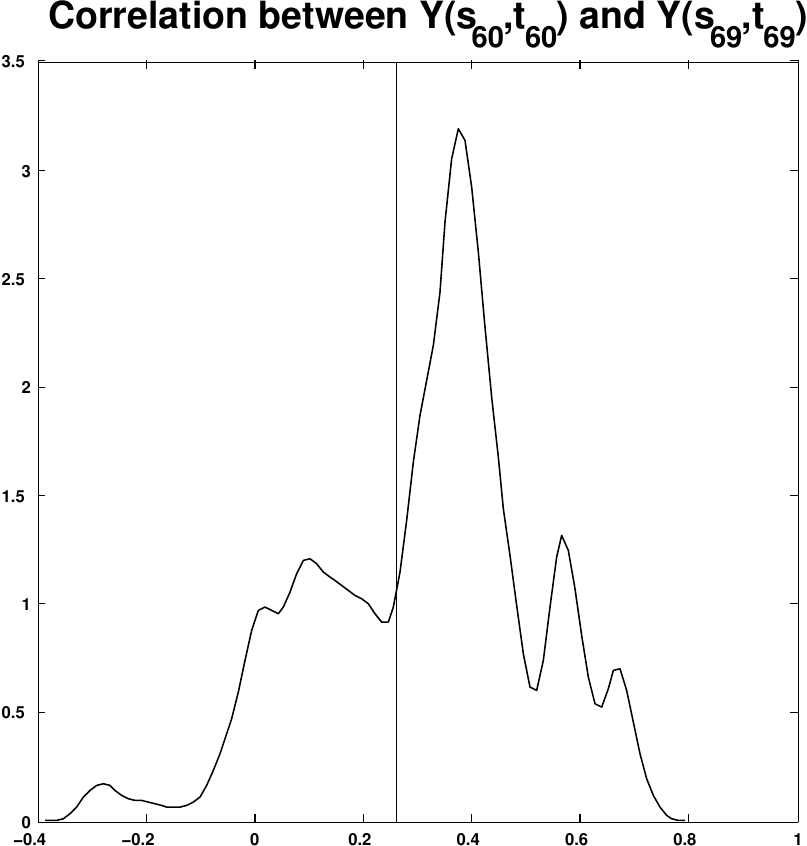}
\includegraphics[height=1.5in,width=1.75in]{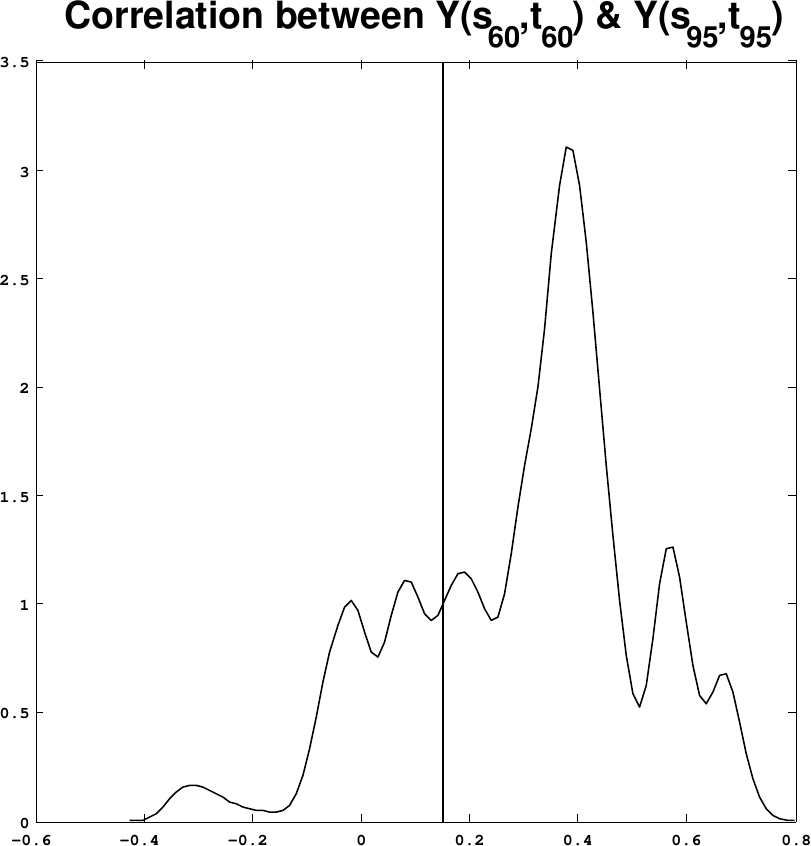}
\caption{{\bf Simulation study}: Posterior densities of the correlations for the 12 different pairs of spatio-temporal points of FR; 
the vertical lines indicate the true correlations.}
\label{fig:correlation_Fuentes}
\end{figure}


\section{ Real data analysis}
\label{sec:realdata_analyses}
\subsection{Spatial data}
\label{sec:realdata}
According to the Clean Air Act certain air quality is to be maintained to protect the public health, 
and to maintain proper survival environment of animals and vegetations. 
As a measure of the quality of air, the Clean Air Act set standard limits for important 
air pollutants such as ozone. For our real data analysis we use the ozone metric called W126 metric. The impact of ozone exposure on trees, plants and ecosystems is often assessed using a seasonal index known as a ``W126 index", which is the annual maximum of consecutive three month running total of weighted sum of hourly concentrations observed between 8AM and 8PM on each day during the high ozone season of April through October. A fundamental principle behind W126 metric is that higher hourly average ozone concentrations should be weighted more than middle and lower values when assessing human and environmental effects. 
The cumulative W126 exposure index uses a sigmoidally weighted function.  The W126 index is a cumulative exposure index and not an ``average" value. As indicated above, it is a biology based index, which is supported by research results (i.e., under both experimental and ambient conditions) that show that the higher hourly average ozone concentrations should be weighted greater than the mid- and lower-level values. The US EPA reviewed the National Ambient Air Quality Standards (NAAQS) for ozone in 2015, and determined that a 3-month W126 index level of 17 ppm-hrs is sufficient to protect the public welfare based on the latest science on effects of ozone on vegetation (US Federal Register, 2015). Also, we have information on 
Community Multiscale Air Quality indices (CMAQ), which is highly correlated with ozone level, so that 
we can use CMAQ as a covariate in our model.

\subsubsection{Calculating the W126 metric}

Let $Q_l(\bs,t)$ denote the observed ozone concentration level in parts per million
(ppm) units at location $\bs$ at hour $l$ on day t, for $t=1, \ldots, T$ and $l=1, \ldots, 12$,
where $T = 214$ days between April 1 and October 31 in a given year. The hours
are the 12 day light hours between 8AM and 7PM. The W126 metric for site $\bs$
is calculated as follows.
\\[2mm] 
The weighted hourly metric is calculated using the transformation:
\[
U_l(\bs,t) = Q_l(\bs,t) \times\left(\frac{1}{1 + 4403\times \exp(-126 \times Q_l(\bs,t))}\right).
\]
This logistic transformation truncates the values smaller than 0.05ppm to zero, 
but does not alter the magnitude of values larger than 0.10ppm.

The daily index from the 12-hourly weighted values in each day is obtained
as
\[
Z(\bs,t)=\sum_{l=1}^{12}U_l(\bs,t).
\]

The monthly index is calculated from the daily indices by summing and
then adjusting for the number of days in the month as follows:
\[M_j(\bs) =\sum_{t\in \text{month}j}Z(\bs,t) ,j=1,\ldots, 7,
\]
where the summation is over all the days $l$ that fall within the calender month
$j$.

The three-month running totals are centered at the last month and are
obtained as:
\[
\bar M_j(\bs)=\sum_{k=j-2}^{j}M_k(s),j = 3,\ldots, 7. 
\]

Finally, the annual W126 index value is calculated by:
\[
Y(\bs)=\max_{j=3}^{7}\bar M_j(\bs).
\]

The secondary ozone standard is met at a site $\bs$ at a given year when the true
value of $Y(\bs)$ is less than 21 ppm-hours.

Corresponding to each observed ozone concentration $Q_l(\bs, t)$ we have a CMAQ
model output $v_l(A, t)$, where the site $\bs$ is contained in the unique grid cell $A$.
Using the output $v_l(A, t)$ and the above details daily and annual indices of CMAQ values 
namely $X(A,t)$ and $X(A)$ are constructed.

We have data on annual indices of ozone values $(W126)$ $Y(\bs)$, and corresponding CMAQ $X(A)$ 
values for 76 locations in the US. Now we fit our model to this real data set. 
Here we model the data on the log scale; we also use the log transformation of the CMAQ values.  
In other words, we consider
\[
\log(Y(\bs))=\alpha_{0} +\alpha_{1}\log (X(A_{i}))+f(\bs_{i})+\epsilon_{i},~ i=1,\ldots 76,
\]
where $\alpha_{0}$ and $\alpha_{1}$ are regression coefficients, $f(\bs_{i})$ is an annual level spatial random
effect at location $\bs_i$ and $\epsilon_{i}$ is an independent nugget effect with variance $\sigma^2$. 
Here $f(\bs_{i})$ is our proposed spatial
model based on kernel convolution with ODPP. 

It is worth mentioning that we had initially considered a stationary kernel for convolution, but obtained poor fit.
This possibly suggested nonstationary process as an appropriate model, but until recently, we were not aware of any formal method
for checking stationarity and nonstationarity in a completely nonparametric setup. Indeed, \ctn{Roy20} proposed a novel recursive Bayesian methodology for 
characterizing stationarity and nonstationarity for general stochastic processes, among various other characterizations, and illustrated their ideas with ample examples
in fields as varied as time series, MCMC convergence diagnosis, spatial and spatio-temporal setups, point processes, as well as (multiple) frequency determination of 
oscillating time series. With their ideas, they also analyse this ozone data to check stationarity. The details of their analyses and the results, 
presented in Section 13.7.1 of their paper, indicate that the ozone data is indeed nonstationary. Further, a simple quantile-quantile plot shows indicates non-normality
of the data.

The above arguments justify our nonparametric model choice and nonstationary kernel used for convolution with ODDP. 
All the prior distributions are the same as mentioned before. 
For the additional parameters $\alpha_{0}$ and $\alpha_1$, we use the vague prior distribution $N(0,10^4)$, 
and for $\sigma$ we use the log-normal prior with mean zero and variance $10^4$.
The TTMCMC trace plots shown in Figure S-10.1 of supplement bear out adequate performance of our model and methodologies.

As before, we assess the predictive power of the model using leave-one-out cross validation. 
For all the locations, the true value of ozone concentration lies within the 95\% credible interval
of the respective cross-validation posterior. This is summarized in the top panel
of Figure \ref{fig:real_surface_plots}, where the middle surface represents the observed data;
the lower and the upper surfaces represent the lower and the upper 95\% credible
regions associated with the respective leave-one-out posterior predictive densities.
The surface in the middle of the bottom panel are the posterior medians, while the lower
and the upper surfaces denote the 95\% credible intervals as before. 
For the convenience of visually comparing the observed data and the posterior medians, we
include Figure \ref{fig:median}, which also contains the 95\% credible intervals. The plots
clearly show that our proposed model is quite adequate for the ozone data.

\begin{figure}
\centering
\includegraphics[height=3in,width=4in]{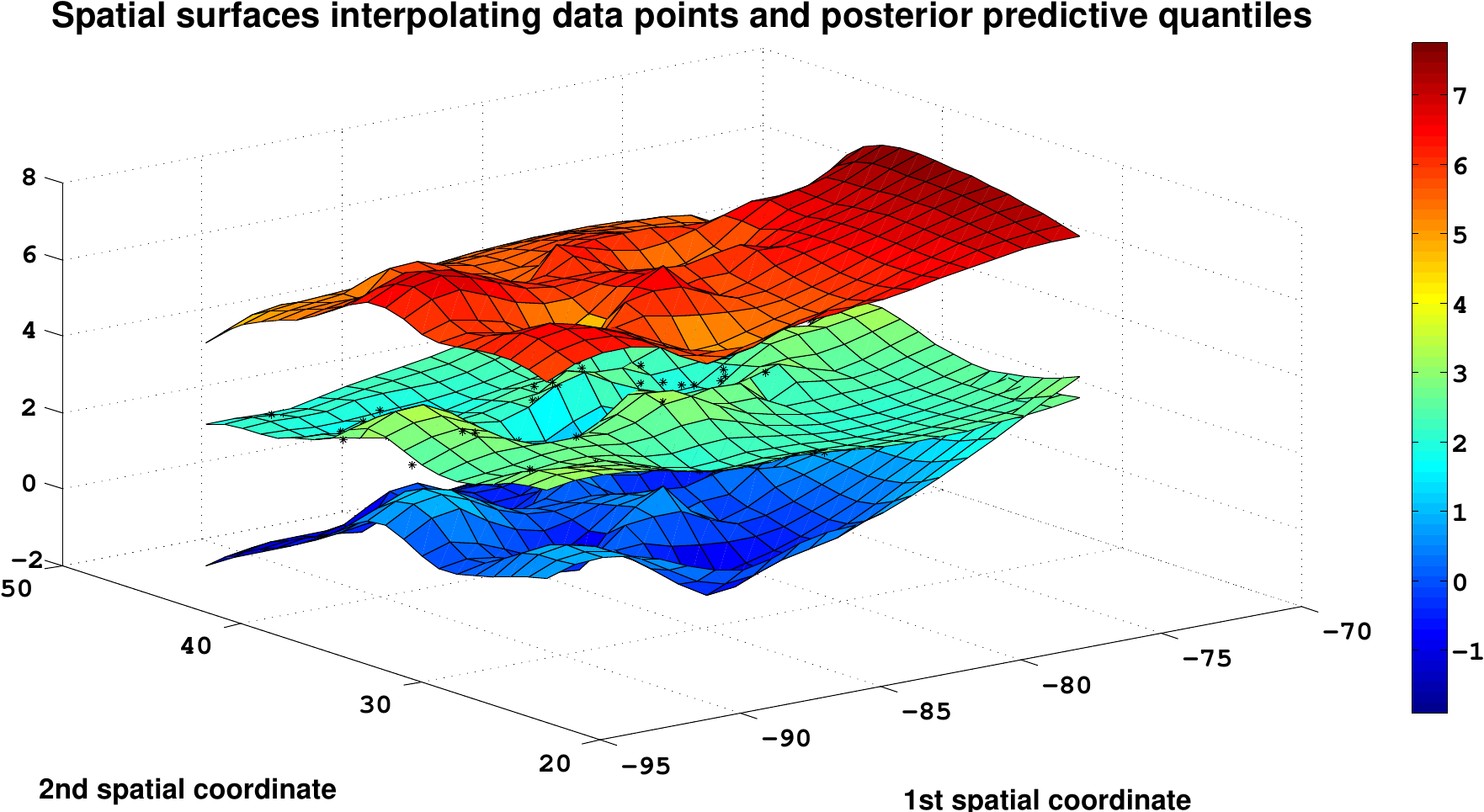}
\\[2cm]
\includegraphics[height=3in,width=4in]{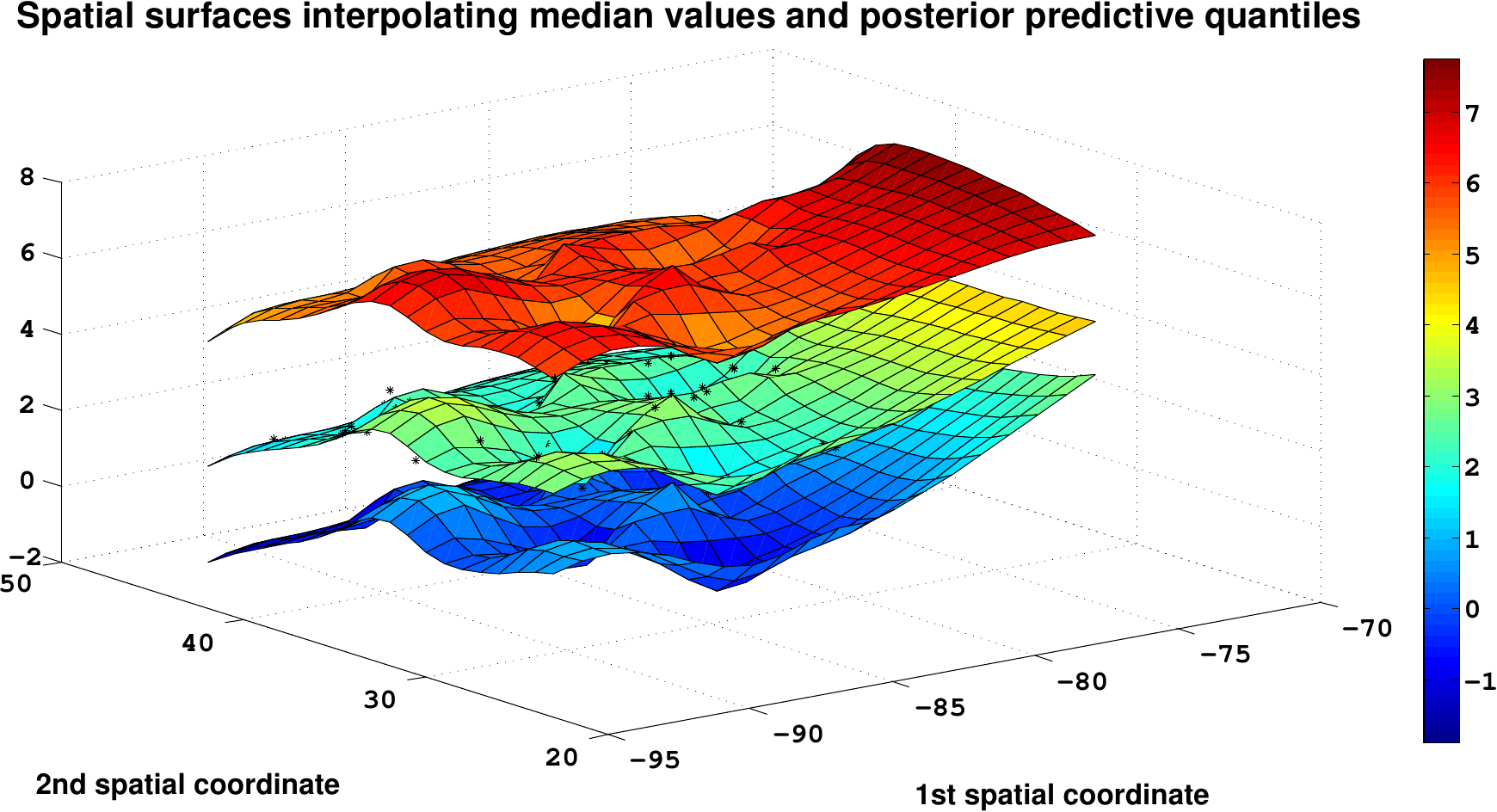}
\caption{{\bf Real spatial data analysis}: The top panel shows the surface plot of ozone concentrations (middle), the lower and the upper
95\% credible intervals associated with the leave-one-out posterior predictive densities, denoted
by the lower and the upper surfaces, respectively. The bottom panel shows the surface
plot of the posterior medians (middle) along with the lower and the upper 95\% credible intervals
associated with the leave-one-out posterior predictive densities (lower and the upper surfaces, respectively).
The observed data points are indicated by `*'.}
\label{fig:real_surface_plots}
\end{figure}

%
\begin{figure}
 \centering
\includegraphics[height=3in,width=6in]{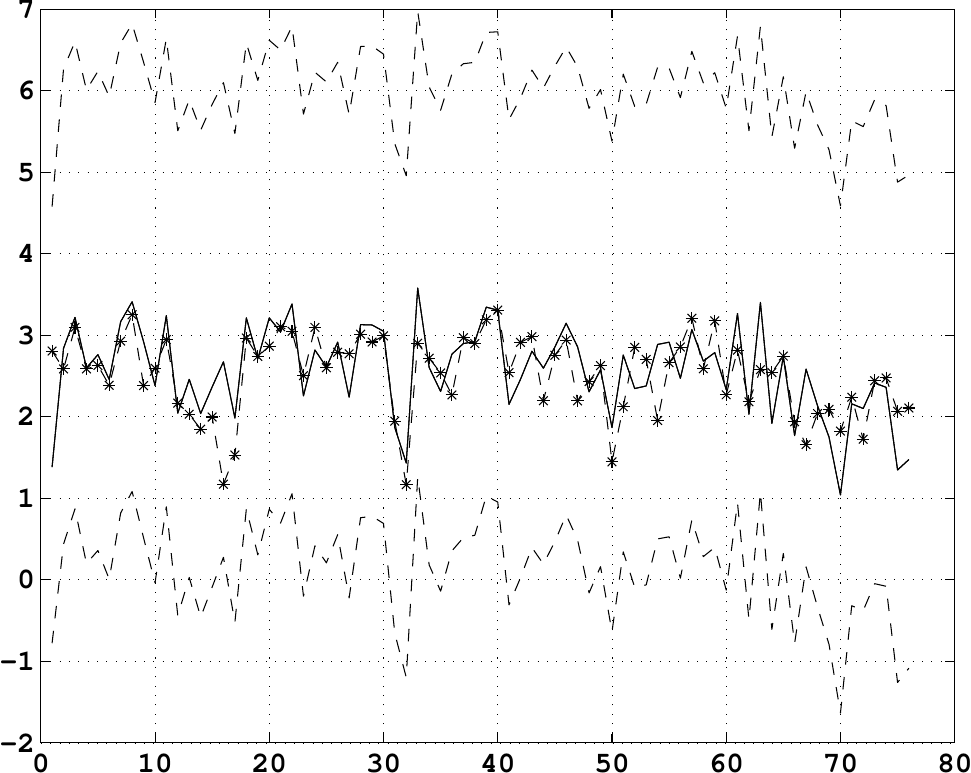}
\caption{{\bf Real spatial data analysis}: Posterior predictive distributions summarised by the median (middle line) and the $95\%$ credible intervals 
as a function of $\bs$. 
The observed data points are denoted by `*'.  }
\label{fig:median}
\end{figure}

Posterior densities of correlations, for 6 pairs of sites, are shown in Figure S-10.2 of supplement. 
All of them seem to give high posterior probability to the approximate range $(0.1,0.3)$.

\subsection{Spatio-temporal data analysis}
\label{sec:realdata2}
`Particulate matter' (PM) is the general term used for a mixture of solid particles and liquid droplets found in the air. 
Airborne PM comes from many different sources. ``Primary" particles are released directly into the 
atmosphere from sources such as cars, trucks, heavy equipment, forest fires, and other burning activities. An extensive body 
of scientific evidence shows that there are adverse effects of this PM particles on health, including cardiovascular problems,
premature death and many more. Ambient air monitoring stations generally measure air concentrations of different ranges of 
particles, but most monitoring station  is for two size ranges: 
$PM_{2.5}$ and $PM_{10}$. 

\subsubsection{Data}

Our data is a part of a big data set 
analysed by \cite{Paciorek09} (Data Source: \url{http://www.stat.berkeley.edu/~paciorek/data/pm/}). 
They specify stationary spatial structures through the use of penalized thin plate splines. 
Assumption of stationarity leads to an important simplification in their model. 
The assumption of stationarity is particularly appropriate for $PM_{2.5}$ values, but there is evidence 
of nonstationarity for $PM_{10}$ values. Indeed, \ctn{Roy20} infer with their novel Bayesian recursive methodology that
the $PM_{10}$ data is strictly, as well as weakly nonstationary (Section 13.7.2 of their paper) and that the $PM_{2.5}$ data is stricly stationary 
(Section 13.7.3 of their paper).

For illustration purpose we fit a nonstationary spatio temporal model for a smaller section of the full data set.
We analyse monthly average values of $PM_{10}$ for the year 1988-2002 (180 time points) at 50 locations. 
There are few locations with fewer sample points. Our model will be appropriate for this kind of data, since we are using
the spatial locations and time points as  arguments of our proposed mean functional. Our data consists of total 3934 
observations for monthly $PM_{10}$ values. To increase the predictive performance of the model, 
we have used available covariate information for different spatial locations and time points.
It is expected that inclusion of covariates may better explain the spatio-temporal heterogeneity. 
The details of the  covariate
selection are discussed in \cite{Yanosky08a}, \cite{Yanosky08b}.
The non-time-varying covariates are as follows: distances to the nearest road within four road size classes; particulate 
point source emissions within 1 and 10 km buffers; the proportion of urban land use of  within 1 km; elevation; 
and block group, tract, and county population density from the 1990 US Census. The time varying covariates are 
wind speed, precipitation and barometric pressure, with hourly values averaged to the month at each station. 

We also analyse the properties of the empirical correlations for increasing spatio-temporal lags with respect to the complete data set consistng of $70572$ observations. 
Figure \ref{fig:lagged_corrs_sma50}, obtained from the raw correlations
after taking moving averages of length $50$ for better visualization, shows that the correlations tend to zero with increasing lags, as realistically expected, 
in spite of the data being nonstationary. Moreover, a simple quantile-quanile plot (not shown for brevity) shows that the data is far from normality. 
These very much support our modeling idea. 
\begin{figure}
\centering
\includegraphics[width=10.5cm,height=9.5cm]{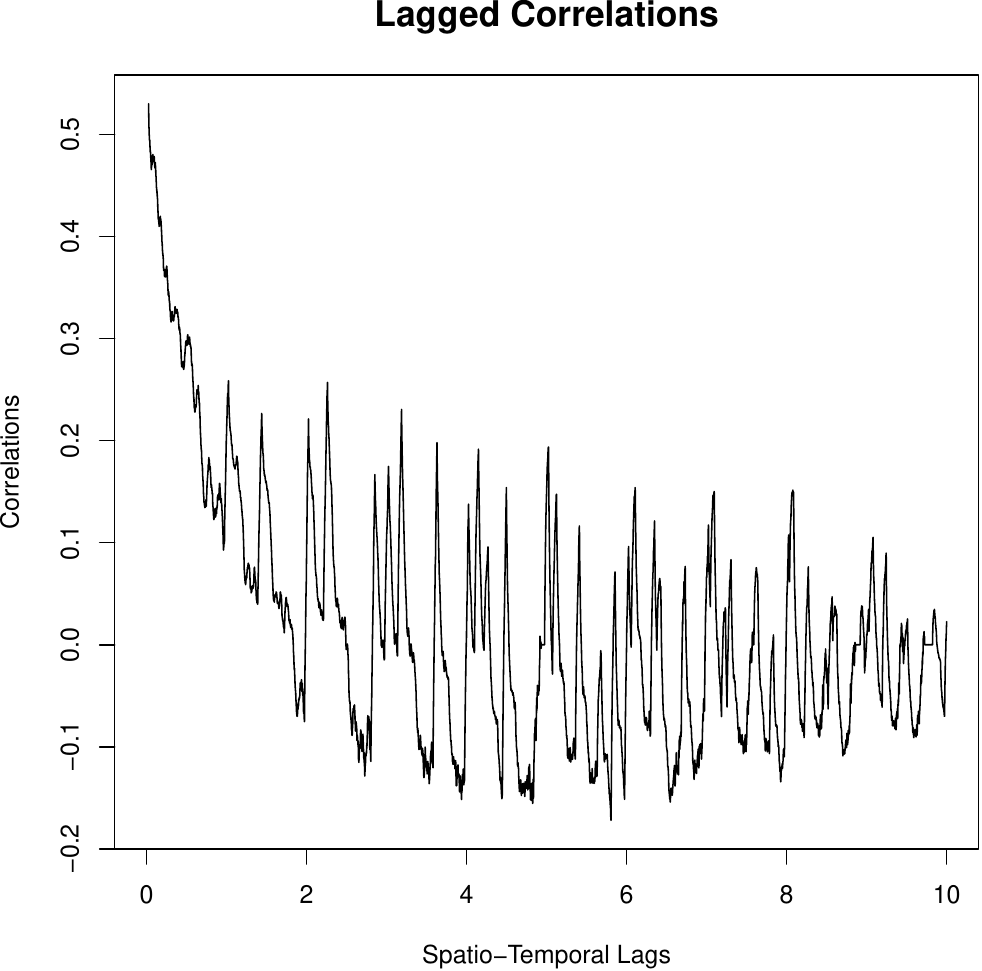}
\caption{Empirical correlations for increasing spatio-temporal lags. The raw correlations are smoothed by taking moving averages of size $50$ for better visualization.} 
\label{fig:lagged_corrs_sma50}
\end{figure}

\subsubsection{Model}
We propose the following model for the real data: 
\[
\log y_{it}=\alpha_{0} +f(\bs_{i},t)+g_1(\tilde\bz_{i})+g_2(\bz_{it})+ \epsilon_{it},~ i=1,\ldots 50,t=1,\ldots, 180,
\]
where $\alpha_{0}$ is intercept term, $f(\bs_{i},t)$ is our proposed spatio-temporal
model based on kernel convolution with ODDP. 
In the above, $g_1$ and $g_2$ are functionis of non-time varying covariates $\tilde\bz_i$ and time varying covariates $\bz_{it}$, respectively.
We assume a Gaussian process prior on $g_1$ such that $\mu_1(\bz)=E\left[g_1(\bz)\right]=\bbeta'\bz$ and $Cov\left(g_1(\bz_i),g_1(\bz_j)\right)=
\exp\left(-\frac{1}{2}\|\bz_i-\bz_j\|\right)$.

We set $g_2$ as a linear function of time varying covariates: $g_2(\bz)={\bgamma}'{\bz}$. The assumption of linearity 
in $g_2$ will simplify our computation to a great extent. Also there is evidence from the previous analysis that using linear terms in places of the 
unknown function led to only negligible decrease in predictive ability. In our model, $\epsilon_{it}$ are independent nugget effects with variance $\sigma^2$. 
 
All the prior distributions are the same as mentioned before.  
For the additional parameters $\alpha_{0},\bbeta,\bgamma$,  we use the vague prior distribution $N(0,10^4)$, 
and for $\sigma$ we use the log-normal prior with mean zero and variance $10^4$.

\subsubsection{Implementation}
Note that here we have total 3934 number of observations. We have to update the number of parameters ranging between between 300 to 400. 
Our TTMCMC based algorithm took 25 minutes to generate 5000 observations following a burn in of 20000. As in the other cases, TTMCMC 
exhibited satisfactory acceptance rate and mixing properties, as evident from the trace plots displayed in S-10.3 of supplement.

\subsubsection{Leave-one-out cross validation}

As before, we assess the predictive power of the model using leave-one-out cross validation. 
For all the spatio-temporal points, the true value of $PM_{10}$ lies within the 95\% credible interval
of the respective cross-validation posterior.
Also we calculate the mean square prediction error ($MSPE$), given by $\frac{\sum(y_{it}-{\hat y_{it}})^2}{n}$, where $\hat y_{it}$ is the median of the
posterior predictive density at the spatial location $(\bs_{i},t)$. In this case, we obtain $MSPE = 0.101$.
Figure \ref{fig:median_spatio_temporal} displays the observed data and posterior medians for at three spatial locations having data for more
than 10 years, which also contains 95\% credible intervals. We have also reported $MSPE$ for these three locations. The values are significantly lower than 
overall $MSPE$. 
It reveals the fact that our model have captured more precise information for the spatial locations, 
having larger number of time points. We also provide a visual representation of the model performance at 50 locations summarised over time points.
In Figure \ref{fig:spatial_surface}, the surface represents the posterior median values, averaged over all month-specific predictions
for 50 spatial locations. From the plots it is clear that our model performs quite satisfactorily for the data. 
Posterior densities of correlations, for 6 pairs of sites, are shown in Figure S-10.4 of the supplement.

\begin{figure}
 \centering
 \subfigure[spatial location:(38.270833	-85.740278), MSPE:0.0344]{ \label{fig:spatial_big_quantile_1}
\includegraphics[height=2in,width=6in]{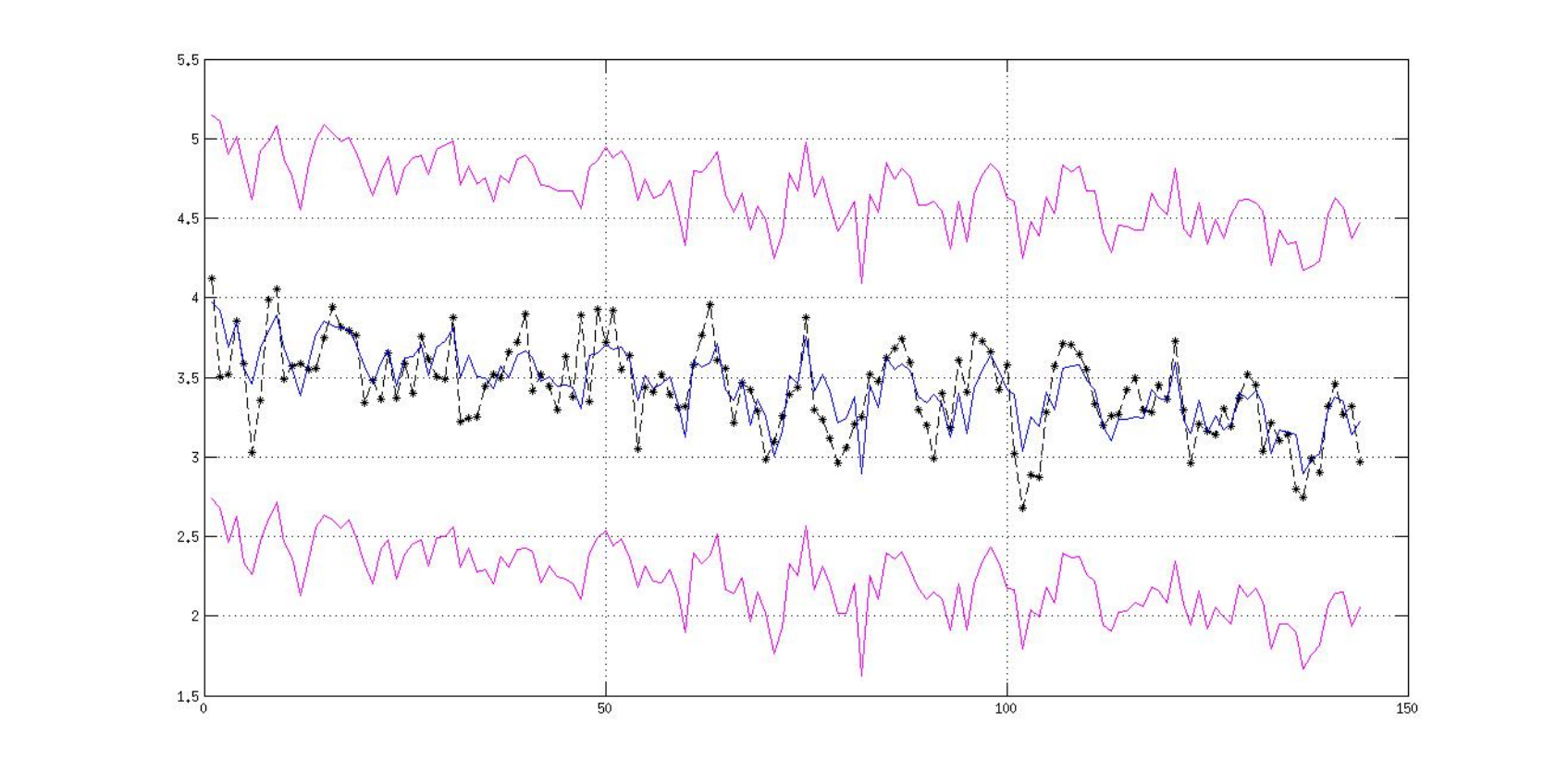}}
\subfigure[spatial location:(41.600278	-87.334722), MSPE:0.051]{ \label{fig:spatial_big_quantile_2}
\includegraphics[height=2in,width=6in]{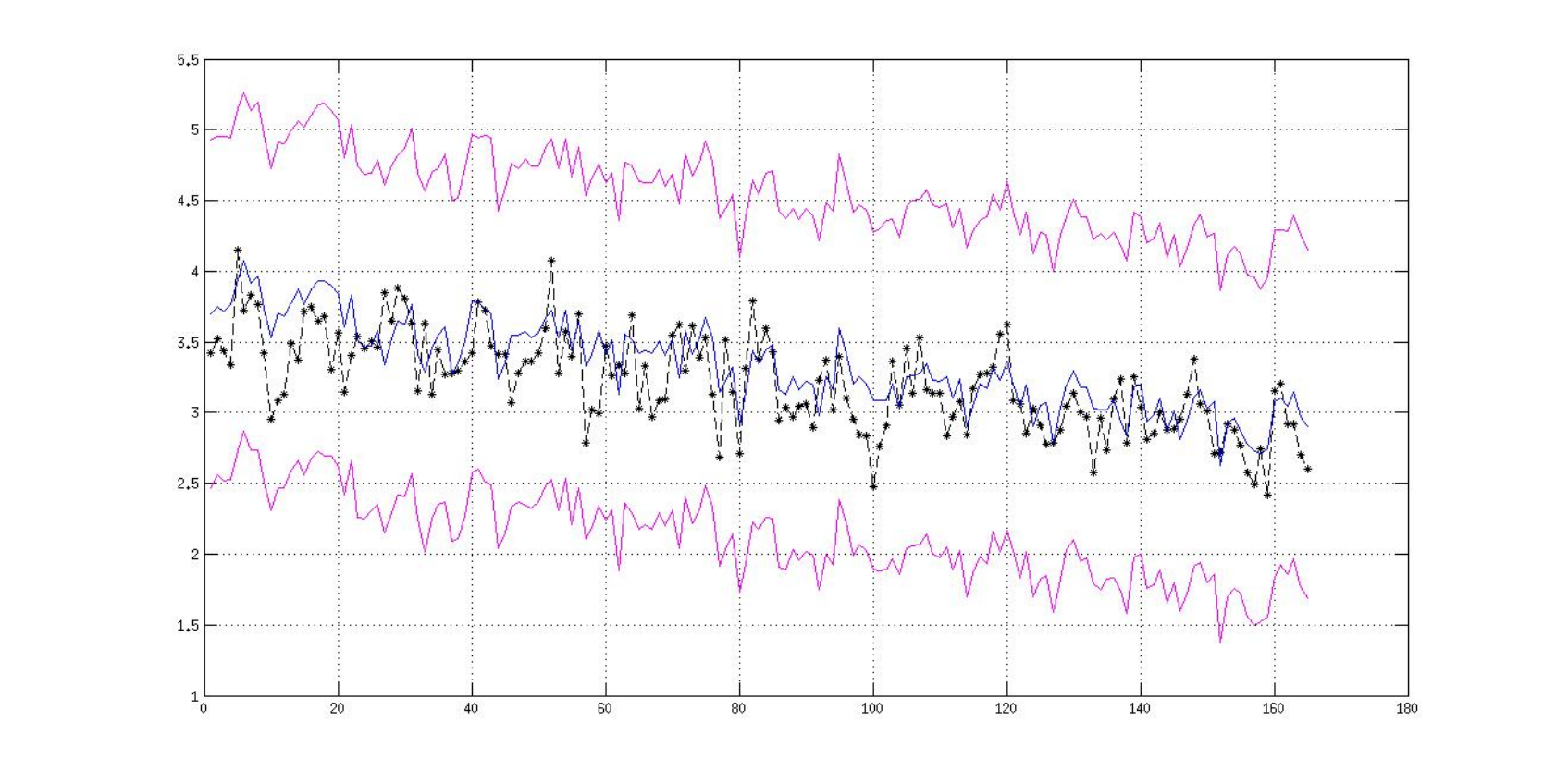}}
\subfigure[spatial location:(39.766111,-86.129167), MSPE:0.048]{ \label{fig:spatial_big_quantile_3}
\includegraphics[height=2in,width=6in]{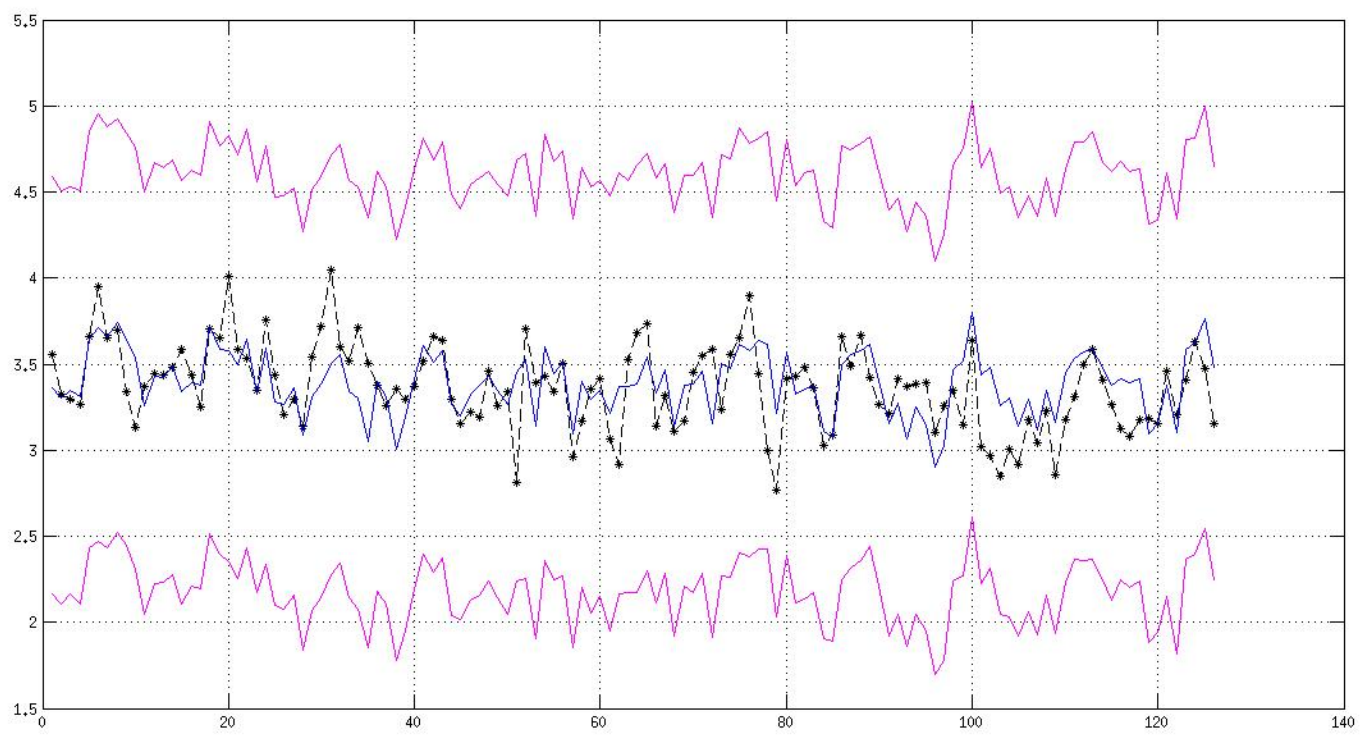}}

\caption{{\bf Real spatio-temporal data analysis:} Posterior predictive distributions summarized by the median 
(middle line) and the $95\%$ credible intervals 
as a function of $t$ for three randomly chosen spatial locations. 
The observed data points are denoted by `*'.  }
\label{fig:median_spatio_temporal}
\end{figure}

\begin{figure}
\centering
\includegraphics[height=6in,width=4in,angle=270]{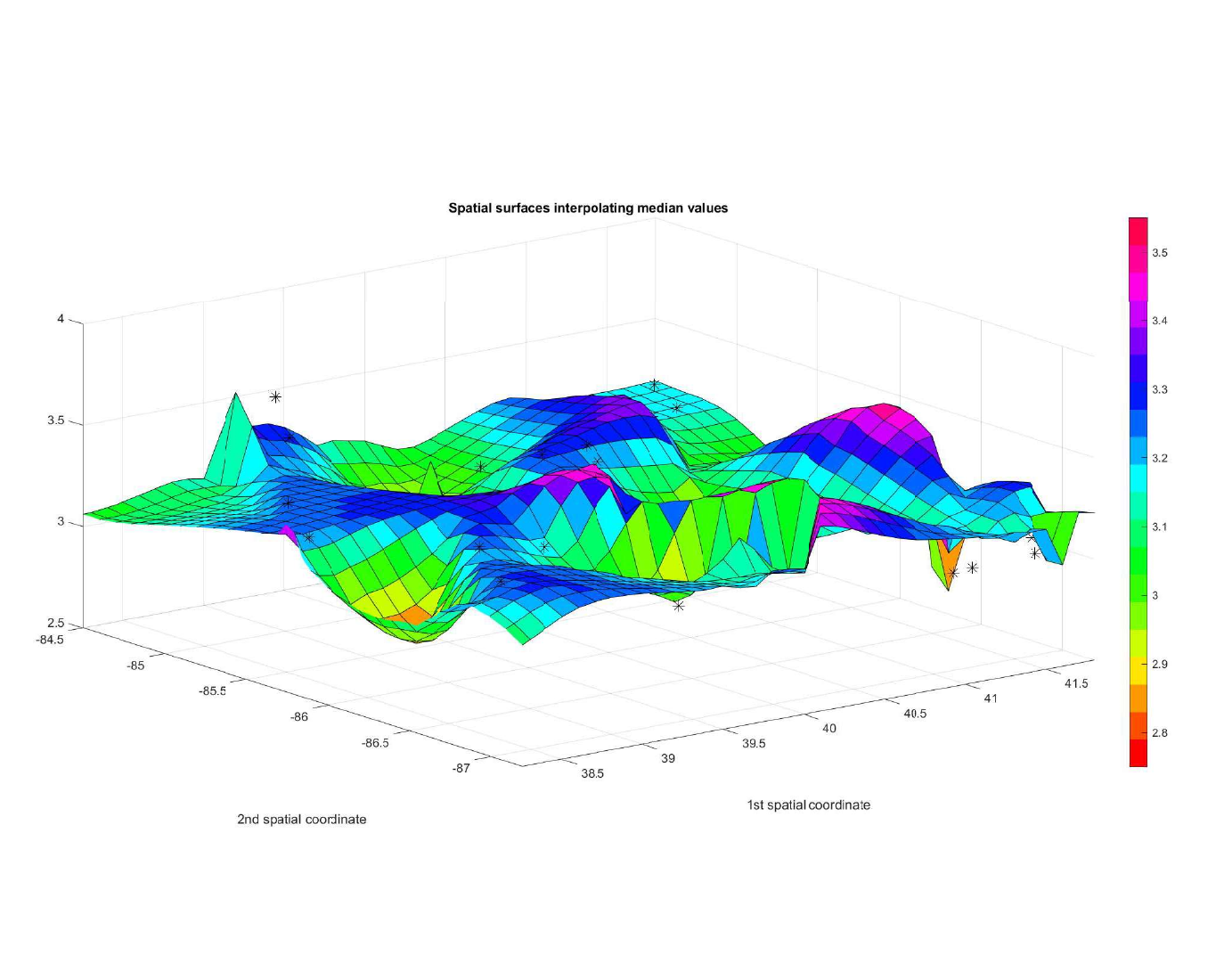}
\caption{{\bf Real spatio-temporal data analysis}: the surface plot of posterior median values at 50 locations, averaged over all month specific predictions 
from 1988-2002. The observed data points are indicated by `*'.}
\label{fig:spatial_surface}
\end{figure}

\section{Summary and conclusion}
\label{sec:conclusion}
In this article, we have developed a non stationary, non-Gaussian spatio-temporal model based 
on kernel convolution of ODDP. Dependence is induced in the weights through similarities in the ordering of the atoms. 
Using this property we could ensure that our model-based correlation between two random data points 
which are widely separated, will be close to zero. We incorporated non-stationarity via appropriate kernels,
which would be convolved with ODDP. Although our proposed model is non stationary and non-separable, 
it includes stationarity and separability as special cases. Moreover, since our model is based on 
kernel convolution, replication is not necessary for inference. If one wishes to achieve different 
degrees of smoothness across space and across time, then that is also allowed by our model framework. 
For example, if we associate the ODDP prior only to the spatial locations, then the process will become smoother 
across time than across space depending on the choice of the kernel.

From the computational point of view, we have developed a fast and efficient TTMCMC based algorithm for 
implementing our variable dimensional spatio-temporal model. Indeed, our model consists of a 
large number of variables, where the number of variables associated with the summands is random. 
Using TTMCMC, we could update all the parameters, as well as number of the parameters simultaneously, 
using simple deterministic transformations of some a one-dimensional random variable. 

We illustrated the performance of our model with a simulation study 
and have compared the performance of our model with the model of FR. 
The comparative study supports our claim that our model
is capable of capturing the zero correlations between two widely separated data points (either in respect to space and/or time) more precisely.
We have also applied our model and methods to two real data examples pertaining to spatial and spatio-temporal dependence. 
As illustrated in detail, in both cases our model exhibited excellent performance.

Although for the current paper we restricted ourselves to spatio-temporal applications only, 
our model is readily applicable in the functional data context. In fact, in the context of 
nonparametric function estimation, a new class of prior distributions can be introduced 
through our proposed model. Note that unknown functions can be modeled as a limit of a weighted sum of  
kernels or generator functions indexed by continuous parameters. In our model the 
weights will be the $p_{i}$'s of ODDP, and kernels are indexed by $\theta_{i}$, where 
for $i=1,2,\ldots$, $\theta_i\stackrel{iid}{\sim}G_0$. 
We have already obtained some sufficient conditions ensuring that our model converges 
in $L_{p}$ norm and Besov semi-norm. These results make our proposed model a promising 
candidate for function estimation.
 
  
\section*{Acknowledgments}
We are grateful to Prof. Sujit Sahu for kindly providing us with the ozone data set
and to Mr. Suman Guha for helpful dicsussions.

\newpage
\renewcommand\thefigure{S-\arabic{figure}}
\renewcommand\thetable{S-\arabic{table}}
\renewcommand\thesection{S-\arabic{section}}

\setcounter{section}{0}
\setcounter{figure}{0}
\setcounter{table}{0}

\begin{center}
	{\bf \LARGE Supplementary Material}
\end{center}

\section{Proof of Theorem 1}
\label{sec:proof_theorem1}
Note that
\begin{equation}
\int |K(\bi{x},\bi{\theta})|dG_{\bi{x}}(\bi{\theta}) =\sum_{i=1}^{\infty}|K(\bi{x},\bi{\theta}_{\pi_{i}(\bi{x})})|p_{i},
\label{eqn:1}
\end{equation}
so that the monotone convergence theorem yields
\begin{align*}
E\left(\int |K(\bi{x},\bi{\theta})|dG_{\bi{x}}(\bi{\theta})\right)
&=\sum_{i=1}^{\infty}E\left(|K(\bi{x},\bi{\theta}_{\pi_{i}(\bi{x})})|\right)E(p_{i})
\\
&=\int|K(\bi{x},\bi{\theta})|dG_{0}(\bi{\theta})\sum_{i=1}^{\infty}E\left(p_{i}\right)
\\
&=\int|K(\bi{x},\bi{\theta})|dG_{0}(\bi{\theta})\quad \mbox{[since  $ \sum_{i=1}^{\infty}p_{i}=1$]}.
\end{align*}
Here we have used the fact that $\bi{\theta}_{i}$ and $V_{i}$ 's are independent. 
Therefore $\int |K(\bi{x},\bi{\theta})|dG_{\bi{x}}(\bi{\theta})$ is finite with probability one and hence
\[
f(\bi{x})=\int K(\bi{x},\bi{\theta})dG_{x}(\bi{\theta})=\sum_{i=1}^{\infty}K(\bi{x},\bi{\theta}_{\pi_{i}(\bi{x})})p_{i}(\bi{x}) 
\]
is absolutely convergent with probability one. Since this series is bounded by (\ref{eqn:1}), which is integrable, 
the bounded convergence theorem implies
\begin{align*}
E(f(\bi{x}))=&\sum_{i=1}^{\infty}E\left(K(\bi{x},\bi{\theta}_{\pi_{i}(\bi{x})})\right)E(p_{i})
\\
&=\int K(\bi{x},\bi{\theta})dG_{0}(\bi{\theta}).
\end{align*}
\hfill$\blacksquare$


\section{Proof of Theorem 2}
\label{sec:proof_theorem2}
\begin{align}
f(\bi{x_{1}})f(\bi{x_{2}})&=\int K(\bi{x_{1}},\bi{\theta})dG_{\bi{x_{1}}}(\bi{\theta})\int K(\bi{x_{2}},\bi{\theta})dG_{\bi{x_{2}}}(\bi{\theta})\notag
\\
&=\sum_{i} K(\bi{x_{1}},\bi{\theta}_{\pi_{i}(\bi{x_{1}})})p_{i}(\bi{x_{1}})\sum_{j} K(\bi{x_{2}},\bi{\theta}_{\pi_{j}(\bi{x_{2}})})p_{j}(\bi{x_{2}})\notag
\\
&=\sum_{i}\sum_{j}K(\bi{x_{1}},\bi{\theta}_{\pi_{i}(\bi{x_{1}})})K(\bi{x_{2}},\bi{\theta}_{\pi_{j}(\bi{x_{2}})})p_{i}(\bi{x_{1}})p_{j}(\bi{x_{2}})\label{eqn:2}
\end{align}
since both the series are absolutely convergent with probability one. This is bounded in absolute value by 
\begin{equation}
\sum_{i}\sum_{j}|K(\bi{x_{1}},\bi{\theta}_{\pi_{i}(\bi{x_{1}})})K(\bi{x_{2}},\bi{\theta}_{\pi_{j}(\bi{x_{2}})})|
p_{i}(\bi{x_{1}})p_{j}(\bi{x_{2}}).
\label{eqn:3}
\end{equation}
Now if this is an integrable random variable, we may take the expectation of (\ref{eqn:2}) inside the summation sign and obtain
\begin{align*}
E\left(f(\bi{x_{1}})f(\bi{x_{2}})\right)&=E\left(\int K(\bi{x_{1}},\bi{\theta})dG_{\bi{x_{1}}}(\bi{\theta})
\int K(\bi{x_{2}},\bi{\theta})dG_{\bi{x_{2}}}(\bi{\theta})\right)\notag
\\
&=\sum_{i}\sum_{j}E\left(K(\bi{x_{1}},\bi{\theta}_{\pi_{i}(\bi{x_{1}})})
K(\bi{x_{2}},\bi{\theta}_{\pi_{j}(\bi{x_{2}})})\right)E\left(p_{i}(\bi{x_{1}})p_{j}(\bi{x_{2}})\right)
\\
&=\sum_{\pi_{i}(\bi{x_{1}})\neq\pi_{j}(\bi{x_{2}})}
E\left(K(\bi{x_{1}},\bi{\theta}_{\pi_{i}(\bi{x_{1}})})K(\bi{x_{2}},\bi{\theta}_{\pi_{j}(\bi{x_{2}})})\right)\,
E\left(p_{i}(\bi{x_{1}})p_{j}(\bi{x_{2}})\right)
\\
&\quad +
\sum_{\pi_{i}(\bi{x_{1}})=\pi_{j}(\bi{x_{2}})}
E\left(K(\bi{x_{1}},\bi{\theta}_{\pi_{i}(\bi{x_{1}})})K(\bi{x_{2}},\bi{\theta}_{\pi_{j}(\bi{x_{2}})})\right)
E\left(p_{i}(\bi{x_{1}})p_{j}(\bi{x_{2}})\right)
\\
&=
\sum_{\pi_{i}(\bi{x_{1}})\neq\pi_{j}(\bi{x_{2}})}
E\left(K(\bi{x_{1}},\bi{\theta})\right)E\left(K(\bi{x_{2}},\bi{\theta})\right)
E\left(p_{i}(\bi{x_{1}})p_{j}(\bi{x_{2}})\right)
\\
&\quad+
\sum_{\pi_{i}(\bi{x_{1}})=\pi_{j}(\bi{x_{2}})}
E\left(K(\bi{x_{1}},\bi{\theta})K(\bi{x_{2}},\bi{\theta})\right)
E\left(p_{i}(\bi{x_{1}})p_{j}(\bi{x_{2}})\right)
\\
&=
E\left(K(\bi{x_{1}},\bi{\theta})\right)E\left(K(\bi{x_{2}},\bi{\theta}\right)
\sum_{\pi_{i}(\bi{x_{1}})\neq\pi_{j}(\bi{x_{2}})}
E\left(p_{i}(\bi{x_{1}})p_{j}(\bi{x_{2}})\right)
\\
&\quad+
E\left(K(\bi{x_{1}},\bi{\theta})K(\bi{x_{2}},\bi{\theta})\right)
\sum_{\pi_{i}(\bi{x_{1}})=\pi_{j}(\bi{x_{2}})}
E\left(p_{i}(\bi{x_{1}})p_{j}(\bi{x_{2}})\right)
\\
%
&=
E\left(K(\bi{x_{1}},\bi{\theta})\right)E\left(K(\bi{x_{2}},\bi{\theta}\right)\notag\\
&\bigg[\sum_{\pi_{i}(\bi{x_{1}})\neq\pi_{j}(\bi{x_{2}})}
E\left(p_{i}(\bi{x_{1}})p_{j}(\bi{x_{2}})\right)
+\sum_{\pi_{i}(\bi{x_{1}})=\pi_{j}(\bi{x_{2}})}
E\left(p_{i}(\bi{x_{1}})p_{j}(\bi{x_{2}})\right)
\bigg]
\\
&\quad+
E\left(K(\bi{x_{1}},\bi{\theta})K(\bi{x_{2}},\bi{\theta})\right)
\sum_{\pi_{i}(\bi{x_{1}})=\pi_{j}(\bi{x_{2}})}
E\left(p_{i}(\bi{x_{1}})p_{j}(\bi{x_{2}})\right)
\\
&\quad-
E\left(K(\bi{x_{1}},\bi{\theta})\right)
E\left(K(\bi{x_{2}},\bi{\theta})\right)
\sum_{\pi_{i}(\bi{x_{1}})=\pi_{j}(\bi{x_{2}})}
E\left(p_{i}(\bi{x_{1}})p_{j}(\bi{x_{2}})\right)
\end{align*}
\begin{align*}
&=
\mbox{Cov}\left(K(\bi{x_{1}},\bi{\theta}),K(\bi{x_{2}},\bi{\theta})\right)
\sum_{\pi_{i}(\bi{x_{1}})=\pi_{j}(\bi{x_{2}})}
E\left(p_{i}(\bi{x_{1}})p_{j}(\bi{x_{2}})\right)
\\
&\quad+
E\left(K(\bi{x_{1}},\bi{\theta})\right)E\left(K(\bi{x_{2}},\bi{\theta})\right)
\bigg[
\sum_{\pi_i(\bi{x_{1}}),\pi_j(\bi{x_{2}})}
E\left(p_{i}(\bi{x_{1}})p_{j}(\bi{x_{2}})\right)
\bigg]
\\
&=
\mbox{Cov}\left(K(\bi{x_{1}},\bi{\theta}),K(\bi{x_{2}},\bi{\theta})\right)
\sum_{\pi_{i}(\bi{x_{1}})=\pi_{j}(\bi{x_{2}})}
E\left(p_{i}(\bi{x_{1}})p_{j}(\bi{x_{2}})\right)
\\
&\quad+
E\left(K(\bi{x_{1}},\bi{\theta})\right)E\left(K(\bi{x_{2}},\bi{\theta})\right).
\end{align*}
An analogous equation shows that (\ref{eqn:3}) is bounded, since, by our assumption 
\[
 \int |K(\bi{x},\bi{\theta})dG_{0}(\bi{\theta}) < \infty \mbox{ and } 
 \int|K(\bi{x_{1}},\bi{\theta})K(\bi{x_{2}},\bi{\theta})|dG_{0}(\bi{\theta}) < \infty.
\]

To obtain the complete analytical expression of $E\left(f(\bi{x_{1}})f(\bi{x_{2}})\right)$ we need to calculate\\ 
$\sum_{\pi_{i}(\bi{x_{1}})=\pi_{j}(\bi{x_{2}})}E\left(p_{i}(\bi{x_{1}})p_{j}(\bi{x_{2}})\right)$.
Define
\[
T(\bi{x_{1}},\bi{x_{2}})=\{k|\mbox{ there exists } i,j \mbox{ such that } \pi_{i}(\bi{x_{1}})=\pi_{j}(\bi{x_{2}})=k\},
\]
for $k\in T(\bi{x_{1}},\bi{x_{2}})$. We further define
$
A_{lk}=\{\pi_{j}(\bi{x_{l}})|j<i \mbox{ where } \pi_i(\bi{x_{l}})=k\},
$
$S_{k}$ = $A_{1k}\cap A_{2k}$ and $S'_{k}$ = $A_{1k}\cup A_{2k}-S_{k}$.
\\[2ex]
Then it can be easily shown that
\[ \sum_{\pi_{i}(\bi{x_{1}})=\pi_{j}(\bi{x_{2}})}
E\left(p_{i}(\bi{x_{1}})p_{j}(\bi{x_{2}})\right)
=
\frac{2}{(\alpha+1)(\alpha+2)}\sum_{k\in T(\bi{x_{1}},\bi{x_{2}})}
\left(\frac{\alpha}{\alpha+2}\right)^{\#S_{k}}\left(\frac{\alpha}{\alpha+1}\right)^{\#S_{k}'}.
\]

Hence
\begin{align*}
\mbox{Cov}(f(\bi{x_{1}}),f(\bi{x_{2}}))
&=
\mbox{Cov}_{G_0}(K(\bi{x_{1}},\bi{\theta}),K(\bi{x_{2}},\bi{\theta}))
\times\frac{2}{(\alpha+1)(\alpha+2)}
\\
&\quad\times\sum_{k\in T(\bi{x_{1}},\bi{x_{2}})}\left(\frac{\alpha}{\alpha+2}\right)^{\#S_{k}}
\left(\frac{\alpha}{\alpha+1}\right)^{\#S_{k}'}.
\end{align*}

\hfill$\blacksquare$

\section{Proof of Theorem 4}
In this context, it is more convenient to deal with the notation used in the context of 
Theorem 2. 
Note that $\|\bx_1-\bx_2\|\rightarrow 0$ implies that 
$\# S_k\rightarrow (k-1)$,
$\# S'_k\rightarrow 0$, and hence, for any realization of the point process, 
$\mbox{Corr}(G_{\bx_1},G_{\bx_2})=\sum_{k\in T(\bx_1,\bx_2)}\left(\frac{\alpha}{\alpha+2}\right)^{\# S_k}
\left(\frac{\alpha}{\alpha+1}\right)^{\# S'_k}
\rightarrow 
\sum_{k=1}^{\infty}\left(\frac{\alpha}{\alpha+2}\right)^{k-1}
=(\alpha+2)/2$.
Since $\mbox{Corr}_{G_0}(K(\bx_1,\btheta),K(\bx_2,\btheta))\rightarrow 1$ as
$\|\bx_1-\bx_2\|\rightarrow 0$, it follows that 
$\mbox{Corr}(f(\bx_1),f(\bx_2))\rightarrow 1$ as $\|\bx_1-\bx_2\|\rightarrow 0$.
On the other hand, as $\|\bx_1-\bx_2\|\rightarrow\infty$, 
$\# T(\bx_1,\bx_2)$ is at most finite, $\# S_k\rightarrow 0$ and $\# S'_k\rightarrow\infty$.
This implies that $\mbox{Corr}(G_{\bx_1},G_{\bx_2})\rightarrow 0$, and hence,
$\mbox{Corr}(f(\bx_1),f(\bx_2))\rightarrow 0$.
Hence, by the dominated convergence theorem it follows that the unconditional
correlation between $f(\bx_1)$ and $f(\bx_2)$ goes to 1 and 0, respectively, as 
$\|\bx_1-\bx_2\|\rightarrow 0$ and $\|\bx_1-\bx_2\|\rightarrow\infty$.

\section{Proof of Theorem 7}

Since for each $\bx$, $f(\bx)=\sum_{i=1}^{\infty}K(\bx,\btheta_{\pi_i(\bx)})p_i(\bx)$, 
each $\bx$ must satisfy $\pi_k(\bx)=i_k$, for every $k=1,2,\ldots$, where $i_k\in\{1,2, \ldots\}$
($i_k\neq i_{k'}$ for any $k\neq k'$), we must have $\bx\in \cap_{k=1}^{\infty} A_{ki_k}$.
For simplicity but without loss of generality let $i_k=k$ for $k=1,2,\ldots$. Then for $\bx\in \cap_{k=1}^{\infty} A_{ki_k}$
it holds that $K(\bx,\btheta_{\pi_i(\bx)})=K(\bx,\btheta_i)$ and $p_i(\bx)=V_i\prod_{j<i}(1-V_j)=p_i$, say.
Then, for any arbitrary $\bx_0$ in the interior of $\cap_{k=1}^{\infty} A_{kk}$, it holds almost surely that
\begin{align}
\underset{\bi{x}\rightarrow \bi{x}_0}{\lim} f(\bi{x})&=\underset{\bi{x}\rightarrow \bi{x}_0}{\lim} 
\sum_{i=1}^\infty K(\bi{x},\btheta_i)p_i
\\[1ex]
&=\sum_{i=1}^\infty \underset{\bi{x}\rightarrow \bi{x}_0}{\lim} K(\bx,\btheta_i)p_i\notag
\\[1ex]
&[\mbox{using (A1) and the dominated convergence theorem}\notag
\\[1ex]
&\mbox{following from the facts that $K(\cdot,\cdot)$ is bounded and} \sum_{i=1}^{\infty}p_i=1.]\notag
\\[1ex]
&=\sum_{i=1}^{\infty}K(\bi{x}_0,\btheta_i)p_i\hspace{3mm}[\mbox{using (A2)}.]\notag
\\[1ex]
&=f(\bx_0)\label{eq:as_continuity1}.
\end{align}
Hence $f(\cdot)$ is almost surely continuous in the interior of $\cap_{k=1}^{\infty} A_{ki_k}$.

To prove mean square continuity first note that the dominated convergence theorem can be applied as before, using 
boundedness of $K(\cdot,\cdot)$ and the fact that $\sum_{i=1}^{\infty}p_i=1$ 
to guarantee that
the following hold almost surely:
\begin{align}
\lim_{\bx\rightarrow\bx_0}f(\bx)^2&=f(\bx_0)^2,\label{eq:as_continuity2}
\\[1ex]
\lim_{\bx\rightarrow\bx_0}f(\bx)f(\bx_0)&=f(\bx_0)^2.\label{eq:as_continuity3}
\end{align}
Combining (\ref{eq:as_continuity1}), (\ref{eq:as_continuity2}) and (\ref{eq:as_continuity3}) 
implies that $(f(\bx)-f(\bx_0))^2\rightarrow 0$ almost surely.
Now since $f(\cdot)$ is bounded almost surely by $M$ (follows from (A1)
and the fact that $\sum_{i=1}^{\infty}p_i=1$), $f(\cdot)^2$ and $f(\cdot)f(\bx_0)$ are almost surely bounded as well.
Hence, taking
expectations and using the dominated convergence theorem using the boundedness of $(f(\bx)-f(\bx_0))^2$, it follows that 
$$\lim_{\bi{x}\rightarrow \bi{x}_0} E[f(\bi{x})-f(\bi{x}_0)]^2=0.$$
Therefore, $f(\bi{x})$ is mean square continuous in the interior of $\cap_{k=1}^{\infty} A_{ki_k}$.

Let us now show that if $\bx_0\in \cap_{k=1}^{\infty} A_{k_k}$ lies at the boundary of $A_{kk}$ for some $k$,
then $f(\cdot)$ is almost surely discontinuous at $\bx_0$. 
It is useful to note that for each $k$, $\pi_k(\cdot)$ is a step function and admits the representation
\[
\pi_k(\bx)=\sum_{i=1}^{\infty}iI_{A_{ki}}(\bx).
\]
For the sake of clarity, without loss of generality, 
let us assume that the dimensionality $d=1$, so that $\bx_0=x_0$ is one-dimensional. Let us further assume, without loss
of generality, that $x_0$ falls on the rightmost boundary of $A_{11}=\{x:\pi_1(x)=1\}$, so that
$x_0=\sup\{x:\pi_1(x)=1\}$. Then, almost surely,
\begin{align}
\lim_{x\downarrow x_0}f(x)=\sum_{i=1}^\infty K({x},\btheta_{i+1})p_{i+1}
\neq \sum_{i=1}^\infty K({x},\btheta_i)p_i=f(x_0),\notag
\end{align}
showing that $f(\cdot)$ is almost surely discontinuous at $x_0$.
\hfill$\blacksquare$

\section{Proof of Theorem 9}
\label{sec:proof_theorem4}

Without loss of generality, let 
$\bx_0$ be an arbitrary point in the interior of $\cap_{k=1}^{\infty} A_{kk}$.
Then, for any direction $\bu$ such that $\bx=\bx_0+\bu\in\mathcal N(\bx_0)\cap\{\cap_{k=1}^{\infty} A_{kk}\}$,
where $\mathcal N(\bx_0)$ is any neighborhood of $\bx_0$,
\begin{equation}
f(\bi{x})=\sum_{i=1}^\infty K(\bi{x},\btheta_i)p_i
\label{eq:diff1}
\end{equation}
and for each $i=1,2,\ldots$,  
$K(\bx,\btheta_i)$ admits the following (multivariate) Taylor's series expansion:
\begin{equation}
K(\bi{x}_0 +\bi{u},\btheta_i)
=K(\bi{x}_0,\btheta_i) + \bi{u}'\bigtriangledown K(\bi{x}_0,\btheta_i)+R(\bi{x}_0,\bu,\btheta_i)
\label{eq:diff2}
\end{equation}
where, $\left|R(\bx_0,\bu,\btheta_i)\right|\leq c(\bx_0,\btheta_i)\|\bu\|^2$,
for some function $c(\cdot,\cdot)$, independent of $\bu$.
The boundedness assumption (B1) guarantees that $c(\cdot,\cdot)$ is bounded above by some
finite constant $M_1$. Hence, for $i=1,2,\ldots$,
\begin{equation}
R(\bx_0,\bu,\btheta_i)\leq M_1\|\bu\|^2.
\label{eq:diff3}
\end{equation}
It follows from (\ref{eq:diff2}) that
\begin{equation}
 f(\bi{x})=f(\bx_0)+\bu'\bigtriangledown f(\bx_0)+R_2(\bx_0,\bu),
 \label{eq:diff4}
\end{equation}
where
\begin{align}
f(\bx_0)&=\sum_{i=1}^{\infty}K(\bx_0,\btheta_i)p_i\notag
\\[1ex]
\bigtriangledown f(\bx_0)&=\sum_{i=1}^{\infty}\bigtriangledown K(\bx_0,\btheta_i)p_i\notag
\\[1ex]
R_2(\bx_0,\bu)&=\sum_{i=1}^{\infty}R(\bi{x}_0,\bu,\btheta_i)p_i\notag
\end{align}
In (\ref{eq:diff4})  $\bu'\bigtriangledown f(\bx_0)$ is clearly a process linear in $\bu$.
Moreover, since $\sum_{i=1}^{\infty}p_i=1$, $\left|R_2(\bx_0,\bu)\right|$ is bounded above by $M_1\|\bu\|^2$. Hence,
almost surely, $\frac{\left|R_2(\bx_0,\bu)\right|}{\|\bu\|}\rightarrow 0$.
Hence, using the dominated convergence theorem it follows that
\[
\underset{\|\bi{u}\|\rightarrow 0}{\operatorname {\mbox{lim}}} E\left[\frac{f(\bi{x}_{0}+\bi{u})-f(\bi{x}_{0})-\bu'\bigtriangledown f(\bx_0)}{\|\bi{u}\|}\right]^2
=\underset{\|\bi{u}\|\rightarrow 0}{\operatorname {\mbox{lim}}} E\left[\frac{R(\bi{x}_{0},\bi{u})}{\|\bi{u}\|}\right]^2
 = 0.
\]
Hence, $f(\cdot)$ is mean square differentiable in the interior of $\cap_{k=1}^{\infty} A_{ki_k}$.
\hfill$\blacksquare$

\section{Proof of Theorem 10}
\label{sec:proof_theorem5}

For our purpose, we first state and prove a lemma.
\begin{lemma}
\label{Theorem:truncation1}
Let $p_k$ denote the random weights from an ODDP. For each positive integer $N\geq 1$ 
and each positive integer $r\geq 1$, let 
\[
 T_N(r,\alpha)=\left(\sum_{k=N}^\infty p_{k}\right)^r, \ U_N(r,\alpha)=\sum_{k=N}^\infty p_{k}^r. 
\]
Then 
\[
 E(T_N(r,\alpha))=\left(\frac{\alpha}{\alpha+r}\right)^{N-1},
\]
and
\[
 E(U_N(r,\alpha))=\left(\frac{\alpha}{\alpha+r}\right)^{N-1}\frac{\Gamma(r)\Gamma(\alpha+1)}{\Gamma(\alpha+r)}.
\]
\end{lemma}

\begin{proof} 
Let $\mathcal{P}$ be a specific random measure from ODDP. Then 
\[
 \mathcal{P}(\cdot)=V_{\pi_{1}}\delta_{\theta_{\pi_{1}}}(\cdot)+ 
 (1-V_{\pi_{1}})\left(V_{\pi_{1}}^*\delta_{\theta_{\pi_{1}}^*}(\cdot) +
(1-V_{\pi_{1}}^*)V_{\pi_{2}}^*\delta_{\theta_{\pi_{2}}^*}(\cdot)
+(1-V_{\pi_{1}}^*)(1-V_{\pi_{2}}^*)V_{\pi_{3}}^*\delta_{\theta_{\pi_{3}}^*}(\cdot)+\cdots\right),
\]
where $V_{\pi_{k}}^*=V_{\pi_{k+1}}$ are independent $Beta(1,\alpha)$ random variables and 
$\theta_{\pi_{k}}^*=\theta_{\pi_{k+1}}$ are $iid$ $G_0$. So we have  
\[
 \mathcal{P}(\cdot) \stackrel{\mathcal{D}}{=} 
 V_{\pi_{1}}\delta_{\theta_{\pi_{1}}}(\cdot) + (1-V_{\pi_{1}})\mathcal{P}^*(\cdot),
\]
where  $V_{\pi_{1}},\theta_{\pi_{1}}$ and  $\mathcal{P^*}(\cdot)$ are independent and  
$ \mathcal{P^*}(\cdot)$ is an ODDP.

Similarly we can show that 
\[
 U_1(r,\alpha)\stackrel{\mathcal{D}}{=}V_{\pi_{1}}^r+(1-V_{\pi_{1}})^r U_1(r,\alpha),
\]
where on the right-hand side $V_{\pi_{1}}$ and $ U_1(r,\alpha)$ are mutually independent. Therefore taking expectations
\[
 E(U_1(r,\alpha))=\frac{\Gamma{(r+1)}\Gamma{(\alpha+1)}}{\Gamma{(\alpha+r+1)}}+\frac{\alpha}{\alpha+r}E(U_1(r,\alpha)).
\]
Then we have 
\begin{equation}
\label{equation:truncation1}
  E (U_1(r,\alpha))=\frac{\Gamma(r)\Gamma(\alpha+1)}{\Gamma(\alpha+r)}.
\end{equation}

Furthermore for $N\geq 2$, we have that 
\[
 U_N(r,\alpha)= (1-V_{\pi_{1}})^r\cdots(1-V_{\pi_{(N-1)}})^r(U_1(r,\alpha)),
\]
where all the variables on the right hand side are mutually independent. Taking expectations, we have 
\[
 E(U_N(r,\alpha))= \left(\displaystyle\prod_{k=1}^{N-1}E(1-V_{\pi_{k}})^r\right)E((U_1(r,\alpha))).
\]
Then using (\ref{equation:truncation1}) we have 
\[
 E(U_N(r,\alpha))=\left(\frac{\alpha}{\alpha+r}\right)^{N-1}\frac{\Gamma(r)\Gamma(\alpha+1)}{\Gamma(\alpha+r)}.
\]
Now, similarly we can show that 
\begin{align*}
 T_N(r,\alpha)&\stackrel{\mathcal{D}}{=}(1-V_{\pi_{1}})^r\cdots(1-V_{\pi_{(N-1)}})^r(T_1(r,\alpha))
\\
&=(1-V_{\pi_{1}})^r\cdots(1-V_{\pi_{(N-1)}})^r
\\
&=\left(\frac{\alpha}{\alpha+r}\right)^{N-1}.
\end{align*}
Hence, the lemma is proved.
\end{proof}

We now proceed to the proof of Theorem 10.
Note that
\begin{align}
&\quad \left|m_N(\bi{y})-m_\infty(\bi{y})\right| \notag
\\
&=\left|\int_\Theta [\bi{y}|P_N][P]d\Theta-\int_\Theta[\bi{y}|P][P]d\Theta\right|\notag
\\
&\leq \int_\Theta \left| [\bi{y}|P_N] - [\bi{y}|P] \right|[P]d\Theta \label{trunacation3}
\end{align}

Now we expand $[\bi{y}|P_N]$ around $P$ using  multivariate Taylor's series expansion up to the second order:
\begin{align*}
& [\bi{y}|P_N]= [\bi{y}|P] +\frac{\partial[\bi{y}|P_N]}{\partial P_{N}}\bigg |_{P_N=P}(P_N-P)
+(P_N-P)'\frac{\partial^2[\bi{y}|P_N]}{\partial P_{N}^2}\bigg |_{P_N=P^*}(P_N-P), 
\\
&\mbox{[ where $P^*$ lies between $P$ and $P_N$, that is,   } \| P-P^*\| \leq \|P-P_{N}\|.\quad ]
\end{align*}
Noting that $ [\bi{y}|P_N] $ and $[\bi{y}|P_N] $ are  multivariate normal densities with mean  $P_N$ and $P$ 
respectively and variance $\sigma^2I$, where $I$ is the $n\times n$ identity matrix, we have 
\begin{align*}
 &\frac{\partial[\bi{y}|P_N]}{d P_{N}}\bigg |_{P_N=P}=[\bi{y}|P](\bi{y}-P)',\ \ \mbox{and}
\\
&\frac{\partial^2[\bi{y}|P_N]}{\partial P_{N}^2}\bigg |_{P_N=P^*}=[\bi{y}|P^*]I(\bi{y}-P^*)(\bi{y}-P^*)'I-[\bi{y}|P^*] I.
\end{align*}
Therefore (\ref{trunacation3}) becomes
\begin{align*}
& \int \left|[\bi{y}|P](\bi{y}-P)'(P_N-P) +(P_N-P)'[\bi{y}|P^*](\bi{y}-P^*)(\bi{y}-P^*)'(P_N-P)\right.\\
&\left. \quad\quad\quad\quad -(P_N-P)'[\bi{y}|P^*] I (P_N-P)\right|[P]d\Theta.
\end{align*}

So, we have 
\begin{align*}
 &\quad \int _{\mathbb{R}^n} \left|m_N(\bi{y})-m_\infty(\bi{y})\right|d\bi{y}
\\
&\leq \int_ {\mathbb{R}^n}\int \left|[\bi{y}|P](\bi{y}-P)'(P_N-P) 
+(P_N-P)'[\bi{y}|P^*](\bi{y}-P^*)(\bi{y}-P^*)'(P_N-P)\right.\\
&\left.\quad\quad\quad\quad-(P_N-P)'[\bi{y}|P^*] I (P_N-P)\right|[P]d\Theta.
\end{align*}
Now, using Fubini's theorem we can interchange the order of integration.
Finally we have 
\begin{align}
&\quad \int_ {\mathbb{R}^n}\int \left|[\bi{y}|P](\bi{y}-P)'(P_N-P) 
+(P_N-P)'[\bi{y}|P^*](\bi{y}-P^*)(\bi{y}-P^*)'(P_N-P)\right.\notag\\
&\left.\quad\quad\quad\quad -(P_N-P)'[\bi{y}|P^*] I (P_N-P)\right|[P]d\Theta d{\bi{y}}\notag
\\
&=\int\int_ {\mathbb{R}^n} \left|[\bi{y}|P](\bi{y}-P)'(P_N-P) 
+(P_N-P)'[\bi{y}|P^*](\bi{y}-P^*)(\bi{y}-P^*)'(P_N-P)\right.\notag\\
&\left.\quad\quad\quad\quad-(P_N-P)'[\bi{y}|P^*] I (P_N-P)\right|d{\bi{y}}[P]d\Theta \notag
\\
&\leq \int\left[\int_ {\mathbb{R}^n}[\bi{y}|P]\left|(\bi{y}-P)'(P_N-P)\right|d{\bi{y}}\right.\notag\\
&\left.\quad\quad\quad\quad+\int_ {\mathbb{R}^n}(P_N-P)'\left\{[\bi{y}|P^*](\bi{y}-P^*)(\bi{y}-P^*)'+[\bi{y}|P^*] I\right\} (P_N-P)d{\bi{y}}\right][P]d\Theta\notag
\\
&= \int\left[\sqrt{\frac{2}{\pi}}\bi{1}_{n\times 1}'|P_N-P|+2(P_N-P)'I(P_N-P)\right][P]d\Theta\notag
\\
&=E_\Theta\left[\sqrt{\frac{2}{\pi}}\bi{1}_{n\times 1}'|P_N-P|+2(P_N-P)'(P_N-P)\right] \notag
\\
&=2E_\Theta\left[\sqrt{\frac{2}{\pi}}\sum_{i=1}^n|P_N(\bx_{i})-P(\bx_{i})|+\sum_{i=1}^n(P_N(\bx_{i})-P(\bx_{i}))^2\right]. 
\label{truncation4}
\end{align}
In the above, $\bi{1}_{n\times 1}$ denotes the $n$-component vector with each element 1. 
%
%
Now,
\begin{align*}
 &\quad \left|P_N(\bx_{i})-P(\bx_{i})\right|\leq  M\left|p_N^N(\bx_{i})-p_N^\infty(\bx_{i})
 -\sum_{i=N+1}^\infty p_{i}(\bx_{i})\right|,
\end{align*}
where $ p_N^N(\bx_i)=1-\sum_{k=1}^{N-1}p_k(\bx_i)
=1-\sum_{k=1}^{N-1}V_{\pi_k(\bx_i)}\prod_{j<k}\left(1-V_{\pi_j(\bx_i)}\right)$ and 
$p_N^\infty(\bx_i)=V_{\pi_N(\bx_i)}\prod_{j<N}\left(1-V_{\pi_j(\bx_i)}\right)$ are the random
weights corresponding to the N-th coefficient 
in the Sethuraman construction of truncated ODDP and the original ODDP respectively.
We also have 
\begin{align*}
 & \left(P_N(\bx_{i})-P(\bx_{i})\right)^2
\\
&\leq  M^2\left(p_N^N(\bx_{i})-p_N^\infty(\bx_{i})-\sum_{i=N+1}^\infty p_{i}(\bx_{i})\right)^2
\\
&=M^2 \left[\left(p_N^N(\bx_{i})-p_N^\infty(\bx_i)\right)^2+\left(\sum_{i=N+1}^\infty p_{i}(\bx_i)\right)^2
+2\left(p_N^N(\bx_i)-p_N^\infty(\bx_i)\right)\left(\sum_{i=N}^\infty p_{i}(\bx_i)\right)\right].
\end{align*}
Therefore 
\begin{align*}
 & E_\Theta\left[\sum_{i=1}^n \left|P_N(\bx_{i})-P(\bx_{i})\right|\right] \notag
\\
&\leq \sum_{i=1}^n E_{\Theta}M \left[\left|p_N^N(\bx_{i})-p_N^\infty(\bx_i)\right |
+\sum_{i=N+1}^\infty p_{i}(\bx_i)\right],\notag
\end{align*}
and
\begin{align}
& E_\Theta\left[\sum_{i=1}^n\left(P_N(\bx_{i})-P(\bx_{i})\right)^2\right] \notag
\\
&\leq \sum_{i=1}^n E_{\Theta}M^2 \left[\left(p_N^N(\bx_{i})-p_N^\infty(\bx_i)\right)^2
+\left(\sum_{i=N+1}^\infty p_{i}(\bx_i)\right)^2+2\left(p_N^N(\bx_i)-p_N^\infty(\bx_i)\right)
\left(\sum_{i=N+1}^\infty p_{i}(\bx_i)\right)\right]\notag
\end{align}

Now,
\begin{align}
 E\left(|p_N^N(\bx_{i})-p_N^\infty(\bx_i)|\right)&=E\left[(1-V_{\pi_{1}})(1-V_{\pi_{2}})\cdots(1-V_{\pi_{N}})\right] \notag
\\
&=\left(\frac{\alpha}{\alpha+1}\right)^N \mbox{ Since $V_{i}$'s are $iid$ $Beta(1,\alpha)$ random variables}. 
\label{IJ11}
\end{align}

Similarly, we have 
\begin{align}
 E(p_N^N(\bx_{i})-p_N^\infty(\bx_i))^2&=E\left[(1-V_{\pi_{1}})(1-V_{\pi_{2}})\cdots(1-V_{\pi_{N}})\right]^2 \notag
\\
&=\left(\frac{\alpha}{\alpha+2}\right)^N. \label{IJ1_2}
\end{align}

Now, using  Theorem \ref{Theorem:truncation1} we have 
\begin{equation}
\label{equation:IJ2}
 E_{\Theta}\left(\sum_{i=N+1}^\infty p_{i}(\bx_i)\right)=\left(\frac{\alpha}{\alpha+1}\right)^N,
\end{equation}
and 
\begin{equation}
\label{equation:IJ22}
 E_{\Theta}\left(\sum_{i=N+1}^\infty p_{i}(\bx_i)\right)^2=\left(\frac{\alpha}{\alpha+2}\right)^N.
\end{equation}

Now,
\begin{align}
&\quad E_{\Theta}\left[2\left(p_N^N(x_i)-p_N^\infty(x_i)\right)\left(\sum_{i=N+1}^\infty p_{i}(x_i)\right)\right]\notag
\\
&=2E(1-V_{\pi_{1}})(1-V_{\pi_{2}})\cdots(1-V_{\pi_{N}}) \sum_{i=N+1}^\infty (1-V_{\pi_{1}})(1-V_{\pi_{2}})
\cdots(1-V_{\pi_{N}})\cdots (1-V_{\pi_{i-1}})V_{\pi_{i}}\notag
\\
&=2\left(\frac{\alpha}{\alpha+2}\right)^N \sum_{i=N}^\infty \left(\frac{\alpha}{\alpha+1}\right)^{i-N-1} \left(\frac{1}{\alpha+1}\right)\notag
\\
&=2\left(\frac{\alpha}{\alpha+2}\right)^N.\label{IJ3}
\end{align}
So, combining (\ref{IJ1_2}), (\ref{equation:IJ2}), (\ref{equation:IJ22}) and (\ref{IJ3}) we have
\begin{align}
&\sum_{i=1}^n E_{\Theta}M^2 \left[\left(p_N^N(\bx_{i})-p_N^\infty(\bx_i)\right)^2
+\left(\sum_{i=N+1}^\infty p_{i}(\bx_i)\right)^2+2\left(p_N^N(\bx_i)-p_N^\infty(\bx_i)\right)
\left(\sum_{i=N+1}^\infty p_{i}(\bx_i)\right)\right]\notag
\\
&=4M^2n \left(\frac{\alpha}{\alpha+2}\right)^N. \label {IJ4}
\end{align}
Therefore, combining (\ref{IJ11}) and (\ref{equation:IJ22}) with (\ref{IJ4}) finally we have 
\[
\int _{\mathbb{R}^n} \left|m_N(\bi{y})-m_\infty(\bi{y})\right|d\bi{y} 
\leq 4M^2n \left(\frac{\alpha}{\alpha+2}\right)^N + 2\sqrt{\frac{2}{\pi}}Mn\left(\frac{\alpha}{\alpha+1}\right)^N,
\]
thus completing the proof.
\hfill$\blacksquare$

\section{Transdimensional transformation based Markov chain Monte Carlo (TTMCMC)}
\label{sec:ttmcmc}

In order to obtain a valid algorithm based on transformations, \ctn{Dutta14} design appropriate ``move
types" so that detailed balance and irreducibility hold. We first illustrate the basic idea on transformation
based moves with a simple example. Given that we are in the current state $x$, we may
propose the ``forward move" $x'=x+\e$, where
$\e>0$ is a simulation from some arbitrary density $g(\cdot)$ which is supported on the positive part
of the real line. To move back to $x$ from $x'$, we need to apply the ``backward transformation" $x'-\e$.
In general, given $\e$
and the current state $x$, we shall denote the forward transformation by $T(x,\e)$, and
the backward transformation by $T^b(x,\e)$.
For fixed $\e$ the forward and the backward transformations must be one to one and onto,
and satisfy $T^b(T(x,\e),\e)=x=T(T^b(x,\e),\e)$; see \ctn{Dutta14} for a detailed discussion
regarding these. 

The simple idea discussed above has been generalized to the multi-dimensional situation by
\ctn{Dutta14}. Remarkably, for any dimension, the moves can be constructed by simple deterministic
transformations of the one-dimensional random variable $\e$, which is simulated from any arbitrary
distribution on some relevant support. 

The idea based on transformations has been generalized to the case of variable dimensionality
by \ctn{Das14}.  In other words, \ctn{Das14} show that using simple deterministic transformations
and a single $\e$ (or just a few $\e$'s) it is possible to devise an effective dimension-hopping
algorithm which changes dimension as well as updates the other parameters, all in a single block,
while maintaining, at the same time, high acceptance rate. In this sense this new methodology
accomplishes automation of move-types. \ctn{Das14} refer to this dimension changing methodology
as Transdimensional Transformation  based Markov Chain Monte Carlo (TTMCMC). 

Before we illustrate the key concept of TTMCMC with a simple example, it is necessary
to define some requisite notation, borrowed from \ctn{Dutta14}.

\newcommand{\z}{\zeta}
\subsection{Notation}
\label{subsec:notation}
Suppose now that $\statesp$ is a $k$-dimensional space of the form $\statesp = \prod_{i=1}^k \statesp_i$ 
so that $T = (T_1,\ldots,T_k)$ where each $T_i : \statesp_i \times \mathcal D \to \statesp_i$, for some 
set $\mathcal D$, are the component-wise transformations. 
Let $\bzeta=(\z_1,\ldots,\z_k)$ be a vector of indicator variables, where, 
for $i=1,\ldots,k$, 
$\z_i=1$ and $\z_i=-1$ indicate, respectively, application of forward transformation and  
backward transformation to $x_i$, and let $\z_i=0$ denote no change to $x_i$.
Given any such indicator vector $\bzeta$, let us define
$ T_{\bzeta} = (g_{1,\z_1},g_{2,\z_2},\ldots,g_{k,\z_k})$
where 
\[ g_{i,\z_i} = \left\{ \begin{array}{ccc}
                  T_{i}^b & \textrm{ if } & \z_i=-1 \\ 
		  x_i   & \textrm{ if } & \z_i=0 \\
		  T_{i} & \textrm{ if } & \z_i=1. 
                 \end{array}
\right.\]
Corresponding to any given $\bzeta$, we also define the following `conjugate' vector
$\bzeta^c=(\z^c_1,\z^c_2,\ldots,\z^c_k)$, where 
\[ \z^c_i = \left\{ \begin{array}{ccc}
                  1 & \textrm{ if } & \z_i=-1 \\ 
		  0   & \textrm{ if } & \z_i=0 \\
		  -1 & \textrm{ if } & \z_i=1. 
                 \end{array}
\right.\]
With this definition of $\bzeta^c$, $T_{\bzeta^c}$ can be interpreted as the conjugate of $T_{\bzeta}$.

Since $3^k$ values of $\bzeta$ are possible, it is clear that $T$, via $\bzeta$,
induces $3^k$ many types of `moves' of the forms $\{T_{\bzeta_i};i=1,\ldots,3^k\}$ on the state-space. 
Suppose now that there is a subset $\Y$ of $\D$ such that the sets $T_{\bzeta_i}(\bm x,\Y)$ and 
$T_{\bzeta_j}(\bm x,\Y)$ are disjoint for every $\bzeta_i \ne \bzeta_j$. In fact, $\Y$ denotes the support
of the distribution $g(\cdot)$ from which $\e$ is simulated.

\subsection{Illustration of TTMCMC with a simple example}
\label{subsec:ttmcmc_illustration}

Let us now illustrate the main idea of TTMCMC informally using the additive transformation. 
Although the example we illustrate TTMCMC with is borrowed from \ctn{Das14}, the algorithm we now present
is somewhat different from that of \ctn{Das14}. 
Assume that the current state is $\bm x=(x_1,x_2)\in\mathbb R^2$. We first randomly select
$u=(u_1,u_2,u_3)\sim Multinomial (w_b,w_d,w_{nc})$, where $w_b,w_d,w_{nc}~(>0)$ such that $w_b+w_d+w_{nc}=1$ are
the probabilities of birth, death, and no-change moves, respectively. That is, if $u_1=1$, then we increase
the dimensionality from 2 to 3; if $u_2=1$, then we decrease the dimensionality from 2 to 1, and if $u_3=1$, then
we keep the dimensionality unchanged. 
In the latter case, when the dimensionality is unchanged, the acceptance probability remains the same as in
TMCMC, as provided in Algorithm 3.1 of \ctn{Dutta14}. 

If $u_1=1$, we can increase the dimensionality by first selecting one of $x_1$ and $x_2$
with probability $1/2$ -- assuming for clarity that $x_1$ has been selected, we then
construct the move-type $T_{b,\bzeta}(\bm x,\e)=(x_1+a_1\e,x_1-a_1\e,x_2+\z_2a_2\e)$
$=(g_{1,\z_1=1}(x_1,\e),~g_{1,\z^c_1=-1}(x_1,\e),~g_{2,\z_2}(x_2,\e))$, say. 
Here, as in TMCMC, we draw $\e\sim g(\cdot)$, where $g(\cdot)$ is supported on the positive part
of the real line, and draw 
$\z_2=1$ with probability $p_2$ 
and $\z_2=-1$ with probability
$1-p_2$. Note that the value $\z_2=0$ is redundant for additive transformation (see \ctn{Dutta14} for the details) 
and so is omitted here. 
We re-label $\bm x'=T_{b,\bzeta}(\bm x,\e)=(x_1+a_1\e,x_1-a_1\e,x_2+\z_2a_2\e)$ as $(x'_1,x'_2,x'_3)$.
Thus, $T_{b,\bzeta}(\bm x,\e)$ increases the dimension from 2 to 3. 

We accept this birth move with probability
\begin{align}
a_b(\bm x,\e)
&=\min\left\{1, \frac{w_d}{w_b}\times\frac{p^{I_{\{1\}}(\z^c_2)}_2q^{I_{\{-1\}}(\z^c_2)}_2}
{p^{I_{\{1\}}(\z_2)}_2q^{I_{\{-1\}}(\z_2)}_2}\right.\notag\\
&\left.\quad\quad\quad\quad\times\frac{\pi(x_1+a_1\e,~x_1-a_1\e,~x_2+\z_2a_2\e)}{\pi(x_1,x_2)}\times
\left|\frac{\partial (T_{b,\bzeta}(\bm x, \e))}{\partial(\bm x, \e)}\right|
\right\}.
\label{eq:acc_birth}
\end{align}
In (\ref{eq:acc_birth}),
\begin{align}
& \left|\frac{\partial (T_{b,\bzeta}(\bm x, \e))}{\partial(\bm x, \e)}\right|
 =\left|\frac{\partial (x_1+a_1\e,x_1-a_1\e,x_2+\z_2a_2\e)}{\partial(x_1,x_2, \e)}\right|
 =\left|\left(\begin{array}{ccc}
1 & 1 & 0 \\
0 & 0 & 1 \\
a_1 & -a_1 & \z_2a_2 \\
 \end{array}
 \right)\right|=2a_1.
 \label{eq:Jacobian_birth}
\end{align}

Now let us illustrate the problem of returning to $\bm =(x_1,x_2)~(\in\mathbb R^2)$ from 
$T_{b,\bzeta}(\bm x,\e)=(x_1+a_1\e,x_1-a_1\e,x_2+\z_2a_2\e)~(\in\mathbb R^3)$.
For our purpose, in this paper, we select one of the first two elements of $T_{b,\bzeta}(\bm x,\e)$ 
with the same probability. Suppose that we select $x_1+a_1\e$ with probability $1/2$. We then deterministically 
choose its right-adjacent $x_1-a_1\e$, and form the average $x^*_1=((x_1+a_1\e)+(x_1-a_1\e))/2=x_1$.
For non-additive transformations we can consider the averages of the backward moves of 
of the selected element and its right-adjacent. Even in this additive transformation example, 
after simulating $\e$ as before we can consider 
the respective backward moves of $x_1+a_1\e$ and $x_1-a_1\e$, both yielding $x_1$, and then take
the average denoted by $x^*_1$. 
For the remaining element $x_2+\z_2a_2\e$, we need to simulate
$\z^c_2$ and then consider the move $(x_2+\z_2a_2\e)+\z^c_2a_2\e=x_2$. Thus, we can return to $(x_1,x_2)$ using
this strategy.

Letting $\bm x'=(x'_1,x'_2,x'_3)=(x_1+a_1\e,x_1-a_1\e,x_2+\z_2a_2\e)$, and
denoting the average
involving the first two elements by $x^*_1$, 
the death move is then given by $\bm x''=T_{d,\bzeta}(\bm x',\e)=(x^*_1,x'_3+\z^c_2a_2\e)$
$=(\frac{x'_1+x'_2}{2},x'_3+\z^c_2a_2\e)$.
Now observe that for returning to $(x'_1,x'_2)$ from $x^*_1$,
we must have $x^*+a_1\e^*=x'_1$ and $x^*-a_1\e^*=x'_2$, which yield $\e^*=(x'_1-x'_2)/2a_1$.
Hence, the Jacobian associated with the death move in this case is given by
\begin{align}
 \left|\frac{\partial \left(T_{d,\bzeta}(\bm x', \e),\e^*,\e\right)}{\partial(\bm x', \e)}\right|
 = \left|\frac{\partial \left(\frac{x'_1+x'_2}{2},x'_3+\z^c_2a_2\e,\frac{x'_1-x'_2}{2a_1},\e\right)}
 {\partial(x'_1,x'_2,x'_3, \e)}\right|
 &=\left|\left(\begin{array}{cccc}
\frac{1}{2} & 0  & \frac{1}{2a_1} & 0\\
\frac{1}{2} & 0 & -\frac{1}{2a_1} & 0 \\
0 & 1 & 0 & 0 \\
0 & \z^c_2a_2 & 0 & 1\\
 \end{array}
 \right)\right|=\frac{1}{2a_1}.
 \label{eq:Jacobian_death}
\end{align}
We accept this death move with probability
 \begin{align}
 a_d(\bm x'', \e, \e^*) &= 
 \min\left\{1,\frac{w_b}{w_d}\times\dfrac{P(\bzeta^c)}{P(\bzeta)} ~\dfrac{\pi(\bm x'')}{\pi(\supr{\bm x'}{t})} 
 ~\left|\frac{\partial (T_{d,\bzeta}(\bm x', \e),\e^*,\e)}{\partial(\bm x',\e)}\right| \right\}\notag\\
 &= \min\left\{1,3\times\frac{w_b}{w_d}\times\dfrac{p^{I_{\{1\}}(\z^c_2)}_2q^{I_{\{-1\}}(\z^c_2)}_2}
 {p^{I_{\{1\}}(\z_2)}_2q^{I_{\{-1\}}(\z_2)}_2} \times\dfrac{\pi(\bm x'')}{\pi(\bm x')} 
 \times\frac{1}{2a_1}\right\}.\notag
 \end{align}

In general, $\bx\in\mathbb R^{mk}$ may be of the form $(\bx_1,\bx_2,\ldots,\bx_m)$, where 
$\bx_{\ell}=(x_{\ell,1},x_{\ell,2},\ldots,x_{\ell,k})$ for $\ell=1,2,\ldots,m$, where $m\geq 1$ is an integer. 
Let us assume that 
if the dimension of any 
one $\bx_{\ell}$ is changed, then the dimensions of all other $\bx_{\ell'};~\ell'\neq \ell$ must also change accordingly. 
For instance, in our model,where we have summands  with unknown
number of components and the $i$-th component is characterized by the parameters associated with ODDP
 $(\btheta_{1i},\btheta_{2i},V_i,z_i)$, when the dimension of the current $k$-dimensional
 vector of the location parameter of the central distribution  $(\theta_{11},\ldots,\theta_{1k})$ 
 is increased by one, then one must simultaneously increase the dimension
of the other set of the current $k$-dimensional location parameter $(\theta_{21},\ldots,\theta_{2k})$, 
the $k$-dimensional vector of the associated point process $(z_1,\ldots,z_k)$, as well as the $k$-dimensional mass vector 
$(V_1,\ldots,V_k)$ by one. In Section \ref{subsec:ttmcmc3} 
we present a TTMCMC algorithm (Algorithm \ref{algo:ttmcmc3}) for situations of this kind, and show that detailed balance holds 
(irreducibility and aperiodicity hold by the same arguments provided in \ctn{Das14}). It is worth mentioning
that although \ctn{Das14} provide a TTMCMC algorithm for these situations (Algorithm 5.1 of their paper), 
their algorithm is somewhat different
from ours in that, for the death move, we select only one element randomly; then we choose the right-adjacent element;
take backward transformations of both of them, finally taking the average. On the other hand, \ctn{Das14} select
two elements randomly without replacement. This difference between the algorithm is reflected in the acceptance ratios -- 
our algorithm is slightly simpler in that the random selection probabilities do not appear in our acceptance ratio,
unlike that of \ctn{Das14}.

\subsection{General TTMCMC algorithm for jumping more than one dimensions at a time when several sets of
parameters are related}
\label{subsec:ttmcmc3}

\begin{algo}\label{algo:ttmcmc3} \topline General TTMCMC algorithm for jumping $m$ dimensions
with $m$ related sets of co-ordinates.
\botline \normalfont \ttfamily
\begin{itemize}
 \item Let the initial value be $\supr{\bm x}{0}\in\mathbb R^{mk}$, where $k\geq m$. 
 \item For $t=0,1,2,\ldots$
\begin{enumerate}
 \item Generate $u=(u_1,u_2,u_3)\sim Multinomial (1;w_{b,k},w_{d,k},w_{nc,k})$.
 \item If $u_1=1$ (increase dimension from $mk$ to $(m+1)k$), then
 \begin{enumerate}
 \item Randomly select one co-ordinate from $\bm x^{(t)}_1=(x^{(t)}_{11},\ldots,x^{(t)}_{1k})$ without replacement.
 Let $j$ denote the chosen co-ordinate.
 \item Generate $\be_m=(\e_1,\ldots,\e_m)\stackrel{iid}{\sim} g(\cdot)$ and for $i\in\{1,\ldots,k\}\backslash\{j\}$ simulate  
 $\z_{\ell,i} \sim Multinomial(1;p_{\ell,i},q_{\ell,i},1-p_{\ell,i}-q_{\ell,i})$ independently, for every $\ell=1,\ldots,m$. 
 \item Propose the birth move as follows: for each $\ell=1,\ldots,m$, 
 apply the transformation $x^{(t)}_{\ell,i}\rightarrow g_{i,\z_{\ell,i}}(x^{(t)}_{\ell,i},\e_1)$ 
 for $i\in\{1,\ldots,k\}\backslash\{j\}$ and, for each $\ell\in \{1,\ldots,m\}$,
 split  $x^{(t)}_{\ell,j}$ into $g_{\ell,\z_{\ell,j}=1}(x^{(t)}_{\ell,j},\e_{\ell})$ and 
 $g_{\ell,\z^c_{\ell,j}=-1}(x^{(t)}_{\ell,j},\e_{\ell})$. In other words, let $\bx'=
 T_{b,\bzeta}(\supr{\bm x}{t}, \be_m)=(\bx'_1,\ldots,\bx'_m)$
 denote the complete birth move, where, for $\ell=1,\ldots,m$, $\bx'_{\ell}$ is given by
 \begin{align} 
 \bm x'_{\ell} & 
 =(g_{\ell,\z_{\ell,1}}(x^{(t)}_{\ell,1},\e_1),\ldots,
 g_{j-1,\z_{\ell,j-1}}(x^{(t)}_{\ell,j-1},\e_1),\notag\\
& g_{j,{\z_{\ell,j}=1}}(x^{(t)}_{\ell,j},\e_{\ell}),g_{j,{\z^c_{\ell,j}=-1}}(x^{(t)}_{\ell,j},\e_{\ell}),
g_{j+1,\z_{\ell,j+1}}(x^{(t)}_{\ell,j+1},\e_1),\ldots,
\notag\\
&\ldots,
g_{k,\z_{\ell,k}}(x^{(t)}_{\ell,k},\e_1)).\notag
\end{align}
Re-label the $k+1$ elements of $\bm x'_{\ell}$ as $(x'_{\ell,1},x'_{\ell,2},\ldots,x'_{\ell,k+1})$.
\item Calculate the acceptance probability of the birth move $\bm x'$:
 \begin{align}
 a_b(\supr{\bm x}{t}, \be_m) &= 
 \min\left\{1, \frac{w_{d,k+1}}{w_{b,k}}\times\dfrac{P_{(j)}(\bzeta^c)}{P_{(j)}(\bzeta)} 
 ~\dfrac{\pi(\bm x')}{\pi(\supr{\bm x}{t})} 
 ~\left|\frac{\partial (T_{b,\bzeta}(\supr{\bm x}{t}, \be_m))}{\partial(\supr{\bm x}{t}, \be_m)}\right| \right\},\notag
 \end{align}
where 
 \[ 
P_{(j)}(\bzeta)=\prod_{\ell=1}^m\prod_{i\in\{1,\ldots,k\}\backslash\{j\}}
p^{I_{\{1\}}(\z_{\ell,i})}_{\ell,i}q^{I_{\{-1\}}(\z_{\ell,i})}_{\ell,i},
 \]
 and
 \[ 
P_{(j)}(\bzeta^c)=\prod_{\ell=1}^m\prod_{i\in\{1,\ldots,k\}\backslash\{j\}}
p^{I_{\{1\}(\z^c_{\ell,i})}}_{\ell,i}q^{I_{\{-1\}}(\z^c_{\ell,i})}_{\ell,i}.
 \]

\item Set \[ \supr{\bm x}{t+1}= \left\{\begin{array}{ccc}
 \bm x' & \mbox{ with probability } & a_b(\supr{\bm x}{t},\be_m) \\
 \supr{\bm x}{t}& \mbox{ with probability } & 1 - a_b(\supr{\bm x}{t},\be_m).
\end{array}\right.\]
\end{enumerate}
\item If $u_2=1$ (decrease dimension from $k$ to $k-m$, for $k\geq 2m$), then
 \begin{enumerate}
 \item Generate $\be_m=(\e_1,\ldots,\e_m)\stackrel{iid}{\sim} g(\cdot)$. 
 \item Randomly select one co-ordinate (say, the $j$-th co-ordinate)  
 from $\bx_1=(x_{1,1},\ldots,x_{1,k-1})$. For $\ell=1,\ldots,m$, 
 let $$x^*_{\ell,j}=\left(g_{j,\z^c_{\ell,j}=-1}(x_{\ell,j},\e_{\ell})
 +g_{j',\z_{\ell,j+1}=1}(x_{\ell,j+1},\e_{\ell})\right)/2;$$ 
 replace the co-ordinate $x_{\ell,j}$
 by the average $x^*_{\ell,j}$ and delete $x_{\ell,j+1}$.
 \item Simulate $\bzeta$ by generating independently, for $\ell=1,\ldots,m$ and for $i\in\{1,\ldots,k\}\backslash\{j,j+1\}$, 
 $\z_{\ell,i} \sim Multinomial(1;p_{\ell,i},q_{\ell,i},1-p_{\ell,i}-q_{\ell,i})$.
 \item For $\ell=1,\ldots,m$ and for $i\in\{1,\ldots,k\}\backslash\{j,j+1\}$, apply the transformation 
 $x'_{\ell,i}=g_{i,\z_{\ell,i}}(x^{(t)}_{\ell,i},\e_1)$.  
 \item Propose the following death move $\bx'=T_{d,\bzeta}(\supr{\bm x}{t}, \be_m)=(\bx'_1,\ldots,\bx'_m)$ 
 where for $\ell=1,\ldots,m$, $\bx_{\ell}$ is given by
 \begin{align} 
 \bm x'_{\ell} 
 &=(g_{1,\z_{\ell,1}}(x^{(t)}_{\ell,1},\e_1),\ldots,
 g_{j-1,\z_{\ell,j-1}}(x^{(t)}_{\ell,j-1},\e_1),x^*_{\ell,j},g_{j+1,\z_{\ell,j+1}}(x^{(t)}_{\ell,j+1},\e_1),\notag\\
& \ldots,g_{k,\z_{\ell,k}}(x^{(t)}_{\ell,k},\e_1)).\notag
\end{align}
Re-label the elements of $\bm x'_{\ell}$ as $(x'_{\ell,1},x'_{\ell,2},\ldots,x'_{\ell,k-1})$.
 \item For $\ell=1,\ldots,m$, solve for $\e^*_{\ell}$ from the equations 
 $g_{\ell,\z_{\ell,j}=1}(x^*_{\ell,j},\e^*_{\ell})=x_{\ell,j}$ and 
 $g_{\ell,\z^c_{\ell,j}=-1}(x^*_{\ell,j},\e^*_{\ell})=x_{\ell,j+1}$
and express $\e^*_{\ell}$ in terms of $x_{\ell,j}$ and $x_{\ell,j+1}$.
Let $\be^*_m=(\e^*_1,\ldots,\e^*_m)$.
\item Calculate the acceptance probability of the death move:
 \begin{align}
 a_d(\supr{\bm x}{t}, \be_m, \be^*_m) &= 
 \min\left\{1,\frac{w_{b,k-m}}{w_{d,k}}\times\dfrac{P_{(j,j+1)}(\bzeta^c)}{P_{(j,j+1)}(\bzeta)} 
 ~\dfrac{\pi(\bm x')}{\pi(\supr{\bm x}{t})} 
 ~\left|\frac{\partial (T_{d,\bzeta}(\supr{\bm x}{t}, \be_m),\be^*_m,\be_m)}
 {\partial(\supr{\bm x}{t},\be_m)}\right| \right\},\notag
 \end{align}
 where
 \[ 
P_{(j,j+1)}(\bzeta)=\prod_{\ell=1}^m\prod_{i\in\{1,\ldots,k\}\backslash\{j,j+1\}}p^{I_{\{1\}}(\z_{\ell,i})}_{\ell,i}
q^{I_{\{-1\}}(\z_{\ell,i})}_{\ell,i},
 \]
 and
 \[ 
P_{(j,j+1)}(\bzeta^c)=\prod_{\ell=1}^m\prod_{i\in\{1,\ldots,k\}\backslash\{j,j+1\}}
p^{I_{\{1\}(\z^c_{\ell,i})}}_{\ell,i}q^{I_{\{-1\}}(\z^c_{\ell,i})}_{\ell,i}.
 \]
\item Set \[ \supr{\bm x}{t+1}= \left\{\begin{array}{ccc}
 \bm x' & \mbox{ with probability } & a_d(\supr{\bm x}{t},\be_m,\be^*_m) \\
 \supr{\bm x}{t}& \mbox{ with probability } & 1 - a_d(\supr{\bm x}{t},\be_m,\be^*_m).
\end{array}\right.\]
 \end{enumerate}
\item If $u_3=1$ (dimension remains unchanged), 
then implement steps (1), (2), (3) of Algorithm 3.1 
of \ctn{Dutta14}. 
\end{enumerate}
\item End for
\end{itemize}
\botline \rmfamily
\end{algo}

\subsection{Detailed balance}
\label{subsec:detailed_balance2}

To see that detailed balance is satisfied for the birth and death moves, note that associated with the birth move, 
the probability of transition $\bm x~(\in\mathbb R^k)\mapsto T_{b,\bz}(\bm x,\be_m)~(\in\mathbb R^{k+m})$, with
$k\geq m$, is given by:
\begin{align}
&\pi(\bm x)\times 
\frac{1}{k}\times w_{b,k}\times \prod_{\ell=1}^mg(\e_{\ell})\times\prod_{\ell=1}^m\prod_{i\in\{1,\ldots,k\}\backslash\{j\}}
p^{I_{\{1\}}(\z_{\ell,i})}_{\ell,i}q^{I_{\{-1\}}(\z_{\ell,i})}_{\ell,i}\notag\\
&\times\min\left\{1, \frac{w_{d,k+m}}{w_{b,k}}\times
\frac{\prod_{\ell=1}^m\prod_{i\in\{1,\ldots,k\}\backslash\{j\}} 
p^{I_{\{1\}}(\z^c_{\ell,i})}_{\ell,i}q^{I_{\{-1\}}(\z^c_{\ell,i})}_{\ell,i}}
{\prod_{\ell=1}^m\prod_{i\in\{1,\ldots,k\}\backslash\{j\}}
p^{I_{\{1\}}(\z_{\ell,i})}_{\ell,i}q^{I_{\{-1\}}(\z_{\ell,i})}_{\ell,i}}\right.\notag\\
&\hspace{6cm}\left.
\times\frac{\pi(T_{b,\bzeta}(\bm x,\be_m))}{\pi(\bm x)}\times 
\left|\frac{\partial (T_{b,\bzeta}(\supr{\bm x}{t}, \be_m))}{\partial(\supr{\bm x}{t}, \be_m)}\right|\right\}\notag\\
&=\frac{1}{k}\times\prod_{i=1}^mg(\e_i)\times\min\left\{\pi(\bm x)\times w_{b,k}\times
\prod_{\ell=1}^m\prod_{i\in\{1,\ldots,k\}\backslash\{j\}}
p^{I_{\{1\}}(\z_{\ell,i})}_{\ell,i}q^{I_{\{-1\}}(\z_{\ell,i})}_{\ell,i},
\right.\notag\\
&\hspace{1cm}\left.\times w_{d,k+m}\times\prod_{\ell=1}^m\prod_{i\in\{1,\ldots,k\}\backslash\{j\}} 
p^{I_{\{1\}}(\z^c_{\ell,i})}_{\ell,i}q^{I_{\{-1\}}(\z^c_{\ell,i})}_{\ell,i}
\pi(T_{b,\bzeta}(\bm x,\be_m))\times \left|\frac{\partial (T_{b,\bzeta}(\supr{\bm x}{t}, \be_m))}{\partial(\supr{\bm x}{t}, \be_m)}\right|
\right\}.
\label{eq:db_birth2}
\end{align}
%


The transition probability of the reverse death move is given by:
\begin{align}
&\pi(\bm x)\times w_{d,k+m}\times \frac{1}{k}\times\prod_{\ell=1}^mg(\e_{\ell})\times
\prod_{\ell=1}^m\prod_{i\in\{1,\ldots,k\}\backslash\{j\}} 
p^{I_{\{1\}}(\z^c_{\ell,i})}_{\ell,i}q^{I_{\{-1\}}(\z^c_{\ell,i})}_{\ell,i}
\times \left|\frac{\partial (T^{-1}_{d,\bzeta}(\supr{\bm x}{t}, \be_m),\be^*_m,\be_m)}{\partial(\supr{\bm x}{t}, \be_m)}\right|
\notag\\
&\times\min\left\{1,\frac{w_{b,k}}{w_{d,k+m}}
\times\frac{\prod_{\ell=1}^m\prod_{i\in\{1,\ldots,k\}\backslash\{j\}}
p^{I_{\{1\}}(\z_{\ell,i})}_{\ell,i}q^{I_{\{-1\}}(\z_{\ell,i})}_{\ell,i}}
{\prod_{\ell=1}^m\prod_{i\in\{1,\ldots,k\}\backslash\{j\}} 
p^{I_{\{1\}}(\z^c_{\ell,i})}_{\ell,i}q^{I_{\{-1\}}(\z^c_{\ell,i})}_{\ell,i}}\right.\notag\\
&\hspace{4cm}\left.
\times\frac{\pi(\bm x)}{\pi(T_{b,\bzeta}(\bm x,\be_m))}\times
\left|\frac{\partial (T_{d,\bzeta}(\supr{\bm x}{t}, \be_m),\be^*_m,\be_m)}{\partial(\supr{\bm x}{t}, \be_m)}\right|\right\}
\notag\\
&=\frac{1}{k}\times\prod_{\ell=1}^m g(\e_{\ell})\times\min\left\{\pi(T_{b,\bzeta}(\bm x,\be_m))\times w_{d,k+m}\times
\prod_{\ell=1}^m\prod_{i\in\{1,\ldots,k\}\backslash\{j\}} 
p^{I_{\{1\}}(\z^c_{\ell,i})}_{\ell,i}q^{I_{\{-1\}}(\z^c_{\ell,i})}_{\ell,i}\right.\notag\\
&\left.\times \left|\frac{\partial (T^{-1}_{d,\bzeta}(\supr{\bm x}{t}, \be_m),\be^*_m,\be_m)}
{\partial(\supr{\bm x}{t}, \be_m)}\right|,
w_{b,k}\times\prod_{\ell=1}^m\prod_{i\in\{1,\ldots,k\}\backslash\{j\}}
p^{I_{\{1\}}(\z_{\ell,i})}_{\ell,i}q^{I_{\{-1\}}(\z_{\ell,i})}_{\ell,i}\times\pi(\bm x)
\right\}.
\label{eq:db_death2}
\end{align}
Noting that $\left|\frac{\partial (T^{-1}_{d,\bzeta}(\supr{\bm x}{t}, \be_m),\be^*_m,\be_m)}{\partial(\supr{\bm x}{t}, \be^*_m,\be_m)}\right|
=\left|\frac{\partial (T_{b,\bzeta}(\supr{\bm x}{t}, \be_m))}{\partial(\supr{\bm x}{t}, \be_m)}\right|$, it follows that
(\ref{eq:db_birth2}) = (\ref{eq:db_death2}), showing that detailed balance holds for the birth and the death moves.

\section{TTMCMC algorithm for our spatio-temporal model}
\label{sec:ttmcmc4}
We now specialize the general TTMCMC algorithm (Algorithm \ref{algo:ttmcmc3}) provided in Section \ref{subsec:ttmcmc3} 
in our spatio-temporal context. For our spatio temporal model, we need to update the variable dimensional mass parameter 
$\bV=(V_1,\ldots,V_k)$, the point process variables $\bz=(z_1,\ldots,z_k)$, 
location parameters $\btheta_1=(\theta_{11},\ldots,\theta_{1k})$,
the other set of location parameters $\btheta_2=(\theta_{21},\ldots,\theta_{2k})$, and fixed dimensional parameters $\alpha$, 
$\lambda$, error variance $\sigma$ and the parameters related to the kernel, namely, $\varphi,a_{\delta}$, $b_{\psi}$, 
$\left(\psi_{1}(\bs_1),\ldots,\psi_{1}(\bs_n)\right)$, $\left(\psi_{2}(\bs_1),\ldots,\psi_{2}(\bs_n)\right)$, 
$\left(\delta(t_1),\ldots,\delta(t_n)\right)$, and $\tau$. For updating the variable dimensional parameters 
we use  proposed TTMCMC algorithm, and for fixed dimension we use the TMCMC algorithm of \ctn{Dutta14}. 
We denote by $\bxi=(\bV,\bz,\btheta_1,\btheta_2)$ the collection of all variable dimensional parameters and 
by
$\boeta=(\varphi,a_{\delta},b_{\psi},\psi_{1}(\bs_1),\ldots,\psi_{1}(\bs_n),\\\psi_{2}(\bs_1),\ldots,\psi_{2}(\bs_n),
\delta(t_1),\ldots,\delta(t_n),\tau,\alpha,\lambda)$, the collection
of all fixed dimensional parameters. 
The detailed updating procedure is provided as Algorithm \ref{algo:ttmcmc4}. 

\begin{algo}\label{algo:ttmcmc4} \topline Detailed updating procedure of our spatio-temporal model
\botline \normalfont \ttfamily
\begin{itemize}
\item Initialise the number of components $k$; let $k^{(0)}$ be the chosen initial value 
(we chose $k^{(0)}=15$ as the initial value for our applications).
\item 
Given $k=k^{(0)}$, let $\bxi^{(0)}$ denote the initial value of $\bxi$. Also, let $\boeta^{(0)}$ denote
the initial value of $\boeta$.
\item Since $\bz$ and $\bV$ are constrained random variables, we consider updating 
the reparameterized versions $\bV^{*}=\log(\bV)$ 
and $\bz^{*}=\log\left(\frac{\bz-a}{b-a}\right)$. After every iteration we invert the transformations to
store the original variables $\bV$ and $\bz$.
For the sake of convenience of presentation of our algorithm we slightly abuse notation by referring to
$\bV^*$ and $\bz^*$ as $\bV$ and $\bz$ respectively.
\item For $t=0,1,2,\ldots$
\begin{enumerate}
 \item Generate $u=(u_1,u_2,u_3)\sim Multinomial \left(1;\frac{1}{3},\frac{1}{3},\frac{1}{3}\right)$.
 \item If $u_1=1$ (increase dimension from $k$ to $k+1$ for each of the variables $\bV,\bz,\btheta_1,\btheta_2$), then
 \begin{enumerate}
 \item Randomly select one co-ordinate from $\{1,\ldots,k\}$.
 Let $j$ denote the chosen co-ordinate.
 \item  Generate $\be_5=(\e_1,\ldots,\e_5)\stackrel{iid}{\sim} N(0,1)\mathbb I_{\{\e>0\}}$ 
 ($\mathbb I_{\{\e>0\}}$ denoting the indicator function). 
 For updating the variable dimensional parameters,  simulate  
 \[ \zeta^{(1)}_{\ell,i} = \left\{ \begin{array}{ccc}
                  1 &\mbox{ w.p. } & \frac{1}{2} \\ 
		  -1 & \mbox{ w.p. } & \frac{1}{2} 
                 \end{array}\right\}
\mbox{for } i\in\{1,\ldots,k\}\backslash\{j\}\\
\mbox{ and }~ \ell=1,\ldots,4,\] 
and for updating the fixed one dimensional parameters, simulate 
\[ \zeta^{(2)}_{\ell} = \left\{ \begin{array}{ccc}
                  1 &\mbox{ w.p. } & \frac{1}{2} \\ 
		  -1 & \mbox{ w.p. } & \frac{1}{2} 
                 \end{array}\right\}
\mbox{ for}~ \ell=5,\ldots,10.\]

For updating fixed multi-dimensional parameters, simulate
\[ \zeta^{(3)}_{\ell,i} = \left\{ \begin{array}{ccc}
                  1 &\mbox{ w.p. } & \frac{1}{2} \\ 
		  -1 & \mbox{ w.p. } & \frac{1}{2} 
                 \end{array}\right\}\
\mbox{for }~ i\in\{1,\ldots,n\}\\
\mbox{ and}~ \ell=11,12,13,14.\]

 \item Propose the birth move as follows. 
 For $i\in\{1,\ldots,k\}\backslash\{j\}$, apply the additive transformation:\\
 $V^{(t)}_{i}\rightarrow (V^{(t)}_{i}+\zeta^{(1)}_{1,i}a_{1}\e_{1})$\\
$z^{(t)}_{i}\rightarrow (z^{(t)}_{i}+\zeta^{(1)}_{2,i}a_{2}\e_{2})$\\
$\theta^{(t)}_{1i}\rightarrow (\theta^{(t)}_{1i}+\zeta^{(1)}_{3,i}a_{3}\e_{3})$\\
$\theta^{(t)}_{2i}\rightarrow (\theta^{(t)}_{2i}+\zeta^{(1)}_{4,i}a_{4}\e_{4})$\\
 and split:\\
 $V^{(t)}_{j}$ into $(V^{(t)}_{j}+a_{1}\e_{1})$ and
 $(V^{(t)}_{j}-a_{1}\e_{1})$\\
$z^{(t)}_{j}$ into $(z^{(t)}_{j}+a_{2}\e_{2})$ and 
 $(z^{(t)}_{j}-a_{2}\e_{2})$\\
$\theta^{(t)}_{1j}$ into $(\theta^{(t)}_{1j}+a_{3}\e_{3})$ and 
 $(\theta^{(t)}_{1j}-a_{3}\e_{3})$\\
$\theta^{(t)}_{2j}$ into $(\theta^{(t)}_{2j}+a_{4}\e_{4})$ and 
 $(\theta^{(t)}_{2j}-a_{4}\e_{4})$ \\
In other words, let $\bx'=
 T_{b,\zeta^{(1)}}(\supr{\bm x}{t}, \be_m)=(\bV',\bz',\btheta'_1,\btheta'_2)$
 denote the complete birth move, where,
 \begin{align} 
 \bm V' & 
 =((V^{(t)}_{1}+\zeta^{(1)}_{1,1}a_{1}\e_{1}) \ldots( V^{(t)}_{j-1}+\zeta^{(1)}_{1,j-1}a_{1}\e_{1}),
 (V^{(t)}_{j}+a_{1}\e_{1}),(V^{(t)}_{j}-a_{1}\e_{1})  
\notag\\
&\ldots (V^{(t)}_{k}+\zeta^{(1)}_{1,k}a_{1}\e_{1}))
\notag\\
\end{align}
\begin{align} 
 \bm z' & 
 =((z^{(t)}_{1}+\zeta^{(1)}_{2,1}a_{2}\e_{2}) \ldots( z^{(t)}_{j-1}+\zeta^{(1)}_{2,j-1}a_{2}\e_{2}),
 (z^{(t)}_{j}+a_{2}\e_{2}),(z^{(t)}_{j}-a_{2}\e_{2})  
\notag\\
&\ldots (z^{(t)}_{k}+\zeta^{(1)}_{2,k}a_{2}\e_{2}))
\notag\\
\end{align}
\begin{align} 
 \btheta'_{1} & 
 =((\theta^{(t)}_{1,1}+\zeta^{(1)}_{3,1}a_{3}\e_{3}) \ldots( \theta^{(t)}_{1,j-1}+\zeta^{(1)}_{3,j-1}a_{3}\e_{3}),
 (\theta^{(t)}_{1,j}+a_{3}\e_{3}),(\theta^{(t)}_{1,j}-a_{3}\e_{3})  
\notag\\
&\ldots (\theta^{(t)}_{1,k}+\zeta^{(1)}_{3,k}a_{3}\e_{3}))
\notag\\
\end{align}

\begin{align} 
 \btheta'_{2} & 
 =((\theta^{(t)}_{2,1}+\zeta^{(1)}_{4,1}a_{4}\e_{4}) \ldots( \theta^{(t)}_{2,j-1}+\zeta^{(1)}_{4,j-1}a_{4}\e_{4}),
 (\theta^{(t)}_{2,j}+a_{4}\e_{4}),(\theta^{(t)}_{2,j}-a_{4}\e_{4})  
\notag\\
&\ldots (\theta^{(t)}_{2,k}+\zeta^{(1)}_{4,k}a_{4}\e_{4}))
\notag\\
\end{align}

Re-label the $k+1$ elements of $\bm V'$ as $(V'_{1},V'_{2},\ldots,V'_{k+1})$,\\ 
$\bm z'$ as $(z'_{1},z'_{2},\ldots,z'_{k+1})$,
$\btheta'_{1}$ as $(\theta'_{1,1},\theta'_{1,2},\ldots,\theta'_{1,k+1})$, 
$\btheta'_{2}$ as $(\theta'_{2,1},\theta'_{2,2},\ldots,\theta'_{2,k+1})$.
\item 
 We apply the additive transformation based on the single $\e_5$ to update all the fixed dimensional parameter 
 $\boeta$ 
 as follows:

$\varphi^{(t)} \rightarrow ( \varphi^{(t)}+\zeta^{(2)}_{5}a_{5}\e_{5})$\\
$a^{(t)}_{\delta} \rightarrow ( a^{(t)}_{\delta}+\zeta^{(2)}_{6}a_{6}\e_{5})$\\
$b^{(t)}_{\psi} \rightarrow ( b^{(t)}_{\psi}+\zeta^{(2)}_{7}a_{7}\e_{5})$\\
$\alpha^{(t)} \rightarrow ( \alpha^{(t)}+\zeta^{(2)}_{8}a_{8}\e_{5})$\\
$\lambda^{(t)} \rightarrow ( \lambda^{(t)}+\zeta^{(2)}_{9}a_{9}\e_{5})$\\
$\tau^{(t)} \rightarrow ( \tau^{(t)}+\zeta^{(2)}_{10}a_{10}\e_{5})$\\
$\sigma^{(t)} \rightarrow ( \sigma^{(t)}+\zeta^{(2)}_{11}a_{11}\e_{5})$\\
$\psi^{(t)}_{1}(\bs_i) \rightarrow ( \psi^{(t)}_{1}(\bs_i)+\zeta^{(3)}_{12,i}a_{12}\e_{5})$\\
$\psi^{(t)}_{2}(\bs_i) \rightarrow ( \psi^{(t)}_{2}(\bs_i)+\zeta^{(3)}_{13,i}a_{13}\e_{5})$\\
$\delta^{(t)}(t_i) \rightarrow ( \delta^{(t)}(t_i)+\zeta^{(3)}_{14,i}a_{14}\e_{5})$\\

Let $\boeta'=
 T_{b,\zeta^{(2)}}(\supr{\boeta}{t},\e_5)=(\varphi',a'_{\delta},b'_{\psi},\psi'_{1}(\bs_1),\ldots,\psi'_{1}(\bs_n),
 \psi'_{2}(\bs_1),\ldots,\psi'_{2}(\bs_n),\\
 \delta'(t_1),\ldots,\delta'(t_n),\tau',\alpha',\lambda',\sigma')$
denote the complete move type for fixed dimensional parameters.\\

In the above transformations the $a_i$'s are the scaling constants to be chosen appropriately; see
\ctn{Das14} and \ctn{Dey14a} (see also \ctn{Dey14b}) for the details. 
In our applications we choose the scales on the basis of pilot runs of our TTMCMC algorithm. 

\item Calculate the acceptance probability:
 \begin{align}
 a_{b,\zeta}(\supr{\bxi}{t},\supr{\boeta}{t},\be_5) &= 
 \min\left\{1, \dfrac{\pi(\bxi',\boeta')}{\pi(\supr{\bxi}{t},\supr{\boeta}{t})} 
 ~\left|\frac{\partial (T_{b,\zeta^{(1)}}(\supr{\bxi}{t}, \be_{4}))}{\partial(\supr{\bxi}{t},\be_{4})}\right| 
~\left|\frac{\partial (T_{b,\zeta^{(2)}}(\supr{\boeta}{t}, \e_5))}{\partial(\supr{\boeta}{t},\e_5)}\right|\right\},\notag
 \end{align}
where 
 \[ 
\frac{\partial (T_{b,\zeta^{(1)}}(\supr{\bxi}{t}, \be_4))}{\partial(\supr{\bxi}{t},\be_4)}=2^4a_{1}a_{2}a_{3}a_{4}
 \]
 and
 \[ 
\frac{\partial (T_{b,\zeta^{(2)}}(\supr{\boeta}{t}, \e_5))}{\partial(\supr{\boeta}{t},e_5)}=1.
 \]

\item Set \[ (\supr{\bxi}{t+1},\supr{\boeta}{t+1})= \left\{\begin{array}{ccc}
 (\bxi',\boeta') & \mbox{ with probability } & a_{b,\zeta}(\supr{\bxi}{t},\supr{\boeta}{t},\be_5), \\
 (\supr{\bxi}{t},\supr{\boeta}{t})& \mbox{ with probability } & 1 - a_{b,\zeta}(\supr{\bxi}{t},\supr{\boeta}{t},\be_5).
\end{array}\right.\]
\end{enumerate}
\item If $u_2=1$ (decrease dimension from $k$ to $k-1$ for each of the variables $\bV,\bz,\btheta_1,\btheta_2$), then
 \begin{enumerate}
 \item Generate $\be_5=(\e_1,\ldots,\e_5)\stackrel{iid}{\sim} N(0,1)\mathbb I_{\{\e>0\}}$. 
 \item Randomly select one co-ordinate from $\{1,2,\ldots,k-1\}$. 
Let $j$ be the selected co-ordinate.
Then let $$V^*_{j}=\left((V_{j}+a_1\e_1) +(V_{j+1}-a_1\e_1)\right)/2;$$ 
replace the co-ordinate $V_{j}$
by the average $V^*_{j}$ and delete $V_{j+1}$. 
Similarly, let
$$z^*_{j}=\left((z_{j}+a_2\e_2) +(z_{j+1}-a_2\e_2)\right)/2;$$ 
 replace the co-ordinate $z_{j}$
 by the average $z^*_{j}$ and delete $z_{j+1}$. Form
$$\theta^*_{1,j}=\left((\theta_{1,j}+a_3\e_3) +(\theta_{1,j+1}-a_3\e_3)\right)/2;$$ 
and replace the co-ordinate $\theta_{1,j}$
 by the average $\theta^*_{1,j}$ and delete $\theta_{1,j+1}$; 
create
$$\theta^*_{2,j}=\left((\theta_{2,j}+a_4\e_4) +(\theta_{2,j+1}-a_4\e_4)\right)/2;$$ 
and replace the co-ordinate $\theta_{2,j}$
 by the average $\theta^*_{2,j}$ and delete $\theta_{2,j+1}$. 

 \item Simulate $\bm\zeta$ similarly as in the case of the birth move.
 \item For the co-ordinates other than $j$ and $j+1$ apply the additive transformation
$V^{(t)}_{i}\rightarrow (V^{(t)}_{i}+\zeta^{(1)}_{1,i}a_{1}\e_{1})$\\
$z^{(t)}_{i}\rightarrow (z^{(t)}_{i}+\zeta^{(1)}_{2,i}a_{2}\e_{2})$\\
$\theta^{(t)}_{1i}\rightarrow (\theta^{(t)}_{1i}+\zeta^{(1)}_{3,i}a_{3}\e_{3})$\\
$\theta^{(t)}_{2i}\rightarrow (\theta^{(t)}_{2i}+\zeta^{(1)}_{4,i}a_{4}\e_{4})$\\
 for $i\in\{1,\ldots,k\}\backslash\{j,j+1\}$.
 \item 
 In other words, let $\bxi'=
 T_{d,\zeta^{(1)}}(\supr{\bxi}{t}, \be_4)=(\bV',\bz',\btheta'_1,\btheta'_2)$
 denote the complete death move, where,
 \begin{align} 
 \bm V' & 
 =((V^{(t)}_{1}+\zeta^{(1)}_{1,1}a_{1}\e_{1}) \ldots( V^{(t)}_{j-1}+\zeta^{(1)}_{1,j-1}a_{1}\e_{1}),V^*_{j},
 ( V^{(t)}_{j+2}+\zeta^{(1)}_{1,j+2}a_{1}\e_{1}) 
\notag\\
&\ldots (V^{(t)}_{k}+\zeta^{(1)}_{1,k}a_{1}\e_{1}))
\notag\\
\end{align}
\begin{align} 
 \bm z' & 
 =((z^{(t)}_{1}+\zeta^{(1)}_{2,1}a_{2}\e_{2}) \ldots( z^{(t)}_{j-1}+\zeta^{(1)}_{2,j-1}a_{2}\e_{2}),
 z^*_{j},( z^{(t)}_{j+21}+\zeta^{(1)}_{2,j+21}a_{2}\e_{2}),
\notag\\
&\ldots (z^{(t)}_{k}+\zeta^{(1)}_{2,k}a_{2}\e_{2}))
\notag\\
\end{align}
\begin{align} 
 \btheta'_{1} & 
 =((\theta^{(t)}_{1,1}+\zeta^{(1)}_{3,1}a_{3}\e_{3}) \ldots( \theta^{(t)}_{1,j-1}+
 \zeta^{(1)}_{3,j-1}a_{3}\e_{3}),\theta^*_{1,j},( \theta^{(t)}_{1,j+2}+\zeta^{(1)}_{3,j+2}a_{3}\e_{3})  
\notag\\
&\ldots (\theta^{(t)}_{1,k}+\zeta^{(1)}_{3,k}a_{3}\e_{3}))
\notag\\
\end{align}

\begin{align} 
 \btheta'_{1} & 
 =((\theta^{(t)}_{2,1}+\zeta^{(1)}_{4,1}a_{4}\e_{4}) \ldots( \theta^{(t)}_{2,j-1}+\zeta^{(1)}_{4,j-1}a_{4}\e_{4}),
 \theta^*_{2,j},( \theta^{(t)}_{2,j+2}+\zeta^{(1)}_{4,j+2}a_{4}\e_{4})  
\notag\\
&\ldots (\theta^{(t)}_{2,k}+\zeta^{(1)}_{4,k}a_{4}\e_{4}))
\notag\\
\end{align}

Re-label the $k-1$ elements of $\bm V'$ as $(V'_{1},V'_{2},\ldots,V'_{k-1})$,\\ 
$\bm z'$ as $(z'_{1},z'_{2},\ldots,z'_{k-1})$,
$\btheta'_{1}$ as $(\theta'_{1,1},\theta'_{1,2},\ldots,\theta'_{1,k-1})$, and\\ 
$\btheta'_{2}$ as $(\theta'_{2,1},\theta'_{2,2},\ldots,\theta'_{2,k-1})$.

 \item Solve for $\e^*_1$ from the equations 
 $V^*_{j}+a_1\e^*_1=V_{j}$ and 
 $V^*_{j}-a_1\e^*_1=V_{j+1}$, which yield
$\e^*_1=\frac{(V_{j}-V_{j+1})}{2a_1}$.
Similarly, we have $\e^*_2=\frac{(z_{j}-z_{j+1})}{2a_2}$
$\e^*_3=\frac{(\theta_{1,j}-\theta_{1,j+1})}{2a_3}$ and
$\e^*_4=\frac{(\theta_{2,j}-\theta_{2,j+1})}{2a_4}$.
Let $\be^*_4=(\e^*_1,\ldots,\e^*_4)$.

\item
For updating the fixed dimensional parameters \\
$\boeta=(\varphi,a_{\delta},b_{\psi},\psi_{1}(\bs_1),\ldots,\psi_{1}(\bs_n),\psi_{2}(\bs_1),\ldots,\psi_{2}(\bs_n),
\delta(t_1),\ldots,\delta(t_n),\tau,\alpha,\lambda,\sigma)$\\
implement step $2~(d)$. 

\item Calculate the acceptance probability of the death move:
 \begin{align}
& a_{d,\zeta}(\supr{\bxi}{t},\supr{\boeta}{t} ,\be_5, \be^*_4)\notag\\  
&= \min\left\{1, 
 ~\dfrac{\pi(\bxi',\boeta')}{\pi(\supr{\bxi}{t},\supr{\boeta}{t})} 
 ~\left|\frac{\partial (T_{d,\bm\zeta^{(1)}}(\supr{\bxi}{t}, \be_{4}),\be^*_{4},\be_{4})}
 {\partial(\supr{\bxi}{t},\be_{4})}\right| 
~\left|\frac{\partial (T_{d,\bm\zeta^{(2)}}(\supr{\boeta}{t}, \e_5))}{\partial(\supr{\boeta}{t},\e_5)}\right|\right\},\notag
 \end{align}
 where
 \[
  \left|\frac{\partial (T_{d,\bm\zeta^{(1)}}(\supr{\bxi}{t}, \be_{4}),\be^*_{4},\be_{4})}
 {\partial(\supr{\bxi}{t},\be_{4})}\right| =\frac{1}{2^4a_1a_2a_3a_4}
 \]
and
\[ 
\frac{\partial (T_{d,\bm\zeta^{(2)}}(\supr{\boeta}{t}, \e_5))}{\partial(\supr{\boeta}{t},\e_5)}=1.
 \]

\item Set \[ (\supr{\bxi}{t+1},\supr{\boeta}{t+1})= \left\{\begin{array}{ccc}
 (\bxi',\boeta') & \mbox{ with probability } & a_{d,\zeta}(\supr{\bxi}{t},\supr{\boeta}{t},\be_5, \be^*_4), \\
 (\supr{\bxi}{t},\supr{\boeta}{t})& \mbox{ with probability } 
 & 1 - a_{d,\zeta}(\supr{\bxi}{t},\supr{\boeta}{t},\be_5,\be^*_4).
\end{array}\right.\]
 \end{enumerate}
\item If $u_3=1$ (dimension remains unchanged),\\
then update $(\supr{\bxi}{t},\supr{\boeta}{t})$  
by implementing steps (1), (2), (3) of Algorithm 3.1 
of \ctn{Dutta14}. 
\end{enumerate}
\item End for
\end{itemize}
\botline \rmfamily
\end{algo}

\section{Simulation study}
\subsection{Algorithm for generating the synthetic data}
\label{subsec:algo_simdata}

We have performed the following steps to simulate a non stationary $95\times 1$ vector:
\begin{enumerate}
\item We first take a grid of size 100.
\item We generate one random number $t_i$ from each interval $(i-1,i];$ $i=1,\ldots, 100$ as 100 time points. 
We store the time points in a vector which we denote by $\bi{t}=(t(1),\ldots,t(100))'$.
\item Next we generate 100 random points of the form $\{\bs_i=(s(1,i),s(2,i));~i=1,\ldots,100\}$ 
from $[0, 50]\times [0,50]$ as locations. 
We store the locations in a $100\times 2$ matrix $\bi{S}=(\bs'_1,\ldots,\bs'_{100})'$.
\item Then we randomly choose 5 time points from $\bi{t}$ and 5 locations from $\bi{S}$ and 
omit these random points from $\bi{t}$ and $\bi{S}$. So, we obtain a new time vector of length 95, say $\bi{t}_{95}$ 
and a new matrix of locations of order $95 \times 2$, say $\bi{S}_{95}$. We store the omitted time points 
in a separate vector, $\bi{t}_{5}$, and the locations in a separate matrix, $\bi{S}_{5}$, for future use. 
\item Next we calculate the covariance matrix $\bi{A}$ = $(\bA(i,j))$ of order $100\times 100$ based on $\bi{t}$ and 
$\bi{S}$, where $(i,j)$-th element of $\bi{A}$ is given by
\begin{align*}
{\small 
\bA(i,j) = \begin{cases}
 1 & \mbox{ if } i=j \\[1ex]
  \exp \left(-0.5\sqrt{(t(i)-t(j))^2+(s(1,i)-s(1,j))^2+(s(2,i)-s(2,j))^2}\right)
 & \mbox{ if } i \neq j
\end{cases}
}
\end{align*}
\item We partition the above covariance function $\bi{A}$ consisting of four component matrices 
$\bi{A}_{11}, \bi{A}_{12}, \bi{A}_{21}$ $\mbox{ and } \bi{A}_{22}$, where $\bi{A}_{11}$ is a $5\times 5$  
covariance matrix based on $\bi{t}_5$ and $\bi{S}_{5}$; $\bi{A}_{22}$ is a $95\times 95$ covariance matrix based on 
$\bi{t}_{95}$ and $\bi{S}_{95}$ (the form of the $(i,j)$-th element being the same as for the matrix $\bi{A}$, 
except now $\bi{t}$, $\bi{S}$ are replaced with $\bi{t}_{5}$ and $\bi{S}_{5}$ for $\bi{A}_{11}$; for $\bi{A}_{22}$, 
$\bi{t}$, $\bi{S}$  are replaced with $\bi{t}_{95}$ and $\bi{S}_{95}$); $\bi{A}_{12}$ = $\bi{A}_{21}^{T}$ is a $5\times 95$ 
matrix, containing the covariances between the deleted points and existing points. The $(i,j)$-th element 
of $\bi{A}_{12}$ is given by 
\[
{\small 
\bi{A}_{12}(i,j)=\left(-0.5\sqrt{({t}_{5}(i)-{t}_{95}(j))^2+({s}_{5}(1,i)-{s}_{95}(1,j))^2
+({s}_{5}(2,i)-{s}_{95}(2,j))^2}\right),
}
\]
for $i=1,\ldots 5$ and $j = 1, \ldots, 95$.
\item Next we generate one 5 dimensional random sample, $\bi{x}_{5}$, from a 5 variate normal distribution with mean function
\[
\bi{\mu}_{5}^{T} = \bi{D}_{5}\bi{\beta},
\]
and covariance matrix $\bi{A}_{11}$, where
$\bi{\beta}_{5}^{T}$ = $(0.1, 0.01, 0.02)$ and $\bi{D}_{5}$ = $(\bi{t}_{5}\vdots\bi{S}_5)$ is the design matrix. 
Note that $\bi{D}_{5}$ is a $5\times 3$ matrix.
\item Given $\bi{x}_{5}$ we simulate a $95\times 1$ random vector, $\bi{x}_{(95|5)}$ from a conditional 95 variate 
normal distribution with mean 
\[
\bi{\mu}_{(95|5)}^{T} = \bi{D}_{95}\bi{\beta} + (\bi{x}_{5}-\bmu_{5})^{T} \bi{A}_{11}^{-1} \bi{A}_{12}
\]
and covariance
\[
\bSigma_{(95|5)} = \bi{A}_{22} - \bi{A}_{21} \bi{A}_{11}^{-1} \bi{A}_{12},
\]
where $\bi{D}_{95}$ is obtained exactly same as $\bi{D}_{5}$, only $\bi{t}_{5}$ and $\bi{S}_{5}$ are replaced with 
$\bi{t}_{95}$ and $\bi{S}_{95}$, respectively.
\item The last step is to simulate a $95\times 1$ vector, $\bi{y}_{(95|\bi{x}_{(95|5)})}$, conditionally on 
$\bi{x}_{(95|5)}$. We simulate $\bi{y}_{(95|\bi{x}_{(95|5)})}$ from a 95 variate normal distribution with mean
\[
\bi{\mu}_{\bi{y}}^{T} = 0.01 (\bi{x}_{(95|5)})^{T}
\]
and covariance matrix
\[
\Sigma_{\bi{y}} = \begin{cases}
1 & \mbox{ if } i=j\\[1ex]
\mbox{exp}\left(-0.5\left(|{x}_{(95|5)}(i)-{x}_{(95|5)}(j)|\right)\right) & \mbox{ if } i\neq j,
\end{cases}
\]
for $i=1,\ldots 95$ and $j=1,\ldots 95$.  
\end{enumerate}

\begin{figure}[H]
\centering
\includegraphics[trim={0 0 0 0},clip, totalheight=0.25\textheight]{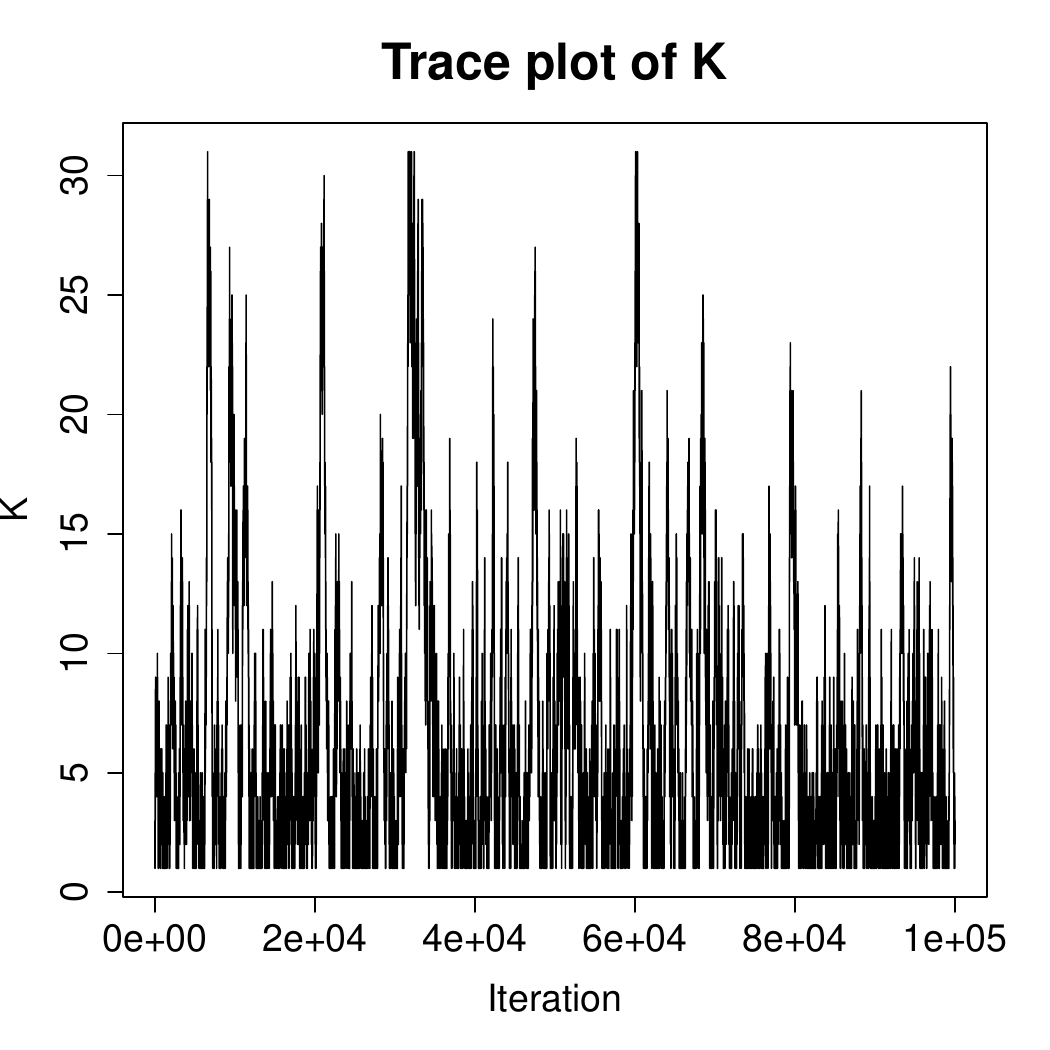}
\includegraphics[trim={0 0 0 0},clip, totalheight=0.25\textheight]{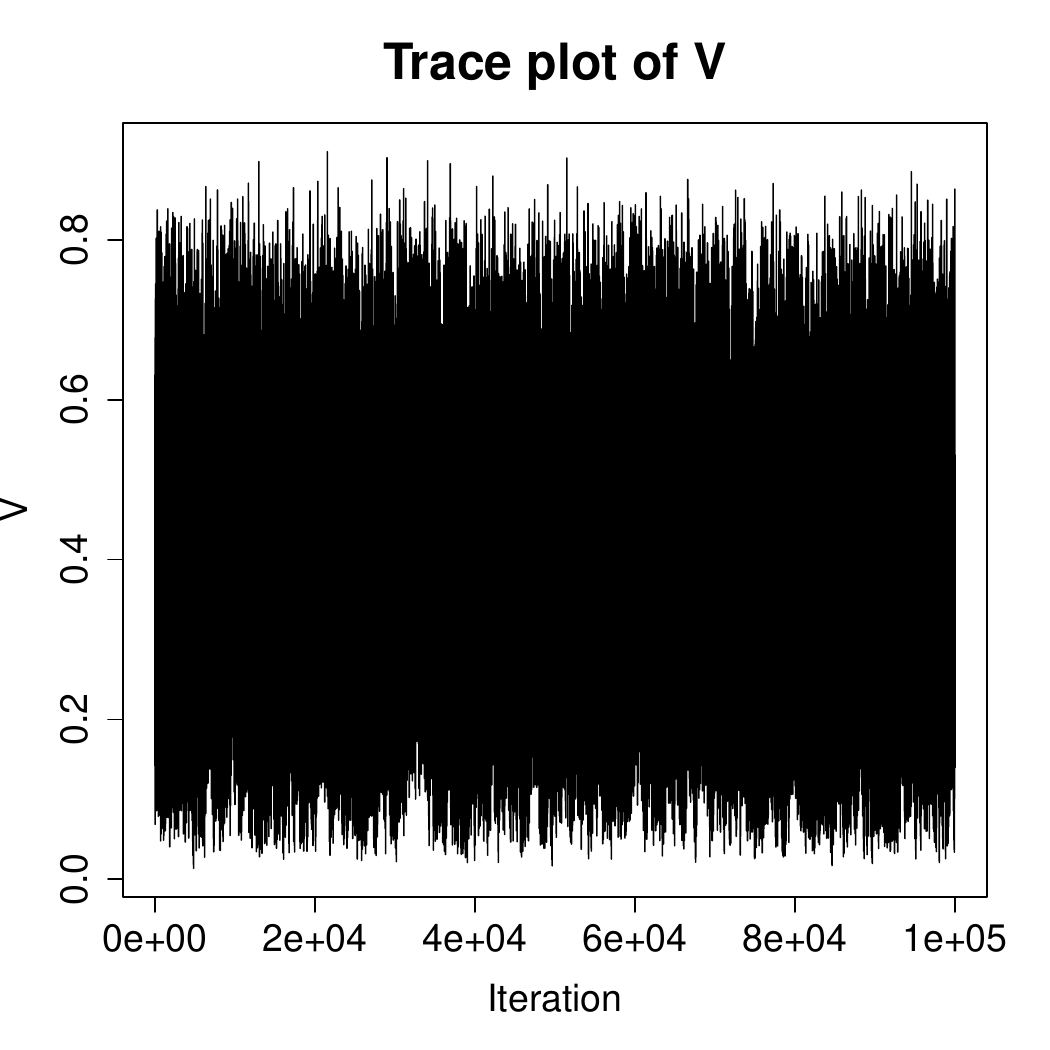}
\includegraphics[trim={0 0 0 0},clip, totalheight=0.25\textheight]{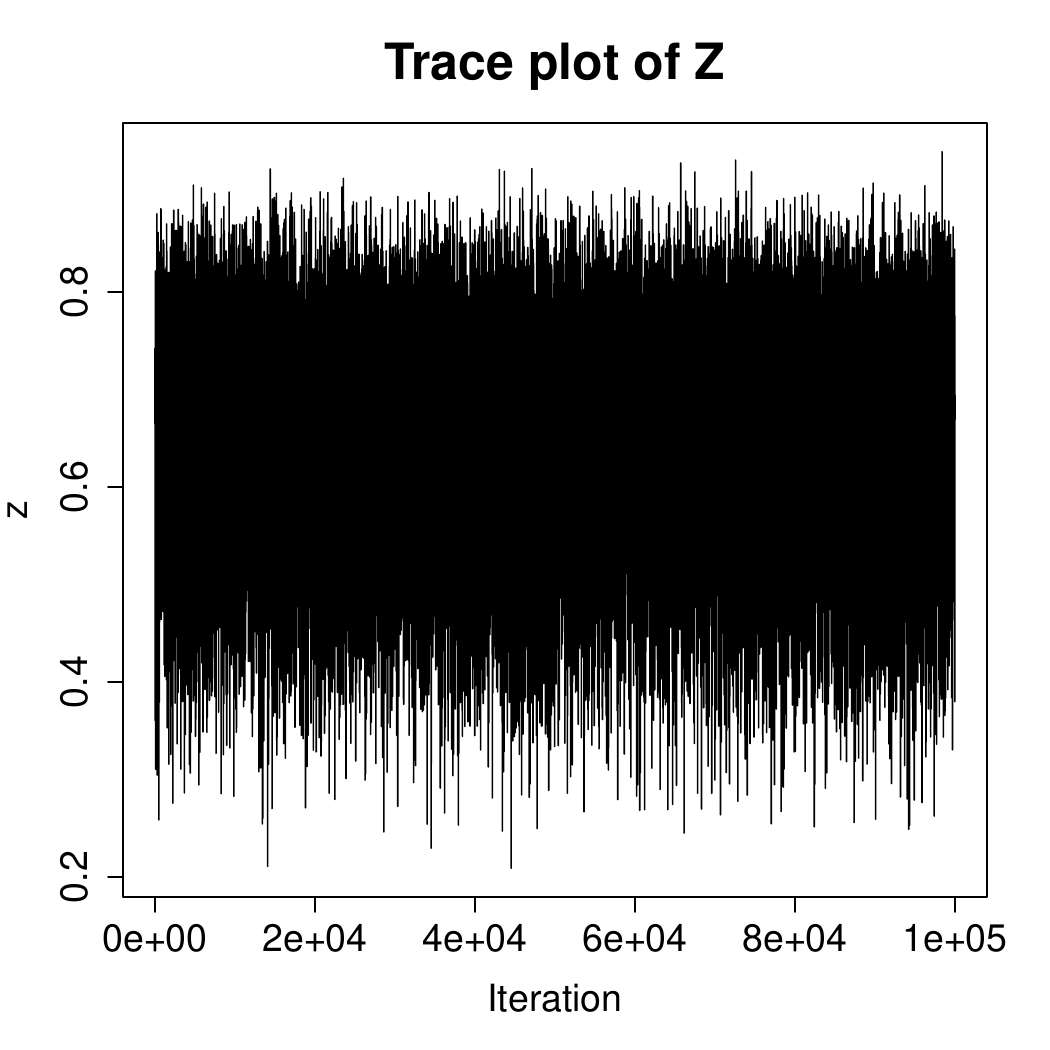}
\includegraphics[trim={0 0 0 0},clip, totalheight=0.25\textheight]{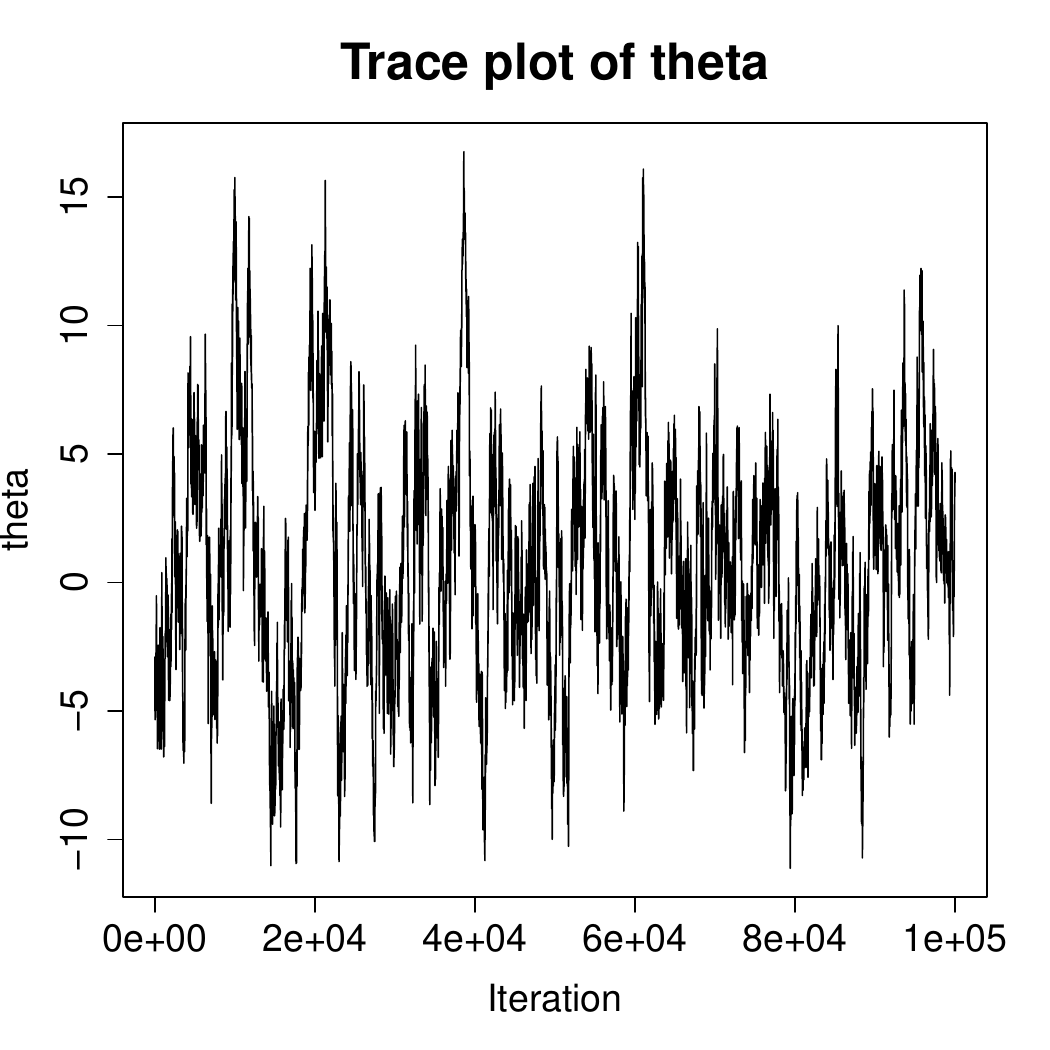}
\caption{{\bf Simulation study}: Traceplots of variable dimensional parameters.}
\label{fig:traceplotsimulation}
\end{figure}

\section{Real  data analysis}

\subsection{Spatial Data}
\begin{figure}[H]
\centering
\includegraphics[trim={0 0 0 0},clip, totalheight=0.25\textheight]{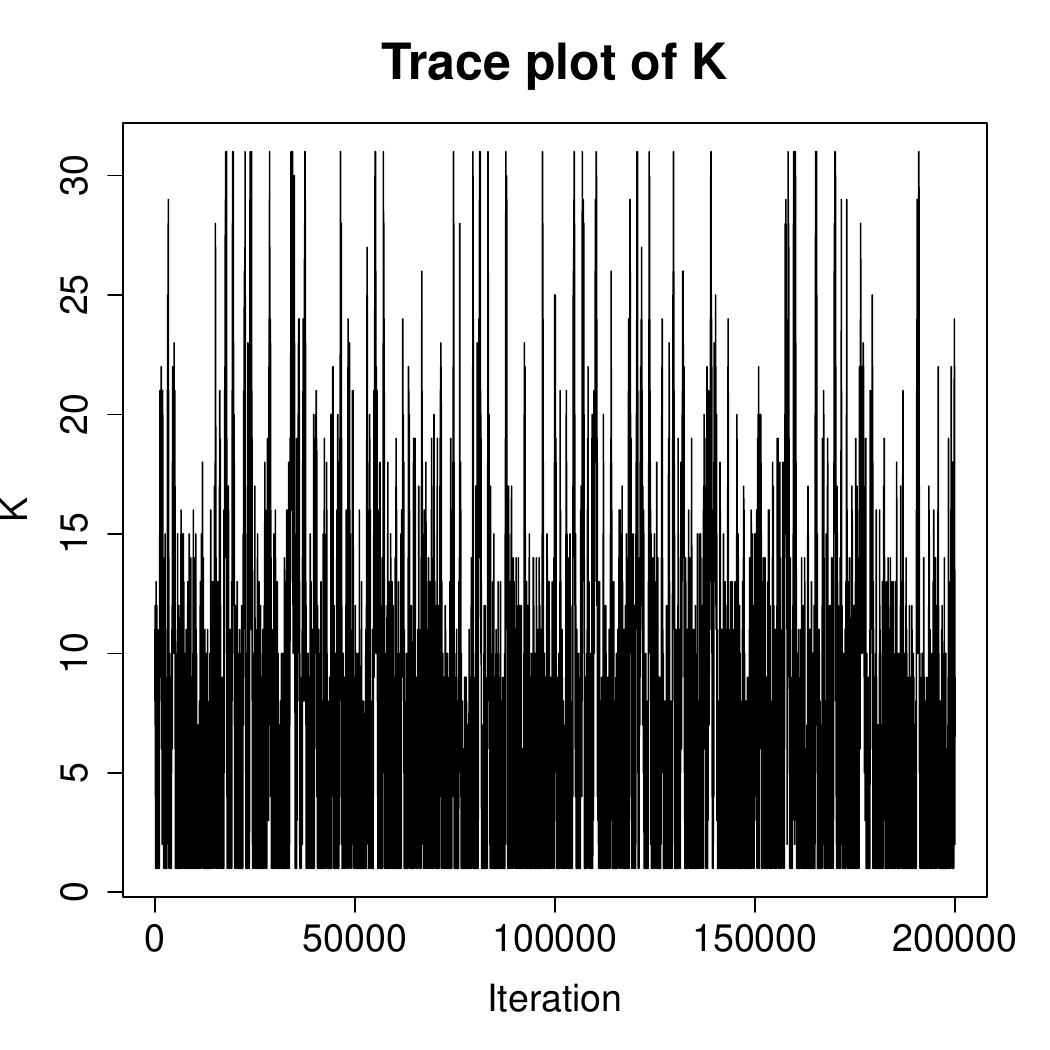}
\includegraphics[trim={0 0 0 0},clip, totalheight=0.25\textheight]{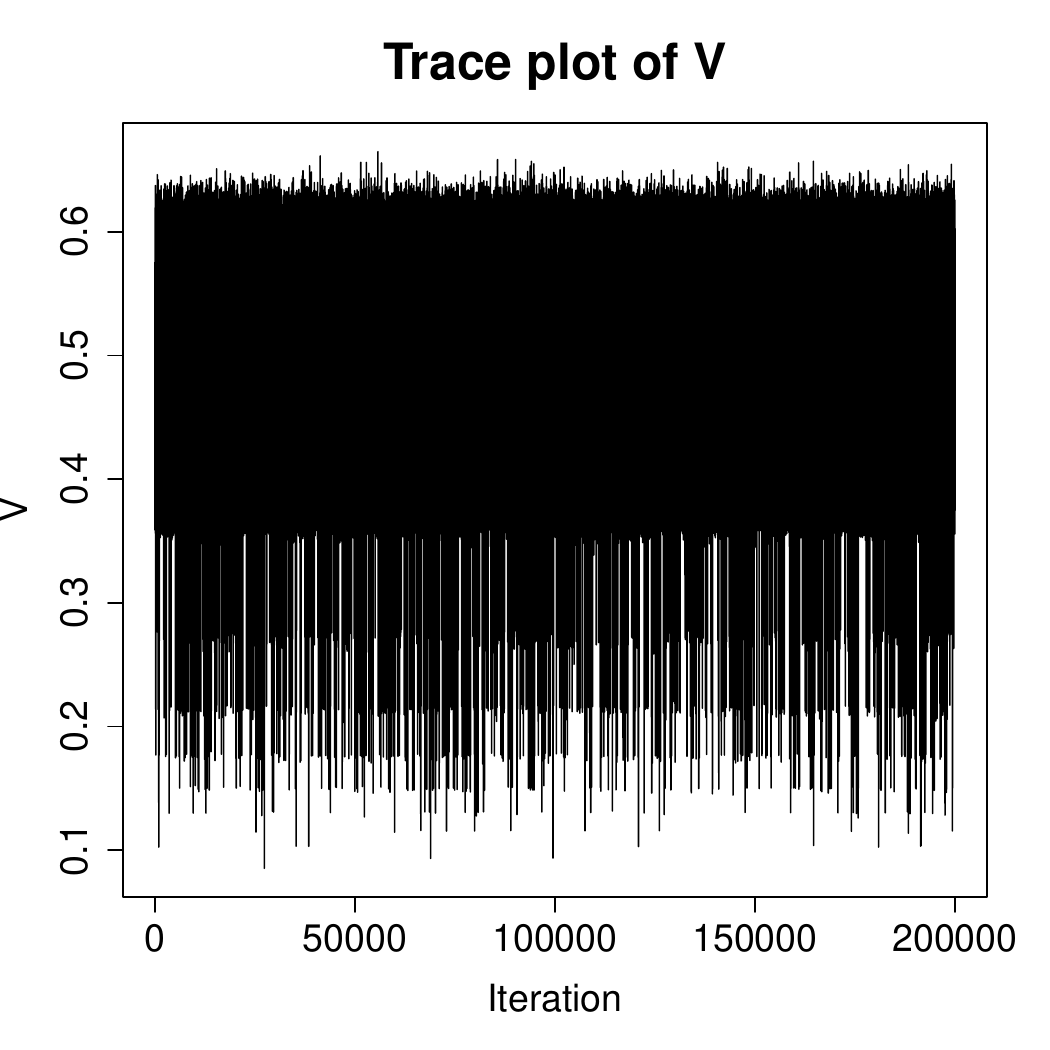}
\includegraphics[trim={0 0 0 0},clip, totalheight=0.25\textheight]{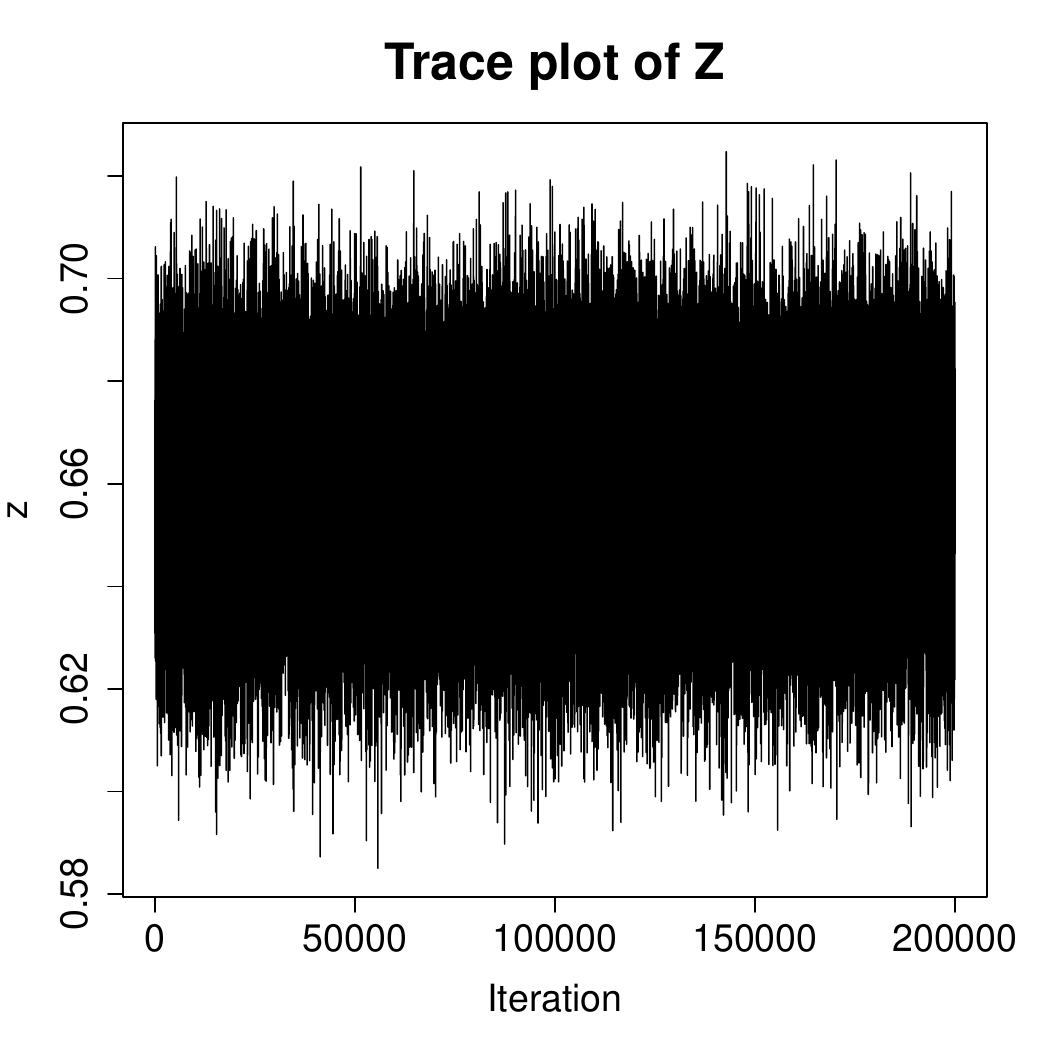}
\includegraphics[trim={0 0 0 0},clip, totalheight=0.25\textheight]{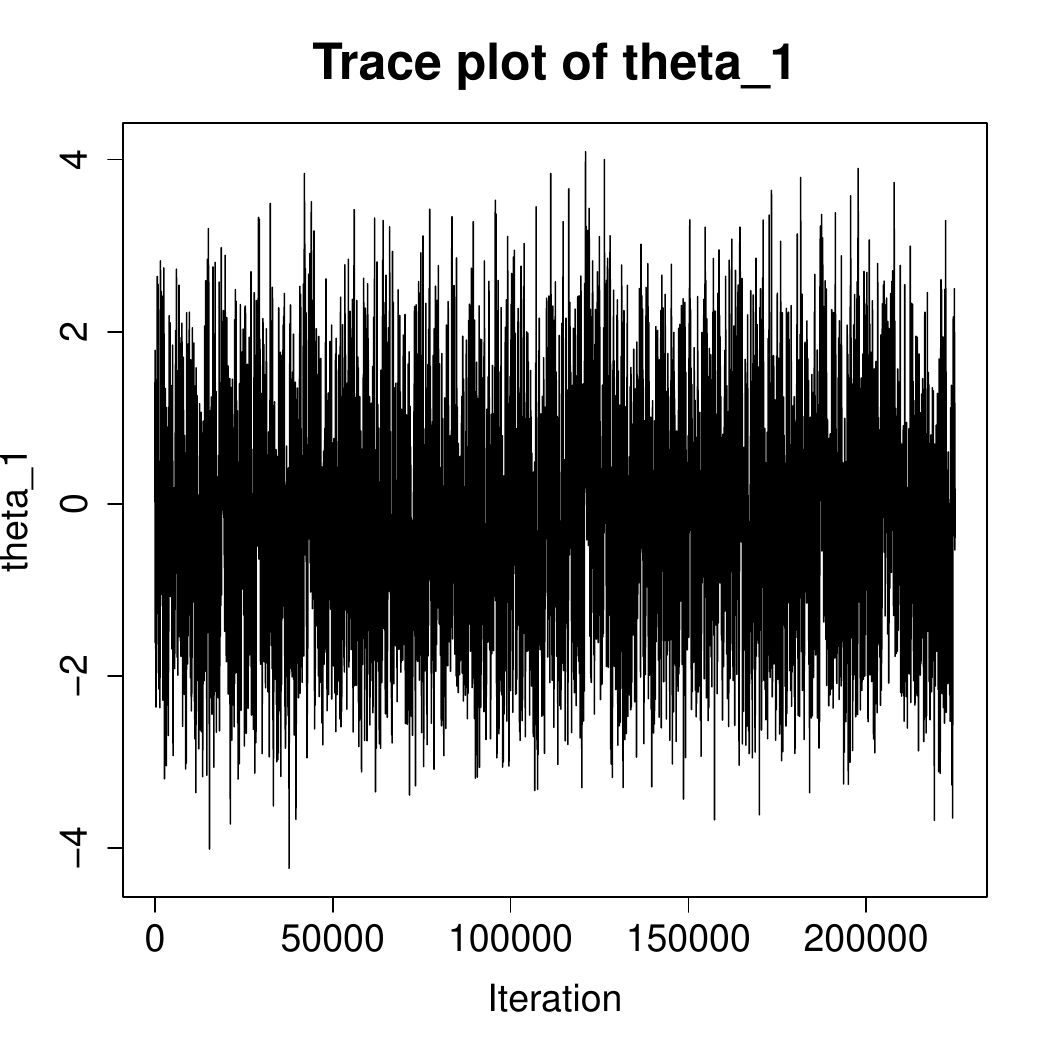}
\caption{{\bf Real spatial data analysis}: Traceplots of variable dimensional parameters.}
\label{fig:traceplot}
\end{figure}

\begin{figure}[htb]
\begin{center}
\centering
\includegraphics[height=1.5in,width=1.5in]{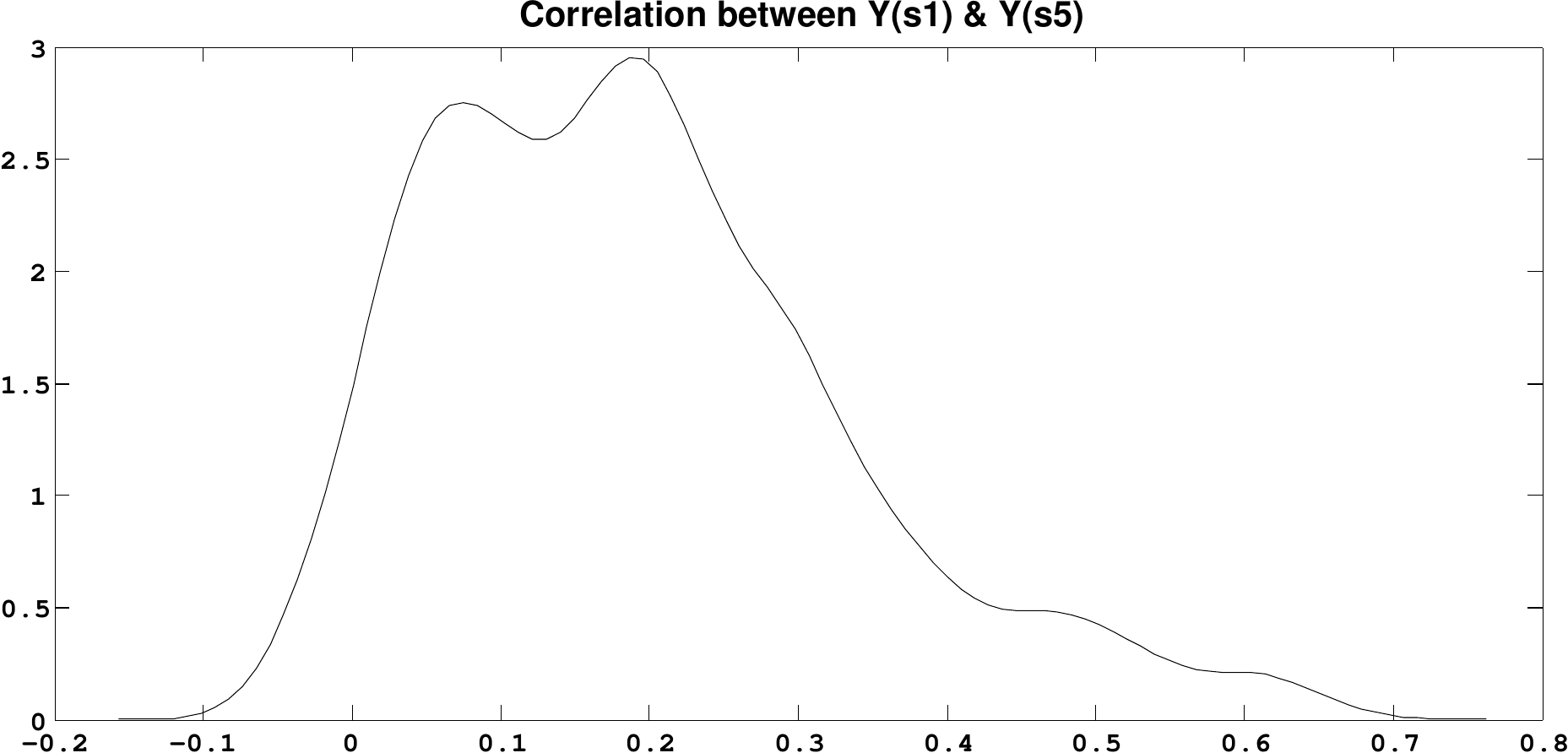}
\includegraphics[height=1.5in,width=1.5in]{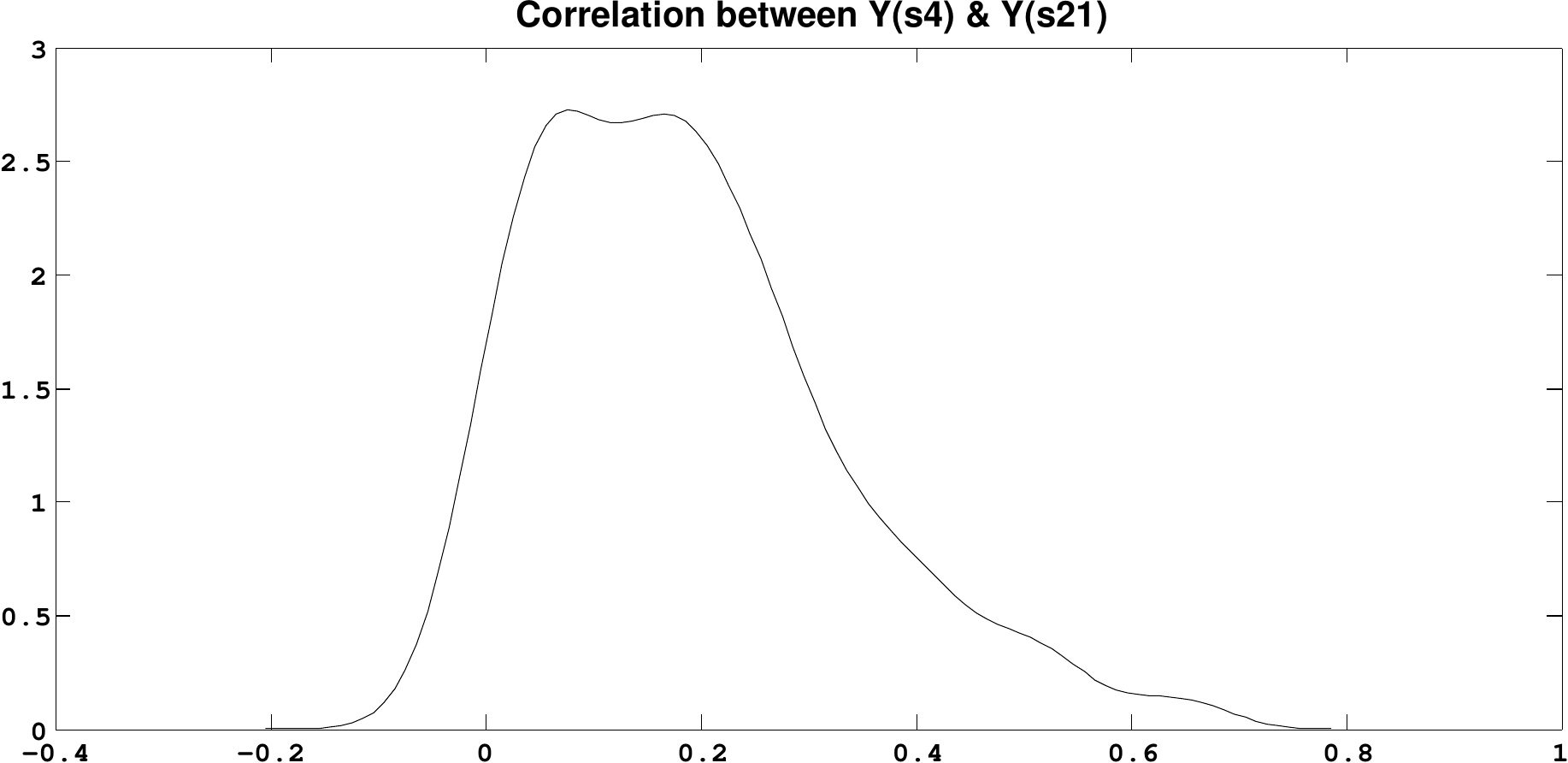}
\includegraphics[height=1.5in,width=1.5in]{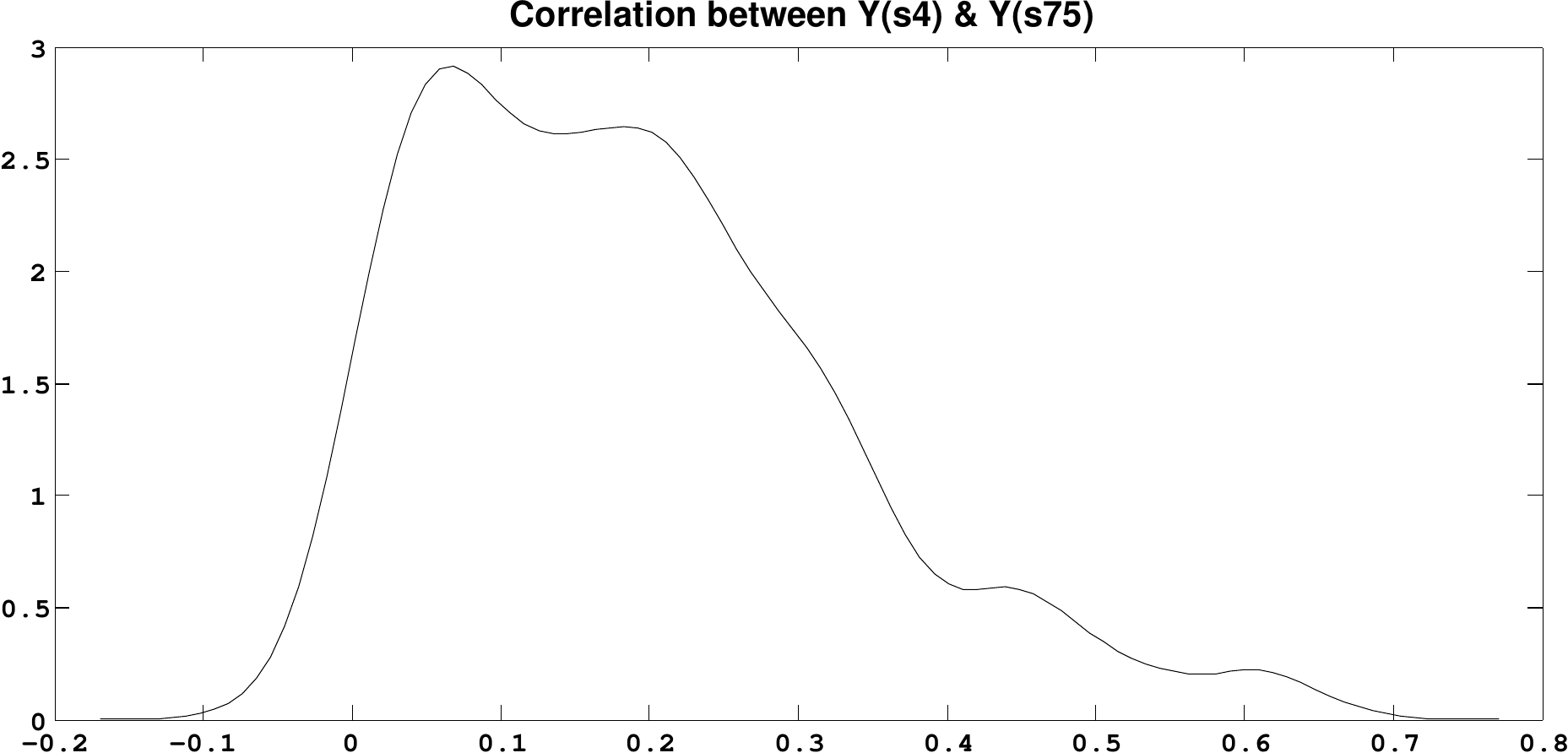}
\\[5mm]
\includegraphics[height=1.5in,width=1.5in]{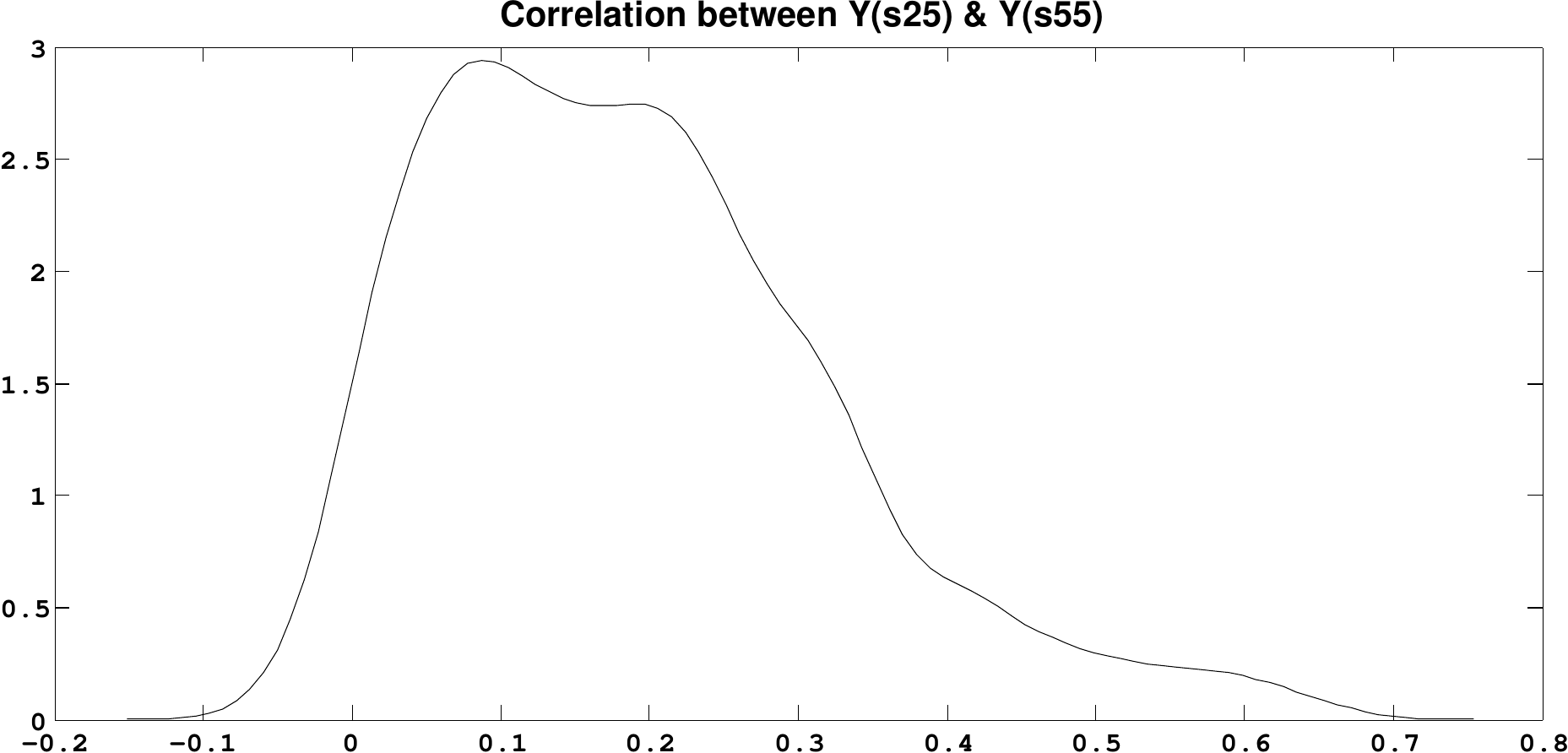}
\includegraphics[height=1.5in,width=1.5in]{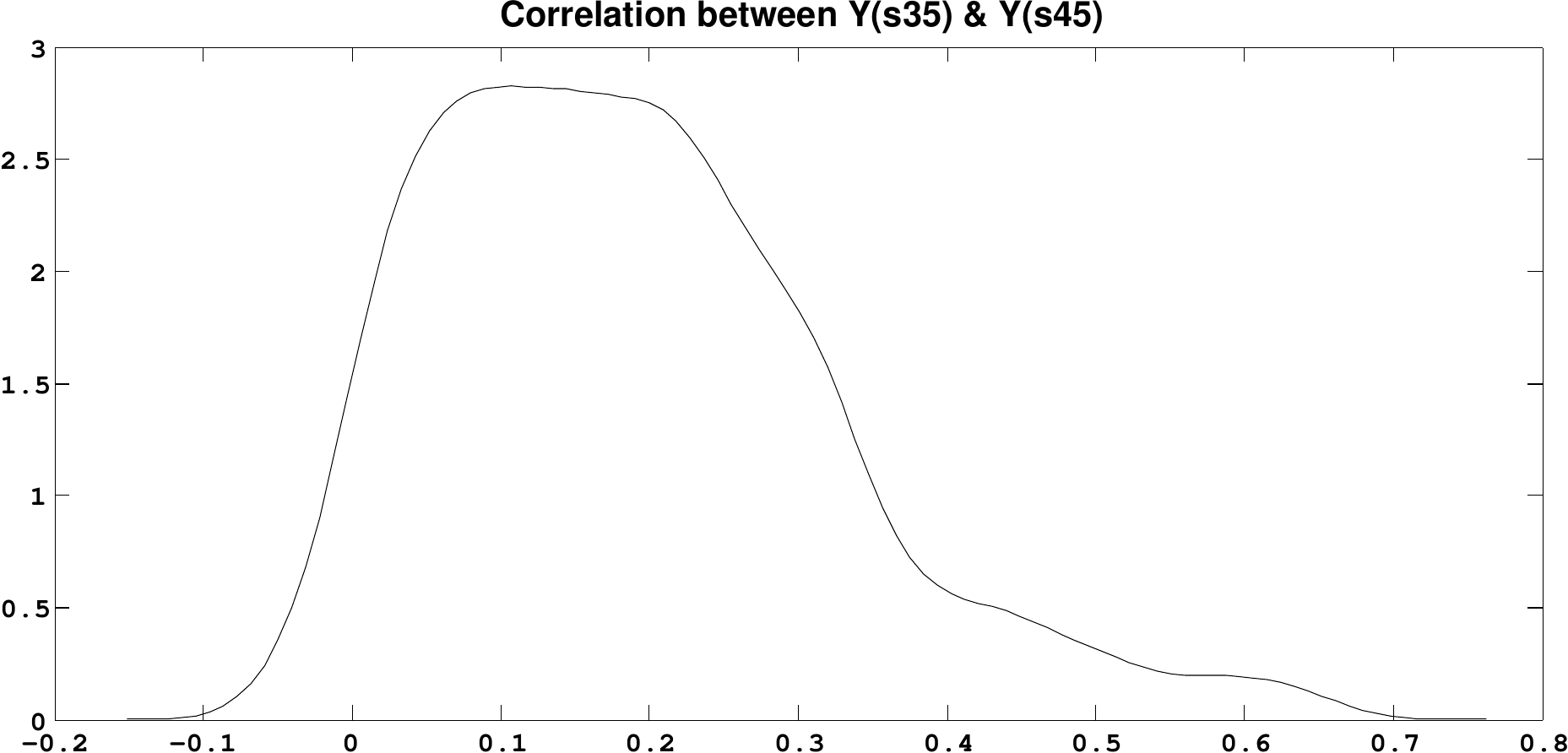}
\includegraphics[height=1.5in,width=1.5in]{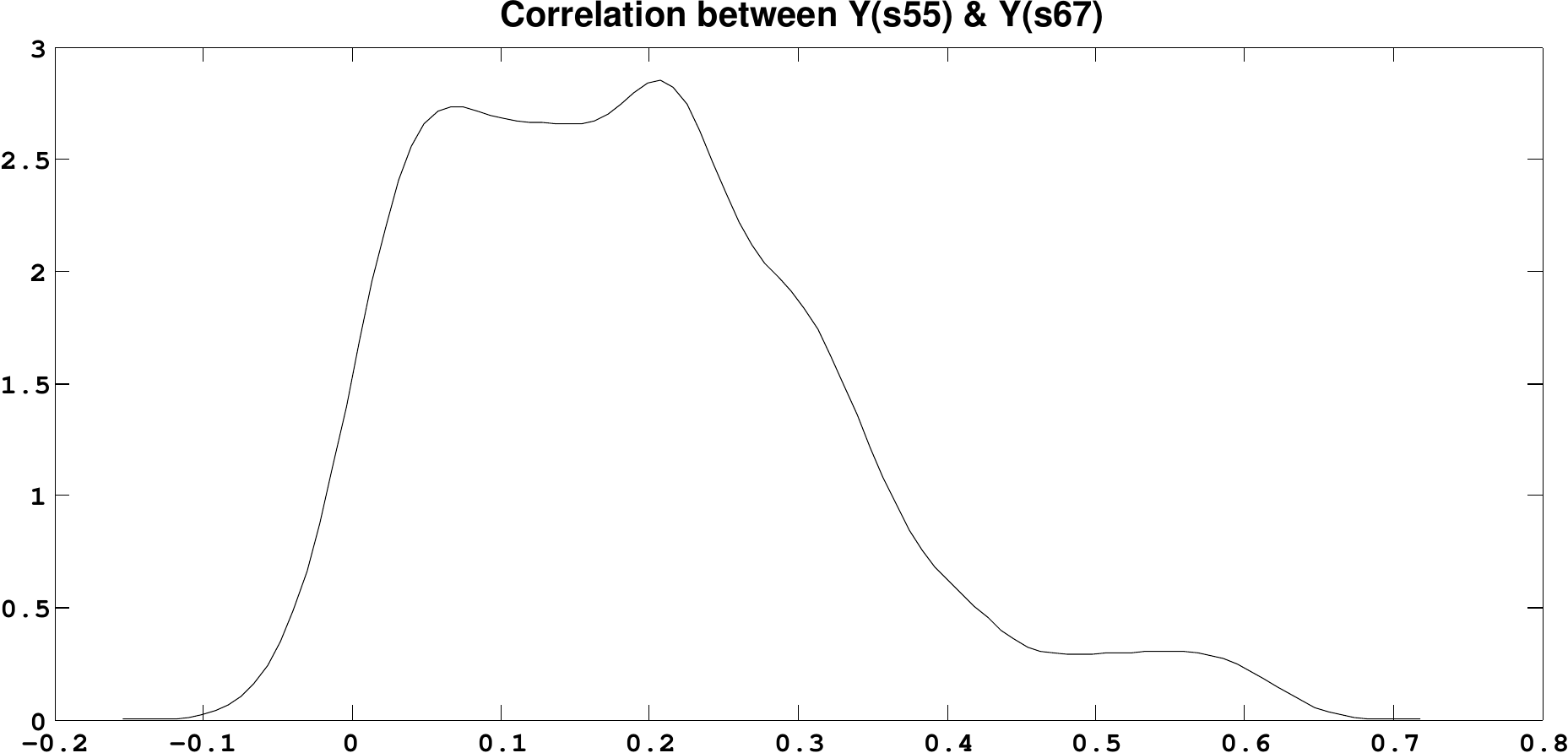}
\caption{{\bf Real spatial data analysis}: Posterior densities of correlations for 6 different pairs of locations.}
\label{fig:correlation_real}
\end{center}
\end{figure}

\subsection{Spatio-temporal Data}

\begin{figure}[H]
\centering
\includegraphics[trim={0 0 0 0},clip, totalheight=0.25\textheight]{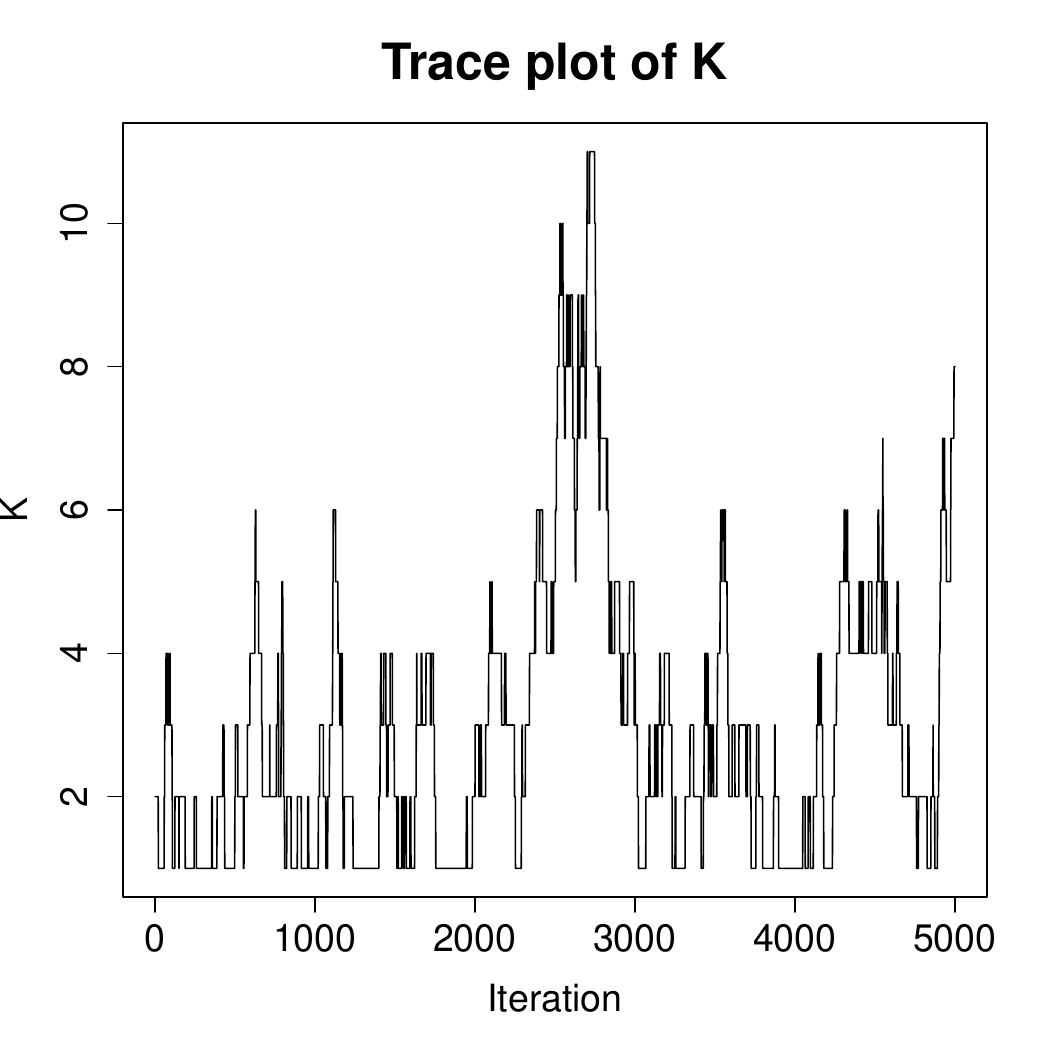}
\includegraphics[trim={0 0 0 0},clip, totalheight=0.25\textheight]{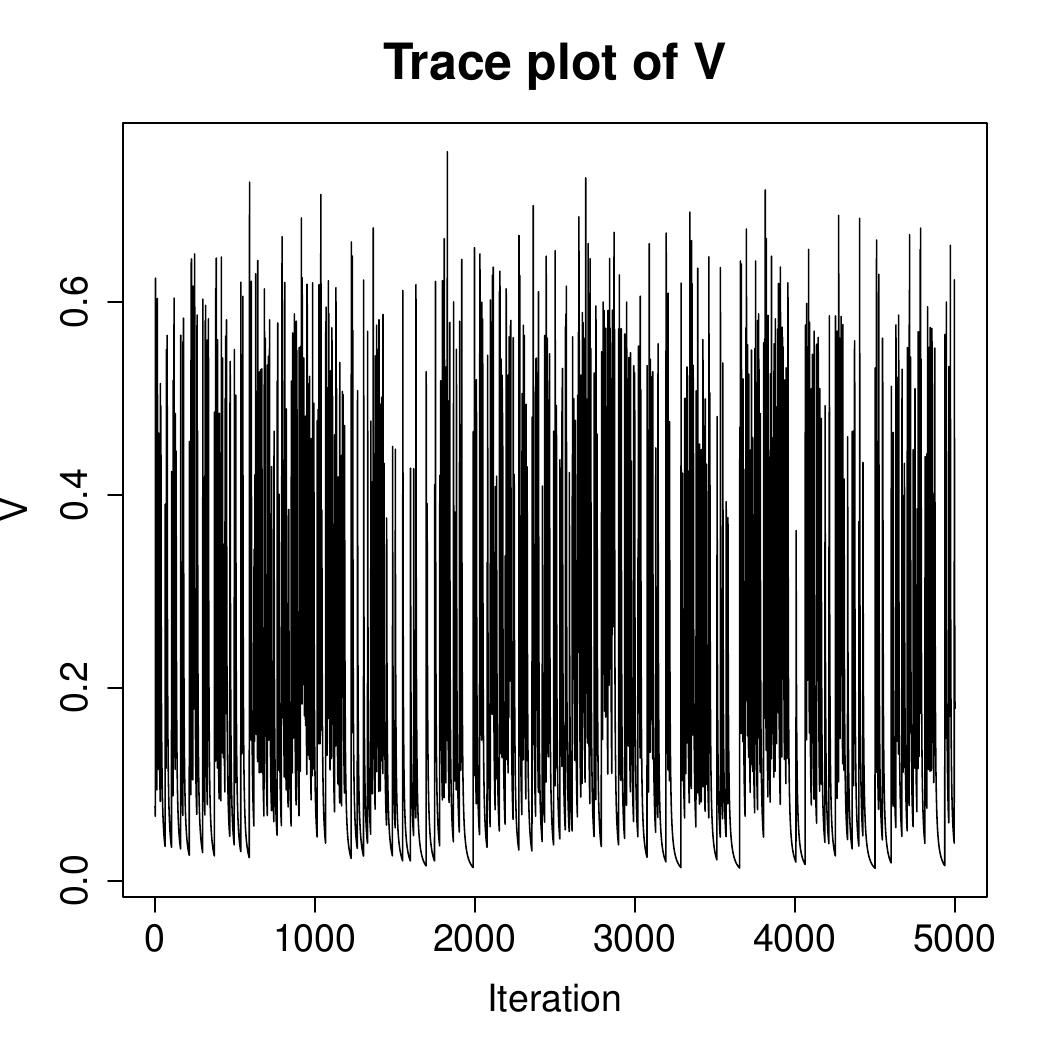}
\includegraphics[trim={0 0 0 0},clip, totalheight=0.25\textheight]{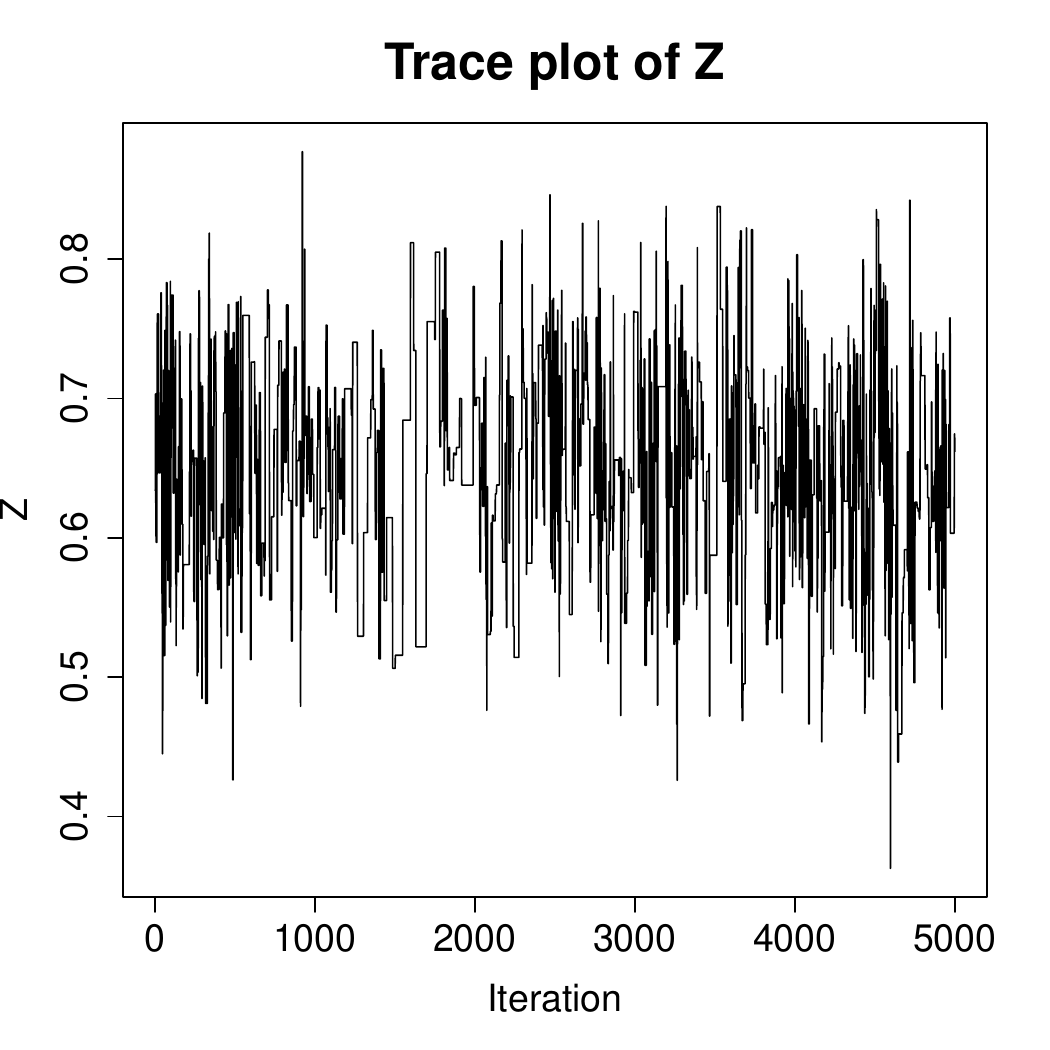}
\includegraphics[trim={0 0 0 0},clip, totalheight=0.25\textheight]{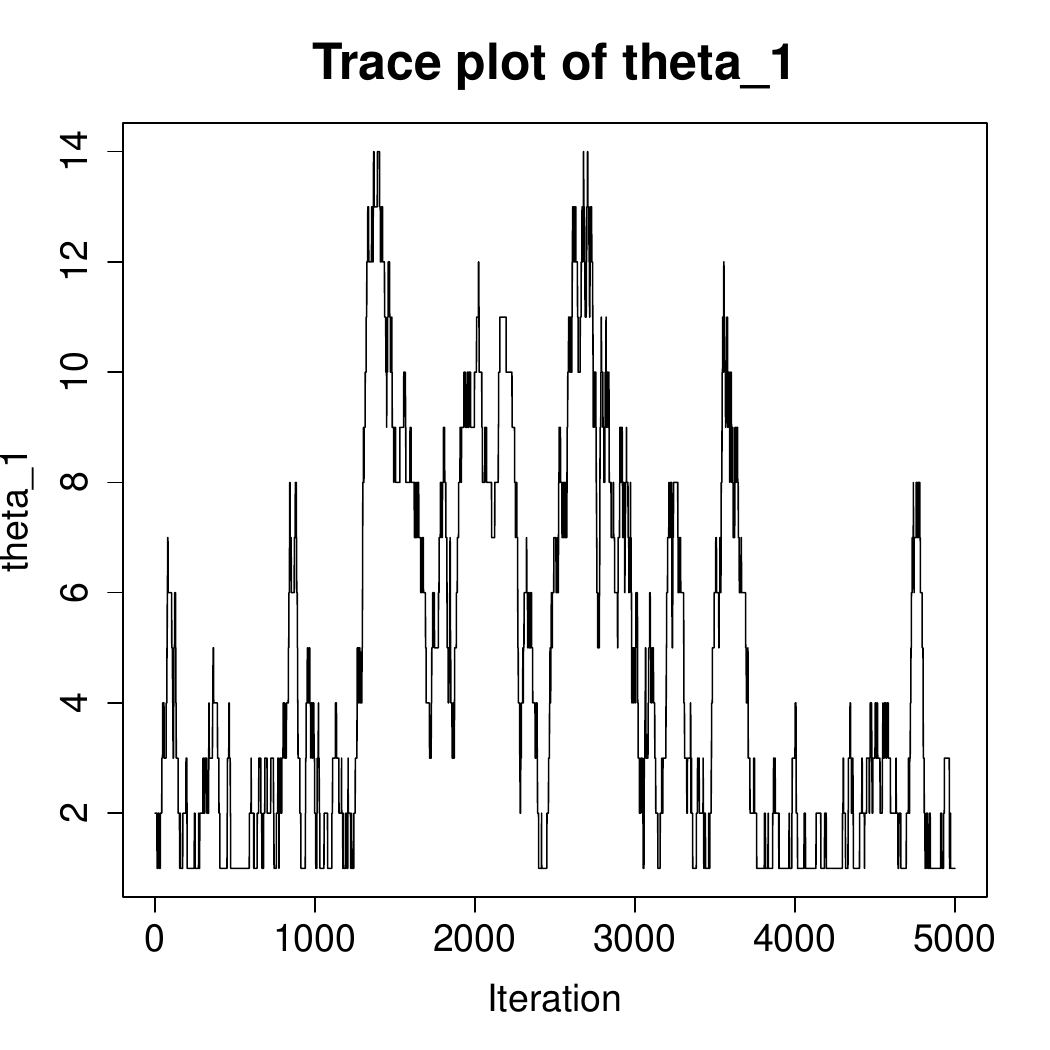}

\caption{{\bf Real spatio-temporal data analysis}: Traceplots of variable dimensional parameters.}
\label{fig:traceplotreal2}
\end{figure}

\begin{figure}
\centering
\subfigure[Correlation between($y_{1,21},y_{2,21}$)]{ 
\includegraphics[width=6cm,height=5cm]{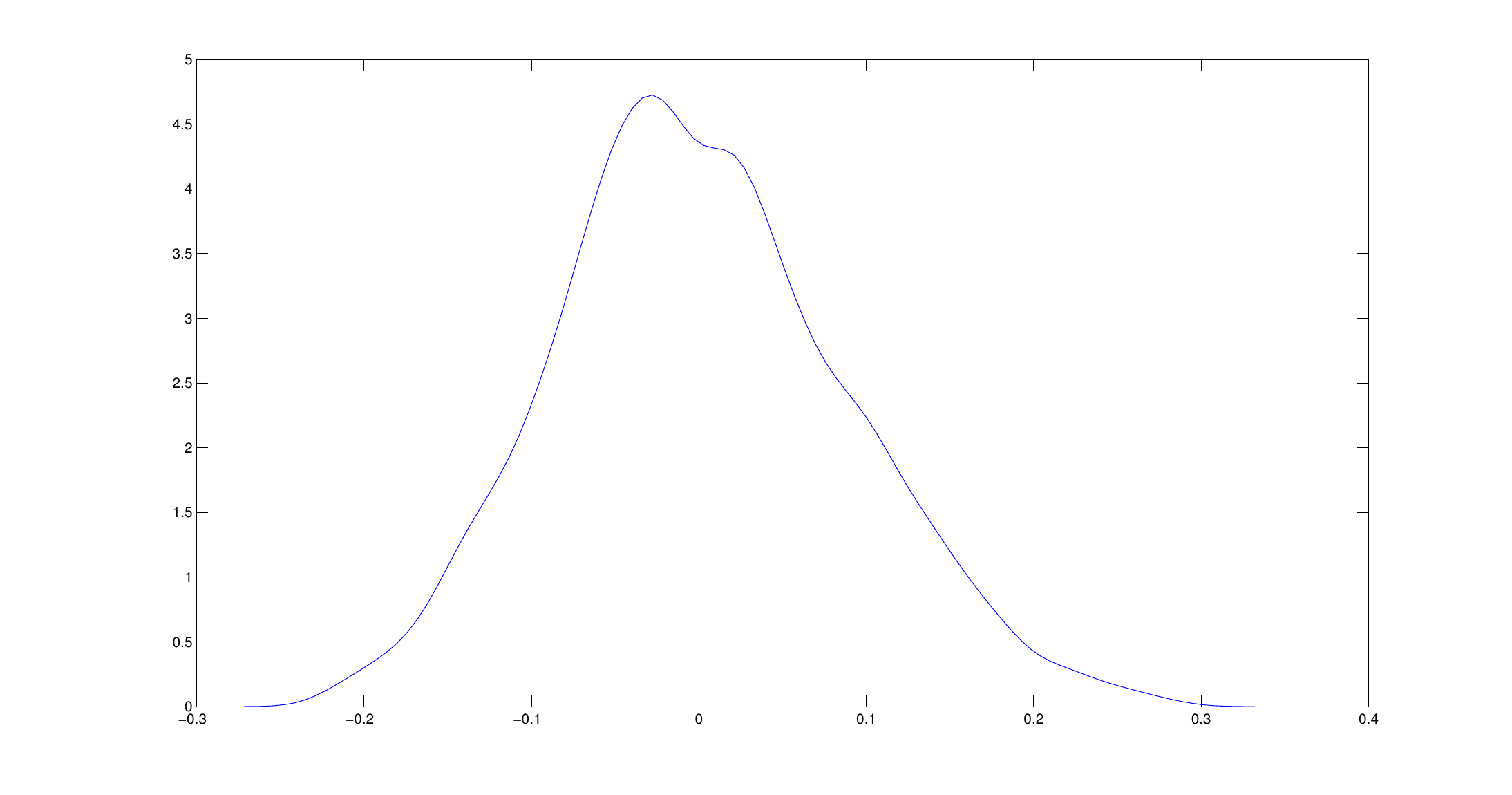}}
\hspace{2mm}
\subfigure[Correlation between($y_{1,21},y_{7,21}$)]{ 
\includegraphics[width=6cm,height=5cm]{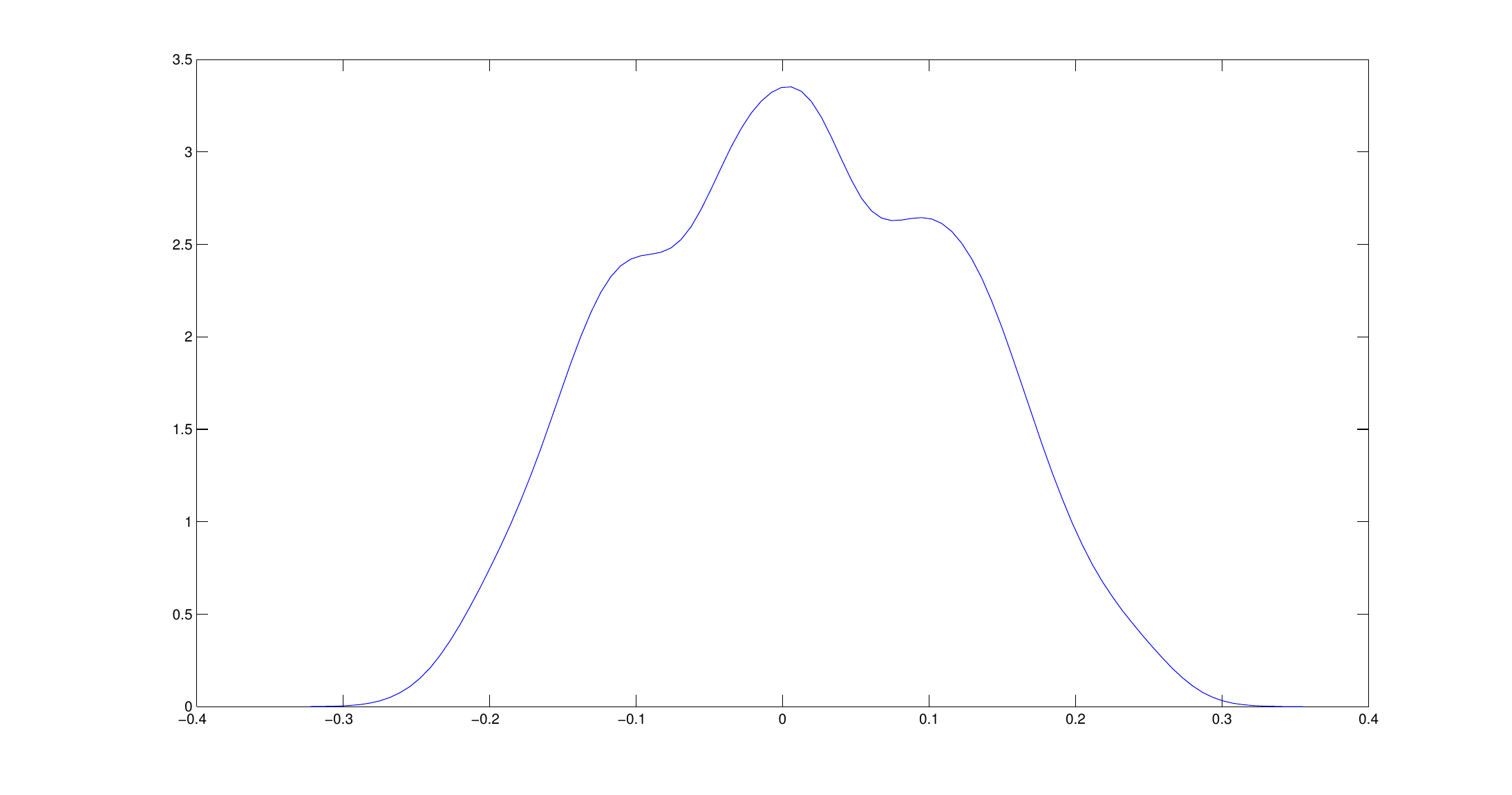}}
\hspace{2mm}
\subfigure[Correlation between($y_{2,21},y_{7,21}$)]{ 
\includegraphics[width=6cm,height=5cm]{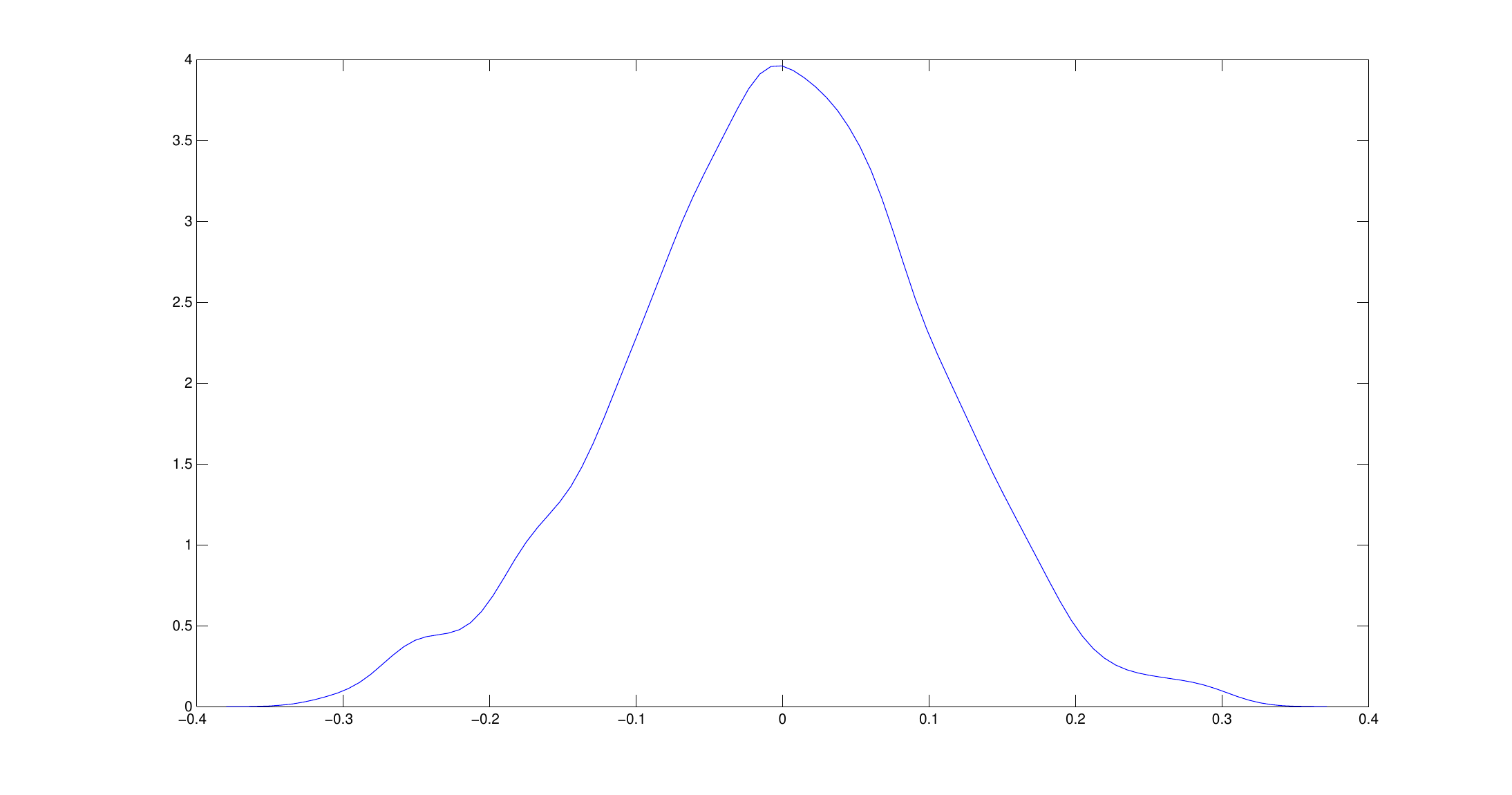}}
\vspace{2mm}
\subfigure[Correlation between($y_{24,15},y_{40,15}$)]{ 
\includegraphics[width=6cm,height=5cm]{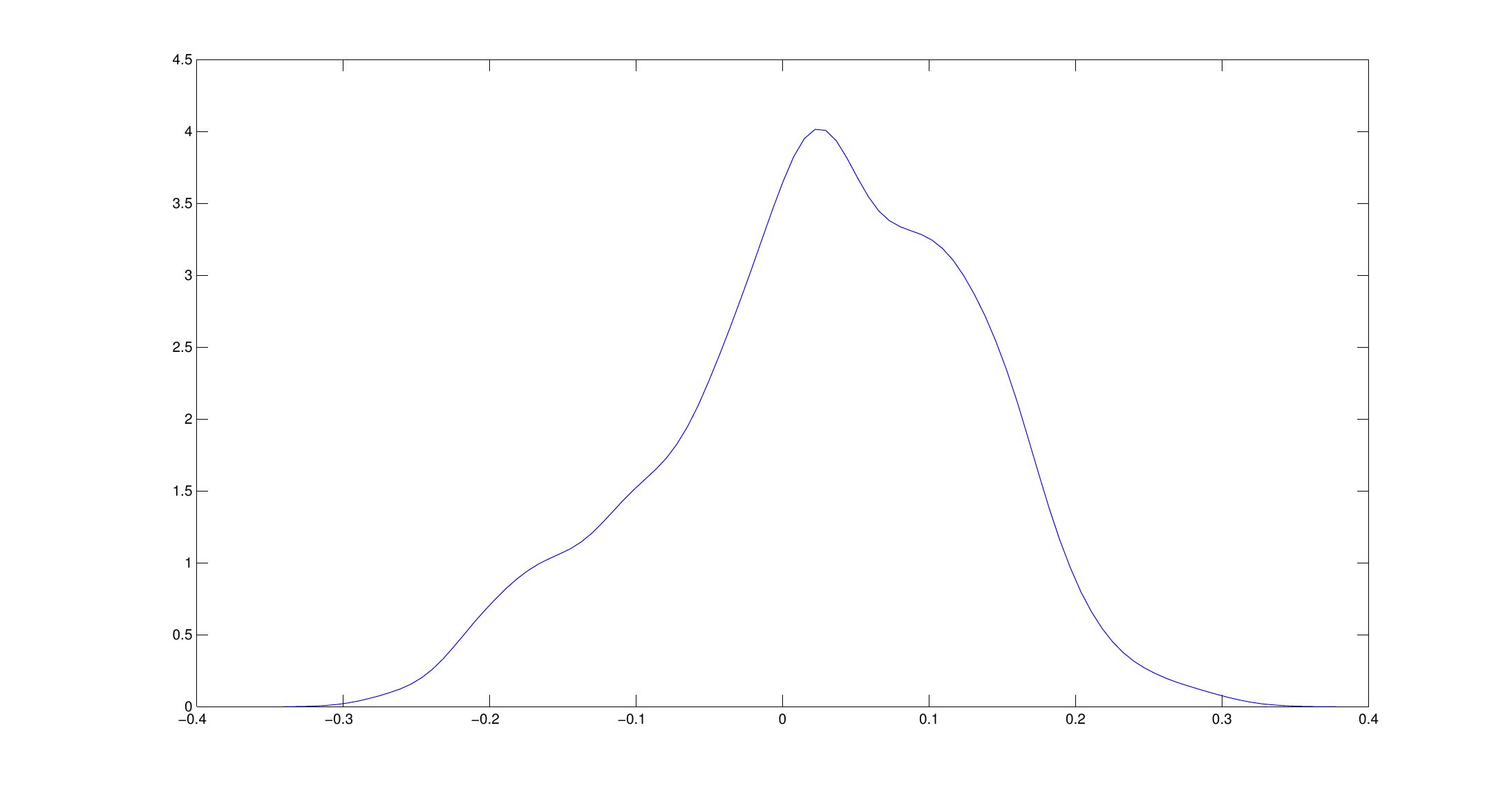}}
\hspace{2mm}
\subfigure[Correlation between($y_{24,15},y_{49,15}$)]{ 
\includegraphics[width=6cm,height=5cm]{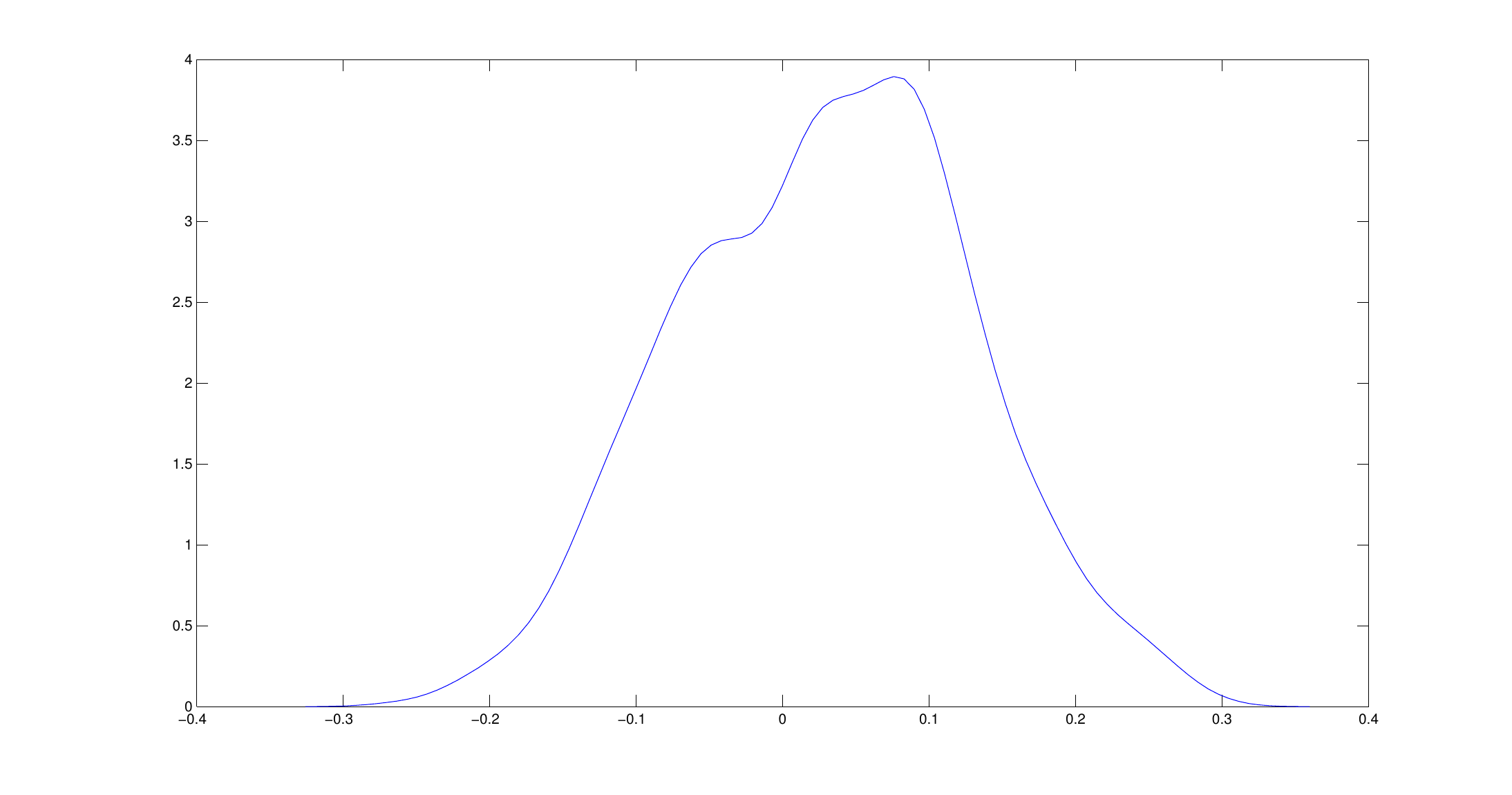}}
\hspace{2mm}
\subfigure[Correlation between($y_{40,15},y_{49,15}$)]{ 
\includegraphics[width=6cm,height=5cm]{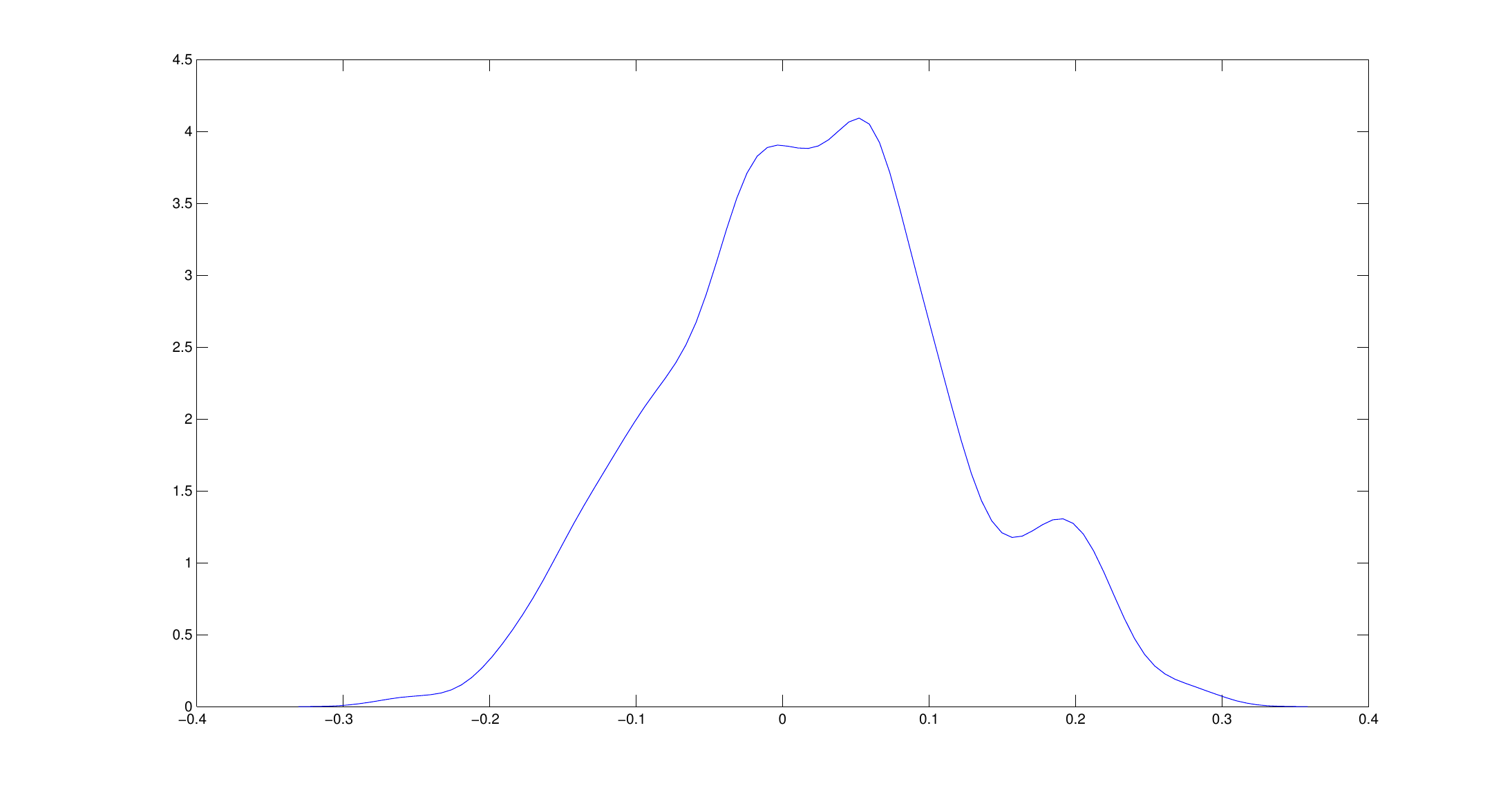}}

\caption{{\bf Real spatio-temporal data analysis:} Posterior densities of correlations for 6 different pairs of spatio-temporal points.}
\label{fig:correlation_real_big}
\end{figure}

\newpage

\renewcommand\baselinestretch{1.3}
\normalsize
\bibliography{references,references1}

\begin{thebibliography}{}

\bibitem[Banerjee and Gelfand(2003)Banerjee and Gelfand]{Banerjee03}
Banerjee, S. and Gelfand, A.~E. (2003).
\newblock {O}n {S}moothness {P}roperties of {S}patial {P}rocesses.
\newblock {\em Journal of Multivariate Analysis\/}, {\bf 84}, 85–100.

\bibitem[Chang {\em et~al.}(2011)Chang, Hsu, and Huang]{Chang10}
Chang, Y.-M., Hsu, N.-J., and Huang, H.-C. (2011).
\newblock {S}emiparametric {E}stimation and {S}election for {N}onstationary
  {S}patial {C}ovariance {F}unctions.
\newblock {\em Journal of Computational and Graphical Statistics\/}, {\bf 19},
  117--139.

\bibitem[Cressie and Wikle(2011)Cressie and Wikle]{Cressie11}
Cressie, N. A.~C. and Wikle, C.~K. (2011).
\newblock {\em {S}tatistics for {S}patio-{T}emporal {D}ata\/}.
\newblock Wiley, New York.

\bibitem[Damian {\em et~al.}(2001)Damian, Sampson, and Guttorp]{Damian01}
Damian, D., Sampson, P.~D., and Guttorp, P. (2001).
\newblock {B}ayesian {E}stimation of {S}emi-{P}arametric {N}on-stationary
  {S}patial {C}ovariance {S}tructures.
\newblock {\em Environmetrics\/}, {\bf 12}, 161--178.

\bibitem[Das and Bhattacharya(2019a)Das and Bhattacharya]{Das14supp}
Das, M. and Bhattacharya, S. (2019a).
\newblock {S}upplement to ``{N}onstationary, {N}onparametric, {N}onseparable
  {B}ayesian {S}patio-{T}emporal {M}odeling {U}sing {K}ernel {C}onvolution of
  {O}rder {B}ased {D}ependent {D}irichlet {P}rocess".
\newblock arXiv preprint.

\bibitem[Das and Bhattacharya(2019b)Das and Bhattacharya]{Das14}
Das, M. and Bhattacharya, S. (2019b).
\newblock {T}ransdimensional {T}ransformation {B}ased {M}arkov {C}hain {M}onte
  {C}arlo.
\newblock {\em Brazilian Journal of Probability and Statistics\/}, {\bf 33},
  87--138.

\bibitem[Dey and Bhattacharya(2018)Dey and Bhattacharya]{Dey14a}
Dey, K.~K. and Bhattacharya, S. (2018).
\newblock {A} {B}rief {T}utorial on {T}ransformation {B}ased {M}arkov {C}hain
  {M}onte {C}arlo and {O}ptimal {S}caling of the {A}dditive {T}ransformation.
\newblock {\em Brazilian Journal of Probability and Statistics\/}, {\bf 31},
  569--617.

\bibitem[Dey and Bhattacharya(2019)Dey and Bhattacharya]{Dey14b}
Dey, K.~K. and Bhattacharya, S. (2019).
\newblock {A} {B}rief {R}eview of {O}ptimal {S}caling of the {M}ain {MCMC}
  {A}pproaches and {O}ptimal {S}caling of {A}dditive {TMCMC} {U}nder
  {N}on-{R}egular {C}ases.
\newblock {\em Brazilian Journal of Probability and Statistics\/}, {\bf 33},
  222--266.

\bibitem[Duan {\em et~al.}(2007)Duan, Guindani, and Gelfand]{Duan07}
Duan, J.~A., Guindani, M., and Gelfand, A.~E. (2007).
\newblock {G}eneralized {S}patial {D}irichlet {P}rocess {M}odels.
\newblock {\em Biometrika\/}, {\bf 94}, 809--825.

\bibitem[Duan {\em et~al.}(2009)Duan, Gelfand, and Sirmans]{Duan09}
Duan, J.~A., Gelfand, A.~E., and Sirmans, C.~F. (2009).
\newblock {M}odeling {S}pace-{T}ime {D}ata {U}sing {S}tochastic {D}ifferential
  {E}quations.
\newblock {\em Bayesian Anslysis\/}, {\bf 4}, 733--758.

\bibitem[Dutta and Bhattacharya(2014)Dutta and Bhattacharya]{Dutta14}
Dutta, S. and Bhattacharya, S. (2014).
\newblock {M}arkov {C}hain {M}onte {C}arlo {B}ased on {D}eterministic
  {T}ransformations.
\newblock {\em Statistical Methodology\/}, {\bf 16}, 100--116.
\newblock Also available at http://arxiv.org/abs/1106.5850. Supplement
  available at http://arxiv.org/abs/1306.6684.

\bibitem[Ferguson(1973)Ferguson]{Ferguson73}
Ferguson, T.~S. (1973).
\newblock {A} {B}ayesian {A}nalysis of {S}ome {N}onparametric {P}roblems.
\newblock {\em The Annals of Statistics\/}, {\bf 1}, 209--230.

\bibitem[Ferguson(1974)Ferguson]{Ferguson74}
Ferguson, T.~S. (1974).
\newblock {P}rior {D}istributions on {S}paces of {P}robability {M}easures.
\newblock {\em The Annals of Statistics\/}, {\bf 2}, 615--629.

\bibitem[Fuentes(2002)Fuentes]{Fuentes02}
Fuentes, M. (2002).
\newblock {S}pectral {M}ethods for {N}onstationary {S}patial {P}rocesses.
\newblock {\em Biometrika\/}, {\bf 89}, 197--210.

\bibitem[Fuentes and Reich(2013)Fuentes and Reich]{Fuentes03}
Fuentes, M. and Reich, B. (2013).
\newblock {M}ultivariate {S}patial {N}onparametric {M}odelling via {K}ernel
  {P}rocess {M}ixing.
\newblock {\em Statistica Sinica\/}, {\bf 23}, 75--97.

\bibitem[Fuentes and Smith(2001)Fuentes and Smith]{Fuentes01}
Fuentes, M. and Smith, R.~L. (2001).
\newblock {A} {N}ew {C}lass of {N}onstationary {S}patial {M}odels.
\newblock Technical Report, Department of Statistics, North Carolina State
  University.

\bibitem[Geisser(1993)Geisser]{Geisser}
Geisser, S. (1993).
\newblock {\em {P}redictive {I}nference : {A}n for {I}ntroduction\/}.
\newblock Chapman \& Hall, London.

\bibitem[Gelfand {\em et~al.}(2005)Gelfand, Kottas, and MacEachern]{Gelfand05}
Gelfand, A.~E., Kottas, A., and MacEachern, S.~N. (2005).
\newblock {B}ayesian {N}onparametric {S}patial {M}odeling {W}ith {D}irichlet
  {P}rocess {M}ixing.
\newblock {\em Journal of the American Statistical Association\/}, {\bf 100},
  1021--1035.

\bibitem[Gilani {\em et~al.}(2016)Gilani, Berrocal, and Batterman]{Gilani16}
Gilani, O., Berrocal, V.~J., and Batterman, S.~A. (2016).
\newblock {N}on-stationary {S}patio-temporal {M}odeling of {T}raffic-related
  {P}ollutants in {N}ear-road {E}nvironments.
\newblock {\em Spatial and Spatio-temporal Epidemiology\/}, {\bf 18}, 24--37.

\bibitem[Griffin and Steel(2004)Griffin and Steel]{Griffin04}
Griffin, J.~E. and Steel, M. F.~J. (2004).
\newblock {S}emiparametric {B}ayesian {I}nference for {S}tochastic {F}rontier
  {M}odels.
\newblock {\em Journal of Econometrics\/}, {\bf 123}, 121--152.

\bibitem[Griffin and Steel(2006)Griffin and Steel]{Griffin06}
Griffin, J.~E. and Steel, M. F.~J. (2006).
\newblock {O}rder-{B}ased {D}ependent {D}irichlet {P}rocesses.
\newblock {\em Journal of the American Statistical Association\/}, {\bf 101},
  179--194.

\bibitem[Guttorp and Sampson(1994)Guttorp and Sampson]{Guttorp94}
Guttorp, P. and Sampson, P.~D. (1994).
\newblock {M}ethods for {E}stimating {H}eterogeneous {S}patial {C}ovariance
  {F}unctions with {E}nvironmental {A}pplications.
\newblock In G.~P. Patil and C.~R. Rao, editors, {\em Handbook of Statistics
  XII: Environmental Statistics\/}, pages 663--690, New York. Elsevier/North
  Holland.

\bibitem[Guttorp {\em et~al.}(2013)Guttorp, Schmidt, Bartlett, and
  Besag]{Guttorp2013}
Guttorp, P., Schmidt, A.~M., Bartlett, M., and Besag, J. (2013).
\newblock {C}ovariance {S}tructure of {S}patial and {S}patiotemporal
  {P}rocesses.
\newblock {\em WIREs Comput Stat\/}, pages 279--287.

\bibitem[Haas(1995)Haas]{Haas95}
Haas, T.~C. (1995).
\newblock {L}ocal {P}rediction of a {S}patio-{T}emporal {P}rocess with an
  {A}pplication to {W}et {S}ulfate {D}eposition.
\newblock {\em Journal of the American Statistical Association\/}, {\bf 90},
  1189–1199.

\bibitem[Higdon(1998)Higdon]{Higdon98}
Higdon, D. (1998).
\newblock {A} {P}rocess-{C}onvolution {A}pproach to {M}odeling {T}emperatures
  in the {N}orth {A}tlantic {O}cean.
\newblock {\em Environmental and Ecological Statistics\/}, {\bf 5}, 173–190.

\bibitem[Higdon(2001)Higdon]{Higdon02}
Higdon, D. (2001).
\newblock {S}pace and {S}pace-{T}ime {M}odeling {U}sing {P}rocess
  {C}onvolutions.
\newblock In C.~W. A.~V. Barnett, P.~C. Chatwin, and A.~H. El-Sharaawi,
  editors, {\em {Q}uantitative {M}ethods for {C}urrent {E}nvironmental
  {I}ssues\/}, pages 37--56, London. Springer-Verlag.

\bibitem[Higdon {\em et~al.}(1999)Higdon, Swall, and Kern]{Higdon99}
Higdon, D., Swall, J., and Kern, J. (1999).
\newblock {N}on-{S}tationary {S}atial {M}odeling.
\newblock In J.~M. Bernardo, J.~O. Berger, A.~P. Dawid, and A.~F.~M. Smith,
  editors, {\em {B}ayesian {S}tatistics 6\/}, pages 761--768, Oxford. Oxford
  University Press.

\bibitem[Ingebrigtsen {\em et~al.}(2014)Ingebrigtsen, Lindgren, and
  Steinsland]{Ing14}
Ingebrigtsen, R., Lindgren, F., and Steinsland, I. (2014).
\newblock {S}patial {M}odels with {E}xplanatory {V}ariables in the {D}ependence
  {S}tructure.
\newblock {\em Spatial Statistics\/}, {\bf 8}, 20--38.

\bibitem[Ishwaran and James(2001)Ishwaran and James]{Ishwaran01}
Ishwaran, H. and James, L.~F. (2001).
\newblock {G}ibbs {S}ampling {M}ethods for {S}tick-{B}reaking {P}rior.
\newblock {\em Journal of the American Statistical Association\/}, {\bf 96},
  161--173.

\bibitem[Kim {\em et~al.}(2005)Kim, Mallick, and Holmes]{Kim2005}
Kim, H.-m., Mallick, B.~K., and Holmes, C.~C. (2005).
\newblock {A}nalyzing {N}onstationary {S}patial {D}ata using {P}iecewise
  {G}aussian {P}rocesses.
\newblock {\em Journal of the American Statistical Association\/}, pages
  653--668.

\bibitem[Kottas {\em et~al.}(2007)Kottas, Duan, and Gelfand]{Kottas07}
Kottas, A., Duan, J.~A., and Gelfand, A.~E. (2007).
\newblock {M}odeling {D}isease {I}ncidence {D}ata with {S}patial and
  {S}patio-{T}emporal {D}irichlet {P}rocess {M}ixtures.
\newblock {\em Biometrical Journal\/}, {\bf 49}, 1--14.

\bibitem[Neto {\em et~al.}(2014)Neto, Schmidt, and Guttorp]{Neto14}
Neto, J. H.~V., Schmidt, A.~M., and Guttorp, P. (2014).
\newblock {A}ccounting for {S}patially {V}arying {D}irectional {E}ffects in
  {S}patial {C}ovariance {S}tructures.
\newblock {\em Journal of the Royal Statistical Society. Series C (Applied
  Statistics)\/}, {\bf 63}, 103--122.

\bibitem[Nott and Dunsmuir(2002)Nott and Dunsmuir]{Nott02}
Nott, D.~J. and Dunsmuir, W. T.~M. (2002).
\newblock {E}stimation of {N}onstationary {S}patial {C}ovariance {S}tructure.
\newblock {\em Biometrika\/}, {\bf 89}, 819--829.

\bibitem[Paciorek(2003)Paciorek]{Paciorek03}
Paciorek, C.~J. (2003).
\newblock {\em Nonstationaty Gaussian Process for Regression and Spatial
  Modeling\/}.
\newblock Doctoral thesis, Carnegie Mellon University.

\bibitem[Paciorek {\em et~al.}(2009)Paciorek, Yanosky, and Puett]{Paciorek09}
Paciorek, C.~J., Yanosky, J.~D., and Puett, R.~C. (2009).
\newblock {P}ractical {L}arge-scale {S}patio-{T}emporal {M}odeling of
  {P}articulate {M}atter {C}oncentrations.
\newblock {\em The Annals of Applied Statistics\/}, {\bf 3}, 370--397.

\bibitem[Petrone {\em et~al.}(2009)Petrone, Guindani, and Gelfand]{Petrone09}
Petrone, S., Guindani, M., and Gelfand, A.~E. (2009).
\newblock {H}ybrid {D}irichlet {M}ixture {M}odels for {F}unctional {D}ata.
\newblock {\em Journal of the Royal Statistical Society. Series B\/}, {\bf 71},
  755--782.

\bibitem[Pettit(1990)Pettit]{Pettit1}
Pettit, L. (1990).
\newblock {T}he {C}onditional of {P}redictive-{O}rdinate for the {N}ormal
  {D}istribution.
\newblock {\em Journal of the Royal Statistical Society: Series B\/}, {\bf 52},
  175--184.

\bibitem[Reich {\em et~al.}(2011)Reich, Fuentes, and Dunson]{Reich11}
Reich, B.~J., Fuentes, M., and Dunson, D.~B. (2011).
\newblock {B}ayesian {S}patial {Q}uantile {R}egression.
\newblock {\em Journal of the American Statistical Association\/}, {\bf 106},
  6--20.

\bibitem[Risser and Calder(2015)Risser and Calder]{Risser15}
Risser, M.~D. and Calder, C.~A. (2015).
\newblock {R}egression-based {C}ovariance {F}unctions for {N}onstationary
  {S}patial {M}odeling.
\newblock {\em Environmetrics\/}, {\bf 26}, 284--297.

\bibitem[Risser {\em et~al.}(2019)Risser, Calder, Berrocal, and
  Berrett]{Risser19}
Risser, M.~D., Calder, C.~A., Berrocal, V.~J., and Berrett, C. (2019).
\newblock {N}onstationary {S}patial {P}rediction of {S}oil {O}rganic {C}arbon:
  {I}mplications for {S}tock {A}ssessment {D}ecision {M}aking.
\newblock {\em The Annals of Applied Statistics\/}, {\bf 13}, 165--188.

\bibitem[Roy and Bhattacharya(2020)Roy and Bhattacharya]{Roy20}
Roy, S. and Bhattacharya, S. (2020).
\newblock {B}ayesian {C}haracterizations of {P}roperties of {S}tochastic
  {P}rocesses with {A}pplications.
\newblock ArXiv Preprint.

\bibitem[Sampson and Guttorp(1992)Sampson and Guttorp]{Sampson92}
Sampson, P.~D. and Guttorp, P. (1992).
\newblock {N}onparametric {E}stimation of {N}onstationary {S}patial
  {C}ovariance {S}tructure.
\newblock {\em Journal of the American Statistical Association\/}, {\bf 87},
  108–119.

\bibitem[Schmidt and O'Hagan(2003)Schmidt and O'Hagan]{Schmidt03}
Schmidt, A.~M. and O'Hagan, A. (2003).
\newblock {B}ayesian {I}nference for {N}onstationary {S}patial {C}ovariance
  {S}tructure via {S}patial {D}eformations.
\newblock {\em Journal of the Royal Statistical Society. Series B\/}, {\bf 65},
  743–758.

\bibitem[Schmidt {\em et~al.}(2011)Schmidt, Guttorp, and O'Hagan]{Schmidt11}
Schmidt, A.~M., Guttorp, P., and O'Hagan, A. (2011).
\newblock {C}onsidering {C}ovariates in the {C}ovariance {S}tructure of
  {S}patial {P}rocesses.
\newblock {\em Enironmetrics\/}, {\bf 22}, 487--500.

\bibitem[Sethuraman(1994)Sethuraman]{Sethuraman94}
Sethuraman, J. (1994).
\newblock {A} constructive definition of {D}irichlet priors.
\newblock {\em Statistica Sinica\/}, {\bf 4}, 639--650.

\bibitem[Stein(1999)Stein]{Stein99}
Stein, M.~L. (1999).
\newblock {\em {I}nterpolation of {S}patial {D}ata: {S}ome {T}heory for
  {K}riging\/}.
\newblock Springer-Verlag, New York.

\bibitem[Wolpert {\em et~al.}(2011)Wolpert, Clyde, and Tu]{Wolpert11}
Wolpert, R.~L., Clyde, M.~A., and Tu, C. (2011).
\newblock {S}tochastic {E}xpansions {U}sing {C}ontinuous {D}ictionaries:
  {L}\'{e}vy {A}daptive {R}egression {K}ernels.
\newblock {\em Annals of Statistics (to appear)\/}.

\bibitem[Yaglom(1987a)Yaglom]{Yaglom87a}
Yaglom, A.~M. (1987a).
\newblock {\em {C}orrelation {T}heory of {S}tationary and {R}elated {R}andom
  {F}unctions--{V}olume-{I}: {B}asic {R}esults\/}.
\newblock Springer-Verlag, New York.

\bibitem[Yaglom(1987b)Yaglom]{Yaglom87b}
Yaglom, A.~M. (1987b).
\newblock {\em {C}orrelation {T}heory of {S}tationary and {R}elated {R}andom
  {F}unctions--{V}olume-{II}: {S}upplemtary {N}otes and {R}eferences\/}.
\newblock Springer-Verlag, New York.

\bibitem[Yanosky {\em et~al.}(2008a)Yanosky, Paciorek, and Suh]{Yanosky08b}
Yanosky, J.~D., Paciorek, C.~J., and Suh, H.~H. (2008a).
\newblock {P}redicting {C}hronic {F}ine {P}articulate {E}xposures {U}sing
  {S}patio-temporal {M}odels for the {N}ortheastern and {M}idwestern {U.S.}
\newblock {\em Environmental Health Perspectives\/}, {\bf 117}, 522--529.

\bibitem[Yanosky {\em et~al.}(2008b)Yanosky, Paciorek, Schwartz, Laden, Puett,
  and Suh]{Yanosky08a}
Yanosky, J.~D., Paciorek, C.~J., Schwartz, J., Laden, F., Puett, R.~C., and
  Suh, H.~H. (2008b).
\newblock {S}patio-temporal {M}odeling of {C}hronic {P}m10 {E}xposures of
  {N}urses {H}ealth {S}tudy.
\newblock {\em Atmospheric Environment\/}, {\bf 47}, 4047--4062.

\end{thebibliography}


\end{document}